\documentclass[a4paper,11pt]{article}

\usepackage{thm-restate}
\usepackage[utf8]{inputenc}
\usepackage[T1]{fontenc}
\usepackage[english]{babel}
\usepackage{lmodern}
\usepackage[pdftex]{graphicx}
\usepackage[usenames,dvipsnames]{xcolor}
\usepackage{fullpage}
\usepackage{amsmath,amssymb,amsfonts,amsthm}
\usepackage{booktabs}
\usepackage{enumitem}
\usepackage[ruled,linesnumbered,vlined]{algorithm2e}
\usepackage{tikz}
\usetikzlibrary{calc,positioning}
\newcommand{\eps}{\varepsilon}
\newcommand{\Oeps}{O_\eps}
\usepackage{xargs} 
\usepackage{nicefrac}
\usepackage[prependcaption,textsize=scriptsize]{todonotes}

\definecolor{forestgreen}{rgb}{0.13, 0.55, 0.13}

\newcommandx{\atodo}[2][1=]{\todo[linecolor=red,backgroundcolor=forestgreen!25,bordercolor=red,#1]{AF:#2}}
\newcommandx{\rev}[2][1=]{\todo[linecolor=red,backgroundcolor=red!25,bordercolor=red,#1]{Rev: #2}}

\definecolor{Darkblue}{rgb}{0,0,0.4}
\usepackage[colorlinks,linkcolor=Darkblue,filecolor=blue,citecolor=blue,urlcolor=Darkblue,pagebackref]{hyperref}
\usepackage[nameinlink]{cleveref}

\usepackage{microtype}

\usepackage{authblk}

\newtheorem{lemma}{Lemma}
\newtheorem{theorem}{Theorem}
\newtheorem{claim}{Claim}

\newtheorem{corollary}{Corollary}

\theoremstyle{definition}

\newcommand{\vpi}{\varphi}
\newcommand{\diam}{{\rm diam}}
\newcommand{\poly}{{\rm poly}}

\title{Constructing Light Spanners Deterministically in Near-Linear Time}

\author[1]{Stephen Alstrup\thanks{Research partly supported by Innovationsfonden DK, DABAI (5153-00004A) and by VILLUM Foundation grant 16582, Basic Algorithms Research Copenhagen (BARC).
}}
\author[1]{Søren Dahlgaard\thanks{Research partly supported by Advanced Grant DFF-0602-02499B from the Danish Council for Independent Research.}}
\author[2]{Arnold Filtser\thanks{Supported in part by the Simons Foundation, ISF grant No. (1817/17), and by BSF grant No. 2015813. The research was done while the author was affiliated with Ben-Gurion University of the Negev, and Columbia University.}}
\author[1]{Morten Stöckel\thanks{Research partly supported by Villum Fonden.}}
\author[1]{Christian Wulff-Nilsen\thanks{Research partly supported by the Starting Grant 7027-00050B from the Independent Research Fund Denmark under the Sapere Aude research career programme.}}
\affil[1]{University of Copenhagen\\\texttt{\{s.alstrup,soerend,most,koolooz\}@di.ku.dk}}
\affil[2]{Bar Ilan University\\\texttt{arnold273@gmail.com}}

\date{}

\begin{document}
	
	\maketitle
\begin{abstract}
Graph spanners are well-studied and widely used both in theory and practice. In a recent breakthrough, Chechik and Wulff-Nilsen \cite{ChechikW16} improved the state-of-the-art for light spanners by constructing a $(2k-1)(1+\eps)$-spanner with $O(n^{1+\nicefrac{1}{k}})$ edges and $\Oeps(n^{\nicefrac{1}{k}})$ lightness. Soon after, Filtser and Solomon \cite{FiltserS16} showed that the classic greedy spanner construction achieves the same bounds. The major drawback of the greedy spanner is its running time of $O(mn^{1+\nicefrac{1}{k}})$ (which is faster than \cite{ChechikW16}). This makes the construction impractical even for graphs of moderate size. 
Much faster spanner constructions do exist but they only achieve lightness $\Omega_\eps(kn^{\nicefrac{1}{k}})$, even when randomization is used. 

The contribution of this paper is deterministic spanner constructions that are fast, and achieve similar bounds as the state-of-the-art slower constructions.   	
Our first result is an $\Oeps(n^{2+\nicefrac{1}{k}+\eps'})$ time spanner construction which achieves the state-of-the-art bounds.
Our second result is an $\Oeps(m + n\log n)$ time construction of a spanner with $(2k-1)(1+\eps)$ stretch, $O(\log k\cdot n^{1+\nicefrac{1}{k}})$ edges and $\Oeps(\log k\cdot n^{\nicefrac{1}{k}})$ lightness. This is an exponential improvement in the dependence on $k$ compared to the previous result with such running time.
Finally, for the important special case where $k=\log n$, for every constant $\eps>0$, we provide an $O(m+n^{1+\eps})$ time construction that produces an $O(\log n)$-spanner with $O(n)$ edges and $O(1)$ lightness which is asymptotically optimal. This is the first known sub-quadratic construction of such a spanner for any $k = \omega(1)$.

To achieve our constructions, we show a novel deterministic incremental
approximate distance oracle. Our new oracle is crucial in our construction,
as known randomized dynamic oracles require the assumption of a
non-adaptive adversary. This is a strong assumption, which has seen recent
attention in prolific venues. Our new oracle allows the order of the edge
insertions to not be fixed in advance, which is critical as our spanner
algorithm chooses which edges to insert based on the answers to distance
queries.
We believe our new oracle is of independent interest.
\end{abstract}

\tableofcontents
\newpage
\section{Introduction}
A fundamental problem in graph data structures is compressing graphs such that
certain metrics are preserved as well as possible. A popular way to achieve
this is through \emph{graph spanners}. Graph spanners are sparse subgraphs that
approximately preserve pairwise shortest path distances for all vertex pairs.
Formally, we say that a subgraph $H =(V,E',w)$ of an edge-weighted undirected
graph $G = (V,E,w)$ is a \emph{$t$-spanner} of $G$ if for all $u,v \in V$ we
have $d_H(u,v) \leq t \cdot d_G(u,v)$, where $d_X$ is the
shortest path distance function for graph $X$ and $w$ is the edge weight
function. Under such a guarantee, we say that our graph spanner $H$ has
\emph{stretch} $t$. In the following, we assume that the underlying graph $G$
is connected; if it is not, we can consider each connected component separately
when computing a spanner.

Graph spanners originate from the 80's \cite{PelegS89,PelegU87} and have seen
applications in e.g. synchronizers~\cite{PelegU87}, compact routing
schemes~\cite{ThorupZ01,PelegU88,Chechik13a},
broadcasting~\cite{FarleyPZW04}, and distance oracles~\cite{CWN12}.

The two main measures of the sparseness of a spanner $H$ is the size (number of
edges) and the \emph{lightness}, which is defined as the ratio
$w(H)/w(MST(G))$, where $w(H)$ resp.~$w(MST(G))$ is the total weight of edges
in $H$ resp.~a minimum spanning tree (MST) of $G$. It has been
established that for any positive integer $k$, a $(2k-1)$-spanner of $O(n^{1+1/k})$ edges exists for any $n$-vertex graph~\cite{Awerbuch:1985:CNS:4221.4227}.
This stretch-size tradeoff is widely believed to be optimal due to a matching lower bound implied by Erd\H{o}s'
girth conjecture~\cite{Erdos64}, and there are several papers concerned with
constructing spanners efficiently that get as close as possible to this lower bound~\cite{ThorupZ05,BaswanaS07,RodittyZ11}.

Obtaining spanners with small lightness (and thus total weight) is
motivated by applications where edge weights denote e.g.~establishing cost.
The best possible total weight that can be achieved in order to ensure finite stretch is the weight of an MST, thus
making the definition of lightness very natural.
The size lower bound of the unweighted case provides a lower
bound of $\Omega(n^{1/k})$ lightness under the girth conjecture, since $H$ must
have size and weight $\Omega(n^{1+1/k})$ while the MST has size and weight
$n-1$. Obtaining this lightness has been the subject of an active line of
work~\cite{Althofer1993,ChandraDNS92,ElkinNS14,ChechikW16,FiltserS16}.
Throughout this paper we say that a spanner is \emph{optimal} when its bounds
coincide asymptotically with those of the girth conjecture.
Obtaining an efficient spanner construction with optimal stretch-lightness
trade-off remains one of the main open questions in the field of graph
spanners.

\subparagraph*{Light spanners.}
Historically, the main approach of obtaining a spanner of bounded lightness has
been through different analyses of the classic \emph{greedy spanner}.
Given $t\geq 1$, the greedy $t$-spanner is constructed as
follows: iterate through the edges in non-decreasing order of weight and add an
edge $e$ to the partially constructed spanner $H$ if the shortest path distance
in $H$ between the endpoints of $e$ is greater than $t$ times the weight of
$e$. The study of this spanner algorithm dates back to the early 90's with its
first analysis by Althöfer et al.~\cite{Althofer1993}.
They showed that this simple procedure with stretch $2k-1$ obtains the optimal
$O(n^{1+1/k})$ size, and has lightness $O(n/k)$. The algorithm was subsequently
analyzed in~\cite{ChandraDNS92,ElkinNS14,FiltserS16} with stretch
$(1+\eps)(2k-1)$ for any $0 < \eps < 1$.
Recently, a break-through result of Chechik and Wulff-Nilsen~\cite{ChechikW16}
showed that a significantly more complicated spanner construction obtains
nearly optimal stretch, size and lightness giving the following theorem.
\begin{theorem}[\cite{ChechikW16}]\label{lem:light_spanner}
    Let $G=(V,E,w)$ be an edge-weighted undirected $n$-vertex graph and let $k$
    be a positive integer. Then for any $0<\eps<1$ there exists a
    $(1+\eps)(2k-1)$-spanner of size $O(n^{1+1/k})$ and lightness
    $\Oeps(n^{1/k})$.\footnote{$\Oeps$ notation hides polynomial factors in
    $1/\eps$.}
\end{theorem}
Following the result of \cite{ChechikW16} it was shown by Filtser and
Solomon~\cite{FiltserS16} that this bound is matched by the greedy spanner. 
In fact, they show that the greedy spanner is \emph{existentially
optimal}, meaning that if there is a $t$-spanner construction achieving an
upper bound $m(n,t)$ resp.~$l(n,t)$ on the size resp.~lightness of any
$n$-vertex graph then this bound also holds for the greedy $t$-spanner. In
particular, the bounds in \Cref{lem:light_spanner} also hold for the greedy
spanner.

\subparagraph*{Efficient spanners.}
A major drawback of the greedy spanner is its $O(m\cdot(n^{1+1/k}+n\log n))$
construction time \cite{Althofer1993}. Similarly, Chechik and
Wulff-Nilsen~\cite{ChechikW16} only state their construction time to be
polynomial, but since they use the greedy spanner as a subroutine, it has the
same drawback. Adressing this problem, Elkin and Solomon~\cite{ElkinS16}
considered efficient construction of light spanners. They showed how to
construct a spanner with stretch $(1+\eps)(2k-1)$, size $\Oeps(kn^{1+1/k})$ and
lightness $\Oeps(kn^{1/k})$ in time $O(km + \min(n\log n,m\alpha (n)))$.
Improving on this, a recent paper of Elkin and Neiman~\cite{ElkinN17c} uses
similar ideas to obtain stretch $(1+\eps)(2k-1)$, size $O(\log k\cdot
n^{1+1/k})$ and lightness $O(kn^{1/k})$ in expected time $O(m + \min(n\log
n,m\alpha(n)))$.

Several papers also consider efficient constructions of sparse spanners,
which are not necessarily light. Baswana and Sen~\cite{BaswanaS07} gave a
$(2k-1)$-spanner with $O(kn^{1+1/k})$ edges in $O(km)$ expected time. This was
later derandomized by Roditty et al.~\cite{RodittyTZ05} (while keeping the same
sparsity and running time). Recently, Miller et al.~\cite{MillerPVX15}
presented a randomized algorithm with $O(m+n\log k)$ running time and $O(\log
k\cdot n^{1+1/k})$ size at the cost of a constant factor in the stretch $O(k)$.

It is worth noting that for super-constant $k$, none of the above spanner
constructions obtain the optimal $O(n^{1+1/k})$ size or $O(n^{1/k})$ lightness
even if we allow $O(k)$ stretch. If we are satisfied with nearly-quadratic
running time, Elkin and Solomon~\cite{ElkinS16} gave a spanner with
$(1+\eps)(2k-1)$ stretch, $\Oeps(n^{1+1/k})$ size and $\Oeps(kn^{1/k})$
lightness in $O(kn^{2+1/k})$ time by extending a result of Roditty and
Zwick~\cite{RodittyZ11} who got a similar result but with unbounded lightness.
However, this construction still has an additional factor $k$ in the
lightness. Thus, the fastest known spanner construction obtaining optimal size
and lightness is the classic greedy spanner -- even if we allow $O(k)$ stretch
or $o(kn^{1/k})$ lightness.

We would like to emphasize that the case $k=\log n$ is of special interest. This is the point on the tradeoff curve allowing spanners of linear size and constant lightness.
Prior to this paper, the state of the art for efficient spanner constructions with constant lightness suffered from distortion at least $O(\log^2 n)$. See the discussion after \Cref{cor:opt_log} for further details.


A summary of spanner algorithms can be seen in \Cref{tab:spanners}.
\begin{table}[h]
	\small
	\centering
	\makebox[0pt][c]{
	\begin{tabular}{c|c|c|c|l}
		\toprule
		\bf Stretch & \bf Size & \bf Lightness & \bf Construction & \bf Ref\\
		\midrule
		$(2k-1)$&   $O \left(n^{1+1/k} \right)$ & $O\left(n/k\right)$ & $O \left(mn^{1+1/k} \right)$ & \cite{Althofer1993}$^*$\\
		$(2k-1)(1+\eps)$&   $O \left(n^{1+1/k} \right)$ & $O \left( k
		n^{1/k}\right)$ & $O \left(mn^{1+1/k} \right)$ & \cite{ChandraDNS92}$^*$\\
		$(2k-1)$ & $O(n^{1+1/k})$ & $\Omega(W)\text{ }^{**}$ & $O \left( kn^{2+1/k} \right)$ & \cite{RodittyZ11}\\	
		$(2k-1)$&   $O \left(kn^{1+1/k} \right)$ & $\Omega \left( n^{1+1/k} \right)\text{ }^{**}$ & $O \left(kmn^{1/k}\right)$ & \cite{ThorupZ05}{\tiny $^{\#}$}\\			
		$(2k-1)(1+\eps)$&   $O \left(n^{1+1/k} \right)$ & $O \left(kn^{1/k}\right)$ & $O \left( kn^{2+1/k} \right)$ & \cite{ElkinS16}
		\\								
		$(2k-1)(1+\eps)$&   $O \left(n^{1+1/k} \right)$ & $O \left(n^{1/k}
		\cdot k/\log k \right)$ & $O \left(mn^{1+1/k} \right)$ & \cite{ElkinNS14}$^*$\\
		$(2k-1)(1+\eps)$&   $O \left(n^{1+1/k} \right)$ & $O\left(n^{1/k}\right)$  & $n^{\Theta(1)}$ & \cite{ChechikW16}\\
		$(2k-1)(1+\eps)$&   $O \left(n^{1+1/k} \right)$ & $O \left(n^{1/k}\right)$ & $O \left(mn^{1+1/k} \right)$ & \cite{FiltserS16}$^*$\\
		\midrule
		$(2k-1)$&   $O \left(kn^{1+1/k} \right)$ & $\Omega(W)\text{ }^{**}$ & $O \left(km\right)$ & \cite{BaswanaS07,RodittyTZ05}\\     
		$(2k-1)(1+\eps)$&   $O \left(kn^{1+1/k} \right)$ & $O \left(kn^{1/k}\right)$ & $O \left(km + n \log n \right)$ & \cite{ElkinS16}\\	
		$O(k)$ &  $O(\log k\cdot n^{1+1/k})$ & $\Omega(W)$ & $O(m+n \cdot \log k)$ & \cite{MillerPVX15}{\tiny $^{\#}$}\\
		$(2k-1)(1+\eps)$ &  $O(\log k \cdot n^{1+1/k})$ &  $O \left(k \cdot n^{1/k} \right)$ & $O(m+n \cdot \log n)$ & \cite{ElkinN17c}{\tiny $^{\#}$}\\				  			 
		\midrule
		$(2k-1)(1+\eps)$ & $O \left(\log k\cdot n^{1+1/k} \right)$ & $\Omega(W)$ & $O(m+n\log k\log^{(s)}k)$ & \Cref{thm:SparseNoLight}\\			
		$(2k-1)(1+\eps)$ & $O \left(\log k \cdot n^{1+1/k} \right)$ & $O \left( \log k \cdot n^{1/k}  \right)$ & $O(m + n \cdot\log n)$ & \Cref{thm:fastCW}\\
		$(2k-1)(1+\eps)$ & $O\left(n^{1+1/k}\right)$ &
		$O\left(n^{1/k}\right)$ &
		$O(n^{2+1/k+\eps'})$ & \Cref{thm:slow_good}\\
		$O(k)$&   $O \left(n^{1+1/k} \right)$ & $O \left( n^{1/k} \right)$ & $O \left(m + n^{1+\eps'+1/k}\right)$ & \Cref{thm:fast_light}\\
		$O(\log n)/\delta$&   $O \left(n \right)$ & $1+\delta$ & $O \left(m + n^{1+\eps'}\right)$ & \Cref{cor:opt_log}\\
		\bottomrule
	\end{tabular}
	}
	\caption{Table of related spanner constructions. In the top of the table we list non-efficient spanner constructions. In the middle we list known efficient spanner constructions. In the bottom we list our contributions. 
	Results marked $^*$ are different analyses of the greedy spanner. 
	Results marked {\tiny $^{\#}$} are randomized.
	Lightness complexities marked
	$^{**}$ are from the analysis in \Cref{sec:bad_lightness} and $W$ denotes
	the maximum edge weight of the input graph. The bounds hold for any constant $\eps,\eps' > 0$.}
\label[table]{tab:spanners}
\end{table}

\subsection{Our results}
We present the first spanner obtaining the same near-optimal guarantees as the
greedy spanner in significantly faster time by obtaining a $(1+\eps)(2k-1)$
spanner with optimal size and lightness in $\Oeps(n^{2+1/k+\eps'})$ time. We also
present a variant of this spanner, improving the running time to $O(m + n\log
n)$ by paying a $\log k$ factor in the size and lightness. Finally, we present
an optimal $\Oeps(\log n)$-spanner which
can be constructed in $O(m + n^{1+\eps})$ time. This special case is of
particular interest in the literature (see e.g. \cite{BartalFN16,KoutisX16}). Furthermore,
all of our constructions are deterministic, giving the first subquadratic
deterministic construction without the additional dependence on $k$ in the size
of the spanner. As an important tool, we introduce a new deterministic
approximate incremental distance oracle which works in near-linear time for
maintaining small distances approximately. We believe this result is of
independent interest.

More precisely, we show the following theorems.
\begin{theorem}\label{thm:slow_good}
	Given a weighted undirected graph $G = (V,E,w)$ with $m$ edges and $n$
    vertices, any positive integer $k$, and  $\eps,\eps'>0$ where $\eps$ arbitrarily close to $0$ and $\eps'$ is a constant, one can
    deterministically construct an $(1+\eps)(2k-1)$-spanner of $G$ with
    $O_{\eps}(n^{1+1/k})$ edges and lightness $O_{\eps}(n^{1/k})$ in $O(n^{2+1/k+\eps'})$
    time.
\end{theorem}

\begin{theorem}\label{thm:fastCW}
	Given a weighted undirected graph $G = (V,E,w)$ with $m$ edges and $n$
	vertices, a positive integer $k\ge 640$, and
	$\epsilon > 0$, one can deterministically construct a    $(2k-1)(1+\eps)$-spanner of $G$ with $\Oeps(\log k\cdot n^{1+1/k})$ edges and
	lightness $\Oeps\left(\log k\cdot n^{1/k}\right)$ in time $O\left(m+n\log
	n\right)$.
\end{theorem}
Note that in \Cref{thm:fastCW} we require $k$ to be larger than $640$. This
is not a significant limitation, as for $k=O(1)$ \cite{ElkinS16} is already
optimal.

Our $O(\log n)$-spanner is obtained as a corollary of the following more
general result.
\begin{theorem}\label{thm:fast_light}
	Given a weighted undirected graph $G = (V,E,w)$ with $m$ edges and $n$
	vertices, any positive integer $k$ and constant $\eps' > 0$, one can
    deterministically construct an $O(k)$-spanner of $G$ with $O(n^{1+1/k})$
    edges and lightness $O(n^{1/k})$ in $O(m+n^{1+\eps'+1/k})$ time.
\end{theorem}
We note that the stretch $O(k)$ of \Cref{thm:fast_light} (and 
\Cref{cor:opt_log} below) hides an exponential factor in $1/\eps'$, thus we only
note the result for constant $\eps'$.
Bartal et. al. \cite{BartalFN16} showed that given a spanner construction that for every  $n$-vertex weighted graph produces a $t(n)$-stretch spanner with $m(n,t)$ edge and  $l(n,t)$ lightness in $T(n,m)$ time, then for every parameter $0<\delta<1$ and every graph $G$, one
can construct a $t/\delta$-spanner with $m(n,t)$ edges and $1+\delta\cdot l(n,t)$ lightness in $T(n,m)+O(m)$ time.
Plugging $k=\log n$ and using this reduction w.r.t. $\delta$ in \Cref{thm:fast_light}, and $\delta'=\frac{\delta}{\log\log n}$ in \Cref{thm:fastCW}, we get
\begin{corollary}\label{cor:opt_log}
	Let $G = (V,E,w)$ be a weighted undirected $n$-vertex graph, let $\eps'>0$
    be a constant and $\delta > 0$ be a parameter arbitrarily close to $0$.
     Then one can construct a spanner of $G$ with:
    \begin{enumerate}
    	\item  $O(\log n)/\delta$ stretch, $O(n)$ edges and
    	$1+\delta$ lightness in time $O(m + n^{1+\eps'})$.
    	\item  $O(\log n\log\log n)/\delta$ stretch, $O(n\log\log n)$ edges and $1+\delta$ lightness
    	 in time $O(m + n\log n)$.
    \end{enumerate}
   
\end{corollary}
\Cref{cor:opt_log} above should be compared to previous attempts to efficiently construct a spanner with constant lightness.
Although not stated explicitly, the state-of-the-art algorithms of
\cite{ElkinS16,ElkinN17c}, combined with the lemma from \cite{BartalFN16}, provide an efficient spanner construction with $1+\delta$ lightness, $O(n\log\log n)$ edges and only $O(\log^2 n/\delta)$ stretch.

We emphasize, that
\Cref{cor:opt_log} is the first sub-quadratic construction of spanner with
optimal size and lightness for any non-constant $k$.

In order to obtain \Cref{thm:fast_light} we construct the following
deterministic incremental approximate distance oracle with near-linear
total update time for maintaining small distances. We believe this result is of
independent interest, and discuss it in more detail in the related work section
below and in \Cref{sec:overview}.
\begin{theorem}\label{thm:dist_oracle}
    Let $G$ be a graph that undergoes a sequence of $m$ edge insertions. For
    any constant $\eps' > 0$ and parameter $d\ge 1$ there exists a data
    structure which processes the $m$ insertions in total time
    $O(m^{1+\eps'}\cdot d)$ and can answer queries at any point in the sequence
    of the following form. Given a pair of nodes $u,v$, the oracle gives, in
    $O(1)$ time, an estimate $\hat{d}(u,v)$ such that $\hat{d}(u,v)\ge d(u,v)$
    and if $d(u,v)\le d$ then $\hat{d}(u,v) = O(1)\cdot d(u,v)$.
\end{theorem}
\Cref{thm:dist_oracle} assumes that $\eps'$ is constant; the $O$-notation hides a factor exponential in $1/\eps'$ for both total update time and stretch whereas the query time bound only hides a factor of $1/\eps'$.

We also obtain the following sparse, but not necessarily light, spanner in
linear time as a subroutine in proving \Cref{thm:fastCW}.

\begin{restatable}[]{theorem}{SparseNoLight}
	\label{thm:SparseNoLight}
	Given a weighted undirected graph $G = (V,E,w)$ with $m$ edges and $n$
	vertices, any positive integer $k$, any $\epsilon > 0$, and any positive integer $s = O(1)$, one can
	deterministically construct a $(2k-1)(1+\eps)$-spanner of $G$ with
	$O_\eps(n^{1+1/k}\cdot \log k)$ edges; the running time is $O(m + n\log^{(s-1)}k\log k)$ and if $k\leq \log n/\log^{(s+1)}n$, the running time is $O(m)$.
\end{restatable}
Here, the function 
$\log^{(s)}$ is $\log$ concatenated with itself $s$ times. Specifically, $\log^{(0)}n=n$, $\log^{(1)}n=\log n$, and in general for $s\ge 1$, $\log^{(s)}n=\log(\log^{(s-1)}n)$. $\log^*n$ is the minimum index $s$ such that $\log^{(s)}n\le2$.

Note that since we may assume that $k = O(\log n)$, the time bound of~\Cref{thm:SparseNoLight} is linear for almost all choices of $k$ and very close to linear for any choice of $k$.
\subparagraph{Organization} In \Cref{sec:fw} we state our framework that used in \Cref{thm:fast_light,thm:fastCW,thm:slow_good}.
\Cref{thm:fast_light} is proved in \Cref{sec:fast_greedy}, and  \Cref{thm:slow_good} in \Cref{sec:quad_greedy}. \Cref{thm:fastCW} is proved in \Cref{sec:fastCW}.
\Cref{thm:dist_oracle} is proved in \Cref{sec:DistOracle}. The proof of \Cref{thm:SparseNoLight} appears in \Cref{sec:proofOfSparseOnly}.

\subsection{Related work}
Closely related to graph spanners are \emph{approximate distance oracles
(ADOs)}. An ADO is a data structure which, after preprocessing a graph $G$, is
able to answer distance queries approximately. Distance oracles are studied
extensively in the literature (see e.g.
\cite{ThorupZ05,DBLP:conf/soda/Wulff-Nilsen13,Chechik14,Chechik15}) and often
use spanners as a building block. The state of the art static distance oracle
is due to Chechik~\cite{Chechik15}, where a construction of space
$O(n^{1+1/k})$,
stretch $2k-1$, and query time $O(1)$ is given. Our distance oracle of
\Cref{thm:dist_oracle} should be compared to the result of Henzinger, et
al.~\cite{HenzingerKN16}, who gave a deterministic construction for incremental
(or decremental) graphs with a total update time of $\Oeps(mn \log n)$, a
query time of $O(\log\log n)$ and stretch $1+\eps$. For our particular
application, we require
near-linear total update time and only good stretch for short distances, which are commonly the most troublesome
when constructing spanners. It should be added that Henzinger et al. give a general deterministic data structure for choosing centers, i.e., vertices which are roots of shortest path trees maintained by the data structure. While this data structure may be fast when the total number of centers is small, we need roughly $n$ centers and it is not clear how this number can be reduced. Having this many centers requires at least order $mn$ time with their data structure.

To achieve our fast update time bound, we are interested in trading worse
stretch for distances above parameter $d$ for construction time. Roditty and
Zwick~\cite{RodittyZ12} gave a randomized distance oracle for this case,
however their construction does not work against an adaptive adversary as is
required for our application, where the edges to be inserted are determined by
the output to the queries of the oracle (see \Cref{sec:overview} for more
discussion on this).
Removing the assumption of a non-adaptive adversary in dynamic graph algorithms
has seen recent attention at prestigious venues, e.g.
\cite{Wul17,DBLP:conf/stoc/BhattacharyaHN16}. 
Our new incremental approximate distance oracle for short distances given in
\Cref{thm:dist_oracle} is deterministic and thus is robust against such
an adversary, and we believe it may be of independent interest as a building
block in deterministic dynamic graph algorithms.

For unweighted graphs, there is a folklore spanner construction by Halperin
and Zwick~\cite{HZ96} which is optimal on all parameters. The construction time is
$O(m)$, it has $O(n^{1+1/k})$ edges and $2k-1$ stretch. In \Cref{sec:fastCW}
we will use this spanner as a building block in proving \Cref{thm:fastCW}.

\section{Preliminaries}
Consider a weighted graph $G=(V,E,w)$, we will abuse notation and refer to as $E$ both a set of edges and the graph itself.
$d_G$ will denote the shortest path metric (that is $d_G(v,u)$ is the weight of the lightest path between $v,u$ in $G$.
Given a subset $V'$ of $V$,  $G[V']$ is the induced graph by $V'$. That is it has $V'$ as it vertices, $E\cap {V' \choose2}$ as its edges and $w$ as weight function.
The \emph{diameter} of a vertex set $V'$ in a graph $G'$ $\diam_{G'}(V')=\max_{u,v\in V'}d_{G'}(u,v)$ is the maximal distance between two vertices in $V'$ under the shortest path metric induced by $G'$.
For a set of edges $A$ with weight function $w$, the \emph{aspect ratio} of $A$ is $\max_{e\in A}w(e)/\min_{e\in A}w(e)$. 
The \emph{sparsity} of $A$ is simply $|A|$ its size. 

We will assume that $k=O(\log n)$ as the guarantee for lightness and sparsity will not be improved by picking larger $k$. 
Instead of proving $(1+\eps)(2k-1)$ bound on stretch, we will prove only $(1+O(\eps))(2k-1)$ bound. This is good enough, as Post factum we can scale $\eps$ accordingly.
By $\Oeps$ we denote asymptotic notation which hides polynomial factors of $1/\eps$, that is $\Oeps(f)=O(f)\cdot \poly(\frac{1}{\eps})$.

\section{Paper overview}\label{sec:overview}

\subparagraph{General framework}
\Cref{thm:slow_good,thm:fast_light,thm:fastCW} are generated via a general framework.
The framework is fed two algorithms for spanner constructions:
$A_1$, an algorithm suitable for graphs with small aspect ratio, and $A_2$, an
algorithm that returns a sparse spanner, but with potentially unbounded
lightness. We consider a partition of the edges into groups according to their weights. 
For treating most of the groups we use exponentially growing clusters,
partitioning the edges according to weight. Each such group has bounded aspect
ratio, and thus we can use $A_1$. Due to the exponential growth rate, we show
that the contribution of all the different groups is converging. Thus only the
first group is significant.
However, with this approach we need a special treatment for edges of small
weight. 
This is, as using the previous approach, the number of clusters needed to treat light edges is unbounded.
Nevertheless, these edges have small impact on the lightness
and we may thus use algorithm $A_2$, which ignores this property.

The main work in proving \Cref{thm:slow_good,thm:fast_light,thm:fastCW} is
in designing the algorithms $A_1$ and $A_2$ described briefly below.

\subparagraph{Approximate greedy spanner}
The major time consuming ingredient of the greedy spanner algorithm is its
shortest path computations. By instead considering approximate shortest
path computations we significantly speed this process up. 
We are the first to apply this idea on general graphs, while it has previously been applied by \cite{DasN97,FiltserS16} on particular graph families. 
Specifically, we consider the following algorithm: given some parameters
$t<t'$, initialize $H\leftarrow \emptyset$ and consider the edges $(u,v)\in E$
according to increasing order of weight. If $d_H(u,v) > t'\cdot w(u,v)$
the algorithm is obliged to add $(u,v)$ to $H$. If $d_H(u,v)<t\cdot w(u,v)$,
the algorithm is forbidden to add $(u,v)$ to $H$. Otherwise,
the algorithm is free to include the edge or not.
As a result, we will get spanner with stretch $t'$, which has the same
lightness and sparsity guarantees of the greedy $t$-spanner. Note however, that the resulting spanner is not necessarily a subgraph of any greedy spanner.

We obtain both \Cref{thm:slow_good} and \Cref{thm:fast_light} using this
approach via an incremental approximate distance oracle. It is important to note that the
edges inserted into $H$ using this approach depend on the answers to the
distance queries. It is therefore not possible to use approaches that do not
work against an \emph{adaptive adversary} such as the result of Roditty and
Zwick~\cite{RodittyZ12}, which is based on random sampling. Furthermore, this
is the case even if we allow the spanner construction itself to be randomized. In order to
obtain \Cref{thm:slow_good}, we use our previously described framework coupled
with the ``approximately greedy spanner'' using an incremental
$(1+\eps)$-approximate distance oracle of Henzinger et
al.~\cite{HenzingerKN16}. For \Cref{thm:fast_light}, we present a novel
incremental approximate distance oracle, which is described below. This is the
main technical part of the paper and we believe that it may be of independent
interest.

\subparagraph{Deterministic distance oracle}
The main technical contribution of the paper and key ingredient in proving
\Cref{thm:fast_light} is our new deterministic incremental approximate
distance oracle of \Cref{thm:dist_oracle}. The oracle supports approximate
distance queries of pairs within some distance threshold, $d$. In particular,
we may set $d$ to be some function of the stretch of the spanner in
\Cref{thm:fast_light}. Similar to previous work on distance oracles, we have
some parameter, $k$, and maintain $k$ sets of nodes $\emptyset = A_{k-1}
\subseteq \ldots \subseteq A_0 = V$, and for each $u\in A_i$ we maintain a
ball of radius $r\le d_i$. Here, $d_i$ is a distance
threshold depending on the parameter $d$ and which set $A_i$ we are
considering, and $r$ is chosen such that the total degree of nodes in the
ball of radius $r$ from $u$ is relatively small. The implementation of each ball can be thought of as an incremental Even-Shiloach tree. The set $A_{i+1}$ is then
chosen as a maximal set of nodes with disjoint balls (see \Cref{fig:oracle} in
\Cref{sec:oracle}). Here we use the fact that the vertices in $A_{i+1}$ are centers of disjoint balls in $A_i$ to argue that $A_{i+1}$ is much smaller than $A_i$. The decrease in size of $A_{i+1}$ pays for an increase in the maximum ball radius $d_i$ at each level. The ball of a node $u$ may grow in size during
edge insertions. In this case, we freeze the ball associated with $u$, shrink
the radius $r$ associated with $u$, and create a new ball with the new radius.
Thus, for each $A_i$ we end up with $O(\log d)$ different radii for which we
pick a maximal set of nodes with disjoint balls. For each node $u_i\in A_i$ we
may then associate a node $u_{i+1}\in A_{i+1}$ whose ball intersects with $u_i$'s. We
use these associated nodes in the query to ensure that the path distance we find is not
``too far away'' from the actual shortest path distance. Consider a query pair $(u,v)$. Then the query algorithm iteratively finds a sequence of vertices $u = u_0 \in A_0, u_1 \in A_1,..., u_i \in A_i$; $d_i$ is picked such that if $v$ is not in the ball centered at $u_i$ with radius $d_i$ then the shortest path distance between $u$ and $v$ is at least $d$ and the algorithm outputs $\infty$. Otherwise, the algorithm uses the shortest path distances stored in the balls that it encounters to output the weight of a $uv$-path $(u = u_0)\leadsto u_1\leadsto\ldots\leadsto u_i\leadsto v$ as an approximation of the shortest path distance between $u$ and $v$.

\subparagraph{Almost linear spanner} Chechik and Wulff-Nilsen \cite{ChechikW16} implicitly used our
general framework, but used the (time consuming) greedy spanner both as their
$A_2$ component and as a sub-routine in $A_1$.
We show an efficient alternative to the algorithm of \cite{ChechikW16}. For
the $A_2$ component we provide a novel sparse spanner construction
(\Cref{thm:SparseNoLight}, see paragraph below).
For $A_1$, we perform a hierarchical clustering, while
avoiding the costly exact diameter computations used in \cite{ChechikW16}.
Finally, we replace the greedy spanner used as a sub-routine of
\cite{ChechikW16} 
by an efficient spanner 
 that exploits bounded aspect ratio (see \Cref{lem:fast_spanner_aspect}). This spanner can be seen as a careful adaptation of Elkin and Solomon~\cite{ElkinS16} analyzed in the case of bounded aspect ratio. The idea here is (again) a hierarchical
partitioning of the vertices into clusters of exponentially increasing size.
However, here the growth rate is only 
$(1+\eps)$. Upon each clustering we construct a super graph with clusters as
vertices and graph edges from the corresponding weight scale as inter-cluster
edges.
To decide which edges in each scale add to our spanner, we execute the
extremely efficient spanner of Halperin and Zwick~\cite{HZ96} for unweighted graphs.

\subparagraph{Linear time sparse spanner}
As mentioned above we provide a novel sparse spanner construction as a building
block in proving \Cref{thm:fastCW}. Our construction is based on partitioning
edges into $O_\eps(\log k)$ ``well separated'' sets $E_1,E_2,\ldots$, such that
the ratio between $w(e)$ and $w(e')$ for edges $e,e'\in E_i$ is either a
constant or at least $k$. This idea was previously employed by Elkin and
Neiman~\cite{ElkinN17c} based on \cite{MillerPVX15}. For these well-separated
graphs, Elkin and Neiman used an involved
clustering scheme based on growing clusters according to exponential
distribution, and showed that the expected number of inter-cluster edges, in
all levels combined, is small enough.
We provide
a linear time \emph{deterministic} algorithm with an arguably simpler
clustering scheme. Our clustering is based upon the clusters defined implicitly
by the spanner for unweighted graphs of Halperin and Zwick \cite{HZ96}.
In particular, we introduce a charging scheme, such that each edge added to our
spanner is either paid for by a large cluster with many coins, or significantly
contributing to reduce the number of clusters in the following level.

\section{A framework for creating light spanners efficiently}\label{sec:fw}
In this section we describe a general framework for creating spanners, which we
will use to prove our main results.
The framework is inspired by a standard clustering approach (see
e.g.~\cite{ElkinS16} and \cite{ChechikW16}).
The spanner framework takes as input two spanner algorithms for restricted
graph classes, $A_1$ and $A_2$, and produces a spanner algorithm for general
graphs. The algorithm $A_1$ works for graphs with unit weight MST edges and
small aspect ratio, and $A_2$ creates a small spanner with no guarantee for
the lightness. The main work in showing
Theorems~\ref{thm:slow_good},~\ref{thm:fastCW}, and~\ref{thm:fast_light} is to
construct the
algorithms, $A_1$ and $A_2$, that go into \Cref{lem:fw} below. We do this
in Sections~\ref{sec:fast_greedy} and~\ref{sec:fastCW}. The framework is
described in the following lemma.

\begin{lemma}\label{lem:fw}
    Let $G=(V,E)$ be a weighted graph with $n$ nodes and $m$ edges and let
    $k>0$ be an integer, $g>1$ a fixed parameters and $\eps>0$.
    Assume that we are given two spanner construction algorithms
    $A_1$ and $A_2$ with the following properties:
    \begin{itemize}
        \item $A_1$ computes a spanner of stretch $f_1(k)$, size $\Oeps(s_1(k)\cdot
            n^{1+1/k})$ and lightness $\Oeps(l_1(k)\cdot n^{1/k})$ in time
            $T_1(n,m,k)$ when given a graph with maximum weight $g^k$, where
            all MST edges have weight $1$. Moreover, $T_1$ has the property
            that
            $\sum_{i=0}^{\infty}T_{1}\!\left(\frac{n}{g^{ik}},m_{i},k\right)=O\left(T_{1}\!\left(n,m,k\right)\right)$,
            where $\sum_im_i=m+O(n)$.
        \item $A_2$ computes a spanner of stretch $f_2(k)$ and size
            $\Oeps(s_2(k)\cdot n^{1+1/k})$ in time $T_2(n,m,k)$.
    \end{itemize}
    Then one can compute a spanner of stretch
    $\max((1+\eps)f_1(k),f_2(k))$, size $\Oeps((s_1(k)+s_2(k))n^{1+1/k})$, and
    lightness $\Oeps((l_1(k) + s_2(k))\cdot n^{1/k})$ in time
    $O(T_1(n,m,k)+T_2(n,m,k)+m + n\log n)$.
\end{lemma}
	As an example, let us assume that we have both an optimal spanner algorithm for
	graphs with small aspect ratio, and an optimal spanner algorithm for sparse spanners in weighted graph. Specifically, we have algorithm $A_1$ that given a graph as above creates a $(1+\eps)(2k-1)$-spanner with $O_\eps(n^{1+1/k})$ edges and
	lightness $O_\eps(n^{1/k})$ in $O_\eps(m+n\log n)$ time. 
	In addition we have algorithm $A_2$ that returns an $(1+\eps)(2k-1)$-spanner with $O_\eps(n^{1+1/k})$ edges in $O_\eps(m)$ time.
	Then, given a general graph, \Cref{lem:fw} provide us with a $(1+\eps)(2k-1)$-spanner of $O_\eps(n^{1+1/k})$ size and $O_\eps(n^{1/k})$ lightness, in time $O_\eps(m+n\log n)$.

Before proving \Cref{lem:fw} we need to describe the clustering
approach. The main tool needed is what we call an
\emph{$(i,\eps)$-clustering}. This clustering procedure is performed on graphs
where all the MST edges have unit weight. Let $G, g, \eps, k$ be as in
\Cref{lem:fw}, then we say that an $(i,\eps)$-clustering is a
partitioning of $V$ into clusters $C_1,\ldots,C_{n_i}$, such that each $C_j$
contains at least $\eps g^{ik}$ nodes and has diameter at most $4\eps g^{ik}$
(even when restricted to MST edges of $G$). Let $G_i$ denote the graph
obtained by contracting the clusters of such an $(i,\eps)$-clustering of $G$,
and keeping the MST edges only. Then $G_i$
has $n_i$ nodes, and we can construct $G_i$ from $G_{i-1}$ as follows.
Start at some vertex $v$ in $G_{i-1}$ (corresponding to an
$(i-1,\eps)$-cluster) and iteratively grow an $(i,\eps)$-cluster
$\varphi_v$ by joining arbitrary un-clustered neighbors to $\varphi_v$ in
$G_{i-1}$ one at a time. If the number of original vertices in $\varphi_v$
reaches $\eps g^{ik}$, make $\varphi_v$ into an $(i,\eps)$-cluster, where
the current vertices in $\varphi_v$ are called its \emph{core}.
We argue that the diameter of the core is bounded by $\eps g^{ik}+4\eps
g^{(i-1)k}$. If the vertex $v$ (from $G_i$) already contains $\eps g^{ik}$ vertices, then $\varphi_v=v$ and by the induction hypothesis the diameter of $\varphi_v$ is at most $4\eps g^{(i-1)k}$.
Otherwise ($|v|<\eps g^{ik}$), consider the last vertex $u\in G_{i-1}$ to join $\varphi_v$. As $u$ joins $\varphi_v$, necessarily  $|\varphi_v\setminus u|<\eps g^{ik}$. In particular, the diameter of $\varphi_v\setminus u$ (restricted to MST edges) is at most $g^{ik}-1$. The diameter of $u$ is at most $4\eps g^{(i-1)k}$ and therefore the diameter of $\varphi_v$ is indeed bounded by $\eps g^{ik}+4\eps
g^{(i-1)k}$.

We perform this procedure starting at an un-clustered vertex until all vertices
of $G_{i-1}$ belong to some $(i,\eps)$-cluster. In the case where
$\varphi_v$ has no un-clustered neighbors, but does not contain $\eps g^{ik}$
vertices, we simply merge it with an existing $(i,\eps)$-cluster
$\varphi_u$ via an MST-edge to the core of $\varphi_u$. Note that the size
of $\varphi_v$, and therefore its diameter, before the merging is at most
$\eps g^{ik}-1$ (as each cluster is connected when restricting to MST
edges). To show that this gives a valid $(i,\eps)$-clustering, consider
an $(i,\eps)$-cluster $\varphi_{v}$ with core $\tilde{\varphi}_v$. Suppose that
the ``sub-clusters'' $\tilde{\varphi}_{u_1},\dots,\tilde{\varphi}_{u_s}$ were
merged into $\tilde{\varphi}_{v}$ during this process. The diameter of
$\varphi_{v}$ is then bounded by
\[
\diam(\tilde{\varphi}_v)+2+\max_{j,j'}(\diam(\tilde{\varphi}_{u_{j}})+\diam(\tilde{\varphi}_{u_{j'}}))\le (\eps g^{ik}+4\eps g^{(i-1)k})+2\eps g^{ik}\le4\eps g^{ik}~.
\]
Moreover, the size of $\varphi_v$ is at least the size of its core,
$|\tilde{\varphi}_v|\ge \eps g^{ik}$.
See figure \Cref{fig:i-eps-clustering} for illustration.
\begin{figure}
	\begin{center}
		\includegraphics[width=0.6\textwidth]{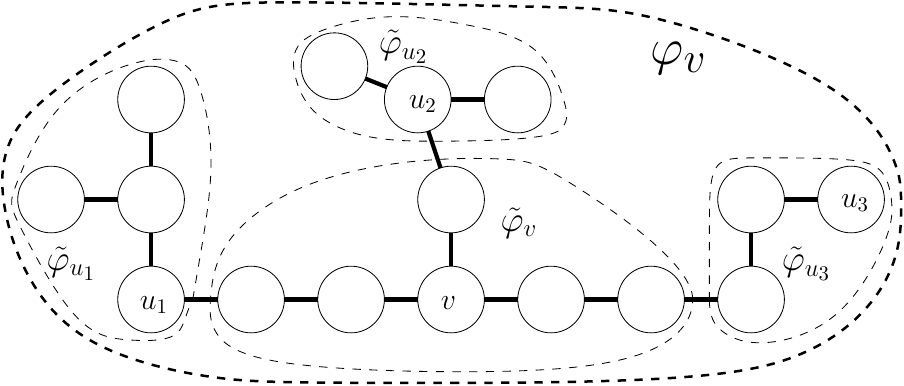}
		\caption{\small 
			The small cycles represent $(i-1,\eps)$-clusters, which are the vertices of $G_{i-1}$. The $(i-1,\eps)$-cluster $v$ iteratively grows a cluster
            around itself until it contains $\eps g^{ik}$ original vertices. This current
            cluster is called $\tilde{\varphi}_v$. $\tilde{\varphi}_v$ also
            called the core of $\varphi_v$, an $(i,\eps)$-cluster that we will
            have at the end of the process.
			Afterwards, the $(i-1,\eps)$-cluster $u_1$ iteratively grow a
            cluster around itself. When the temporary cluster is
            $\tilde{\varphi}_{u_1}$ there are no outgoing edges to unclustered
            vertices. However, since $G_{i-1}$ is connected, there is some
            outgoing edge. Necessarily the second endpoint of this edge belongs to the core of an existing cluster (here,
            $\tilde{\varphi}_v$), that is, as $\tilde{\varphi}_v$ stop growing while still having unclustered neighbors.
            All the vertices of $\tilde{\varphi}_{u_1}$ joins the cluster of $v$.
			In a similar manner, $\tilde{\varphi}_{u_2}$ and
            $\tilde{\varphi}_{u_3}$ are also joined into the cluster of $v$. In
            the end of the algorithm, the $(i,\eps)$-cluster of $v$ is
            $\varphi_{v}=\tilde{\varphi}_{v}\cup\tilde{\varphi}_{u_1}\cup\tilde{\varphi}_{u_2}\cup\tilde{\varphi}_{u_3}$.
			}
		\label{fig:i-eps-clustering}
	\end{center}
	\hrule
\end{figure}

Note that we have $n_i\le \frac{n}{\eps g^{ik}}$. Using
the above procedure, we can construct the $(i+1,\eps)$ clustering from the
$(i,\eps)$ clustering in $O(n_{i})$ time. Therefore we can construct the
clusters for all the levels in
$O\left(\sum_{i\ge0}n_{i}\right)=O\left(n\sum_{i\ge0}\frac{1}{g^{ik}}\right)=O(n)$ 
time (if we are given the MST).

With this tool in hand, we may now prove \Cref{lem:fw}.
\begin{proof}[Proof of \Cref{lem:fw}]
    The proof constructs an algorithm consisting of two phases. The
    \emph{preparation phase} where $A_2$ is used to reduce the problem to a
    graph where all MST edges have weight $1$,
    and the \emph{bootstrapping phase} where we perform an iterative
    clustering of the graph to obtain several graphs with small aspect-ratio,
    where we can apply $A_1$.

    \textbf{Preparation phase:}
    Let $T$ be an MST of $G$ and let $w' = \sum_{(u,v)\in T}
    \frac{w(u,v)}{n-1}$. Define $G_2$ to be $G$ with all edges of weight
    greater than $w'/\eps$ removed and let $H^2$ be the spanner resulting from
    running $A_2$ on $G_2$. Next, we construct $G_1$ from $G$ as follows. First,
    round up the weight of each edge in $G$ to the nearest multiple of $w'$.
    For each edge $e\in T$ subdivide
    it such that each edge of the resulting
    MST has weight $w'$ 
    \footnote{Formally, for an edge $e=\{v,v'\}\in T$ of weight $cw'$, we add $c-1$ new vertices $u_1,\dots,u_{c-1}$ and replace the edge $\{v,v'\}$ with the edges $\{v,u_1\},\{u_1,u_2\},\dots,\{u_{c-2},u_{c-1}\},\{u_{c-1},v'\}$, all with weight $w'$.}. 
    As the weight of each edge increase by at most an additive factor of
    $w'$, the weight of the MST increase by at most $(n-1)w'\le w(T)$. The new number of vertices is bounded by $\sum_{(u,v)\in T}\left\lceil \frac{w(u,v)}{w'}\right\rceil \le(n-1)+\frac{1}{w'}\cdot\sum_{(u,v)\in T}w(u,v)<2n$. 
    Finally, divide the weight of each edge by $w'$. This finishes the construction of $G_1$.

    \textbf{Bootstrapping phase:}
    We will now use $A_1$ to make a spanner $H^1$ for the graph $G_1$ created
    above. We start by partitioning the edges into sets $E_i$, where
    $E_i$ contains all edges of $G_1$ with weights in
    $[g^{ik},g^{(i+1)k})$. Note that since each MST edge of $G_1$ has
    weight $1$ we only need to consider edges with weight up to $O(n)$.
    Next, we let $T$ be an MST of $G_1$ and for all $i=0,1,\ldots,O(\log n)$ we
    create $T_i$ by contracting all clusters of an $(i,\eps)$-clustering of
    $T$, where the $(i,\eps)$-clustering is computed as described above. Note
    that $T_i$ is also a tree since each cluster is a connected subtree of $T$.
    We now construct graphs $G_i$ by taking $T_i$ and adding any minimum
    weight edge of $E_i$ going between each pair of clusters (i.e.
    nodes corresponding to clusters). Finally, we divide the weight of each
    non-MST edge of $G_i$ by $g^{ik}$. This gives us a graph with maximum
    weight $g^k$, where MST edges have weight $1$. We call this new weight
    function $w_i$. Let $H_i$ be the spanner
    obtained by running algorithm $A_1$ on $G_i$. Finally, let $H^1$ be the
    union of all $H_i$s, where each edge of $H_i$ is replaced by the
    corresponding edge(s) from $G$.

    \textbf{Analysis:}
    We set the final spanner $H = MST(G) \cup H^1\cup H^2$. To bound the
    stretch of $H$ first note that any edge of $G_2$ has stretch at most
    $f_2(k)$ from $H^2$. What remains is to bound the stretch of non-MST edges
    $(u,v)$ with $w(u,v) \ge w'/\eps$. First, observe that the rounding
    procedure used to create $G_1$ can at most increase the weight of $(u,v)$
    in $G_1$ by a factor of $(1+\eps)$ compared to $G$.

    Now assume that $(u,v)\in E_i$ for some $i$. Let $\vpi_{u}$ and
	$\vpi_{v}$ denote the clusters containing $u$, respectively, $v$ in $G_i$.
	If $\vpi_{u} = \vpi_{v}$ we know that the distance between $u$ and $v$
	using the MST is at most $4\eps g^{ik}$ and we are done. Thus, assume that
    $\vpi_{u} \ne \vpi_{v}$. By definition of $G_i$, there
	must be some edge $(\vpi_{u},\vpi_{v})$ in $G_i$ with
    $w_i(\vpi_{u},\vpi_{v})\le (1+\eps)\cdot w(u,v)/g^{ik}$. We know that there
    is a path $\{\vpi_{u}=\vpi_{z_0}\vpi_{z_1}\dots \vpi_{z_s}=\vpi_{v}\}$
    from $\vpi_{u}$ to $\vpi_{v}$ in $H_i$ of length at most $f_1(k)\cdot
    w_i(\vpi_{u},\vpi_{v})$.
    Recall that the minimum weight in $H_i$ is $1$, thus we have $s \le
    f_1(k)\cdot w_i(\vpi_u,\vpi_v)$. Furthermore, the diameter of each cluster,
    $\vpi_{z_q}$, is at most $4\eps g^{ik}$. We now conclude (see
    \Cref{fig:stretchCluster} for illustration)
	\begin{align}
        d_{H}(u,v) &
        \le\sum_{q=0}^{s}\diam_{T}(\vpi_{z_{q}})+g^{ik}\sum_{q=0}^{s-1}w_i(\vpi_{z_q},
        \vpi_{z_{q+1}}) \nonumber\\
        &\le (s+1)\cdot 4\eps g^{ik}+g^{ik}f_1(k)\cdot
        w_i(\vpi_u,\vpi_v)\nonumber\\
        &\le (1+8\eps)\cdot g^{ik}\cdot f_1(k)\cdot w_i(\vpi_u,\vpi_v)\nonumber\\
        &\le (1+O(\eps))\cdot f_1(k)\cdot w(u,v) \label{eq:cluster-Stretch}
	\end{align}
    
	\begin{figure}
		\begin{center}
			\includegraphics[width=0.45\textwidth]{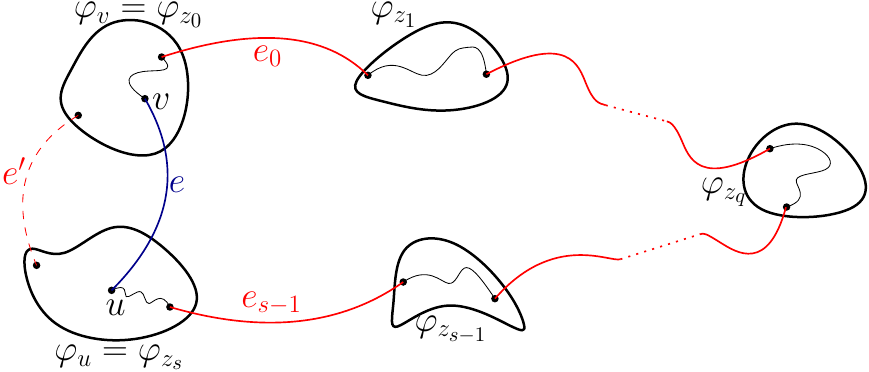}
			\caption{\small 
				$e=(u,v)$ is an edge (colored in blue) in $E_i$, such that $u$ and $v$ belong to the $i$-clusters  $\vpi_u,\vpi_v$, respectively.
				The closed bold black curves represent $i$-clusters. The red edges represent edges in $H_i$. The thin black curves represent MST paths.
				There is an edge $e'$ between $\vpi_{v}$ to $\vpi_{u}$ in $G_i$. Therefore $H_i$ contains a short path between $\vpi_{v}$ to $\vpi_{u}$.}
			\label{fig:stretchCluster}
		\end{center}
		\hrule
	\end{figure}

    Next we consider the size and lightness of $H$. First we see that, since
    $G_2$ is a subgraph of $G$, the spanner $H^2$ has size at
    most $O(s_2(k)\cdot n^{1+1/k})$. Furthermore since every edge in $G_2$
    has weight at most $w'/\eps$ the total weight of $H^2$ is
    \[
        O\!\left(s_2(k)\cdot n^{1+1/k} \cdot \frac{w(MST(G))}{(n-1)\eps}\right)
        = \Oeps\!\left(s_2(k)\cdot n^{1/k}\right)\ .
    \]
    Recall that $n_i$ is the number of $(i,\eps)$ clusters, and therefore also the number of nodes in $T_i$. We can bound the total weight of
    $H^1$ by
    \begin{align*}
        O\!\left(\sum_{i=0}^\infty g^{ik} \cdot w(T_i)\cdot l_1(k)\cdot
            n_i^{1/k}\right)
        &= O\!\left(\sum_{i=0}^\infty g^{ik}\cdot n_i\cdot l_1(k)\cdot
            n_i^{1/k}\right) \\
        &= \Oeps\!\left(n\cdot \sum_{i=0}^\infty l_1(k)\cdot
            \left(\frac{n}{g^{ik}}\right)^{1/k}\right) \\
        &= \Oeps\!\left(l_1(k)\cdot n^{1+1/k}\cdot \sum_{i=0}^\infty
            \frac{1}{g^i}\right) \\
        &= \Oeps\!\left(l_1(k)\cdot n^{1+1/k}\right)\ .
    \end{align*}
    Since the MST of $G_1$ has weight $n-1$ it follows that
    $H^1$ has lightness $O(l_1(k)\cdot n^{1/k})$ w.r.t. $G_1$ and thus also
    $G$. The size can be bounded in a similar fashion.

	The total running time of the algorithm is $O(m + n\log n)$ to find the
    MST of $G$ and divide edges to $G_1$ and $G_2$, $T_2(n,m,k)$ for creating $H^2$, 
    $O(n)$ for creating the different $(i,\eps)$-clusters and additional $O(m+nk)=O(m+n\log n)$ to create the graphs $G_i$,
    as described above. What is left is to bound the time needed to create
    the spanners $H_i$. Let $m_i = |E_i|$, then this time can be bounded by
	\[
	O\!\left(\sum_{i=0}^{\infty}m_{i}+n_{i}+T_{1}(n_{i},m_{i},k)\right)=O\!\left(m+\sum_{i=0}^{\infty}T_{1}\!\left(\frac{n}{g^{ik}},m_{i},k\right)\right)=O\!\left(m+T_{1}\!\left(n,m,k\right)\right)\ .\qedhere
	\]

\end{proof}

\section{Efficient approximate greedy
spanner}\label{sec:fast_greedy}
In this section we will show how to efficiently implement algorithms $A_1$ and
$A_2$ of \Cref{lem:fw} in order to obtain
\Cref{thm:slow_good,thm:fast_light}. We do this by implementing an
``approximate-greedy'' spanner, which uses an incremental approximate distance
oracle to determine whether an edge should be added to the spanner or not.

We first prove \Cref{thm:fast_light} and then show in \Cref{sec:quad_greedy}
how to modify the algorithm to give \Cref{thm:slow_good}. We will use
\Cref{thm:dist_oracle} as a main building block, but defer the proof of this
theorem to \Cref{sec:DistOracle}. Our $A_1$
is obtained by the following lemma giving stretch $O(k)$ and optimal size
$O(n^{1+1/k})$ and lightness $O(n^{1/k})$ for small weights. 
\begin{lemma}\label{lem:fast_greedy}
    Let $G = (V,E,w)$ be an undirected graph with $m = |E|$ and $n =
    |V|$ and integer edge weights bounded from above by $W$. Let $k$ be
    a positive integer and let $\eps'>0$ be a constant. Then one can deterministically construct an $O(k)$-spanner of $G$ with size
    $O(n^{1+1/k})$ and lightness $O(n^{1/k})$ in time $O\!\left(m +
    kWn^{1+1/k+\eps'}\right)$.
\end{lemma}
We note that \Cref{lem:fast_greedy} above requires integer edge
weights, but we may obtain this by simply rounding up the weight of each edge
losing at most a factor of $2$ in the stretch. Alternatively we can use the
approach of \Cref{lem:quad_greedy} in \Cref{sec:quad_greedy} to
reduce this factor of $2$ to $(1+\eps)$.

Our $A_2$ will be obtained by the following lemma, which is essentially a
modified implementation of \Cref{lem:fast_greedy}.
\begin{lemma}\label{lem:fast_greedy_notlight}
    Let $G = (V,E,w)$ be an edge-weighted graph with $m = |E|$ and $n = |V|$. Let $k$ be a positive integer and let $\eps'>0$ be a constant. Then one can deterministically construct an $O(k)$-spanner of
    $G$ with size $O(n^{1+1/k})$ in time $O\!\left(m +
    kn^{1+1/k+\eps'}\right)$.
\end{lemma}

Combining \Cref{lem:fw} of \Cref{sec:fw} with
Lemmas~\ref{lem:fast_greedy} and~\ref{lem:fast_greedy_notlight} above
immediately gives us a spanner with stretch $O(k)$, size
$O(n^{1+1/k})$ and lightness $O(n^{1/k})$ in time $O(m+n^{1+1/k+\eps''})$ for
any constant $\eps''>0$. This is true because we may assume that $k\le \gamma\log n$ for
any constant $\gamma > 0$ (as the improvement in sparsity and lightness obtained by picking $k>\gamma\log n$ is bounded by $2^\gamma$), and thus by picking $\gamma$ and $\eps'$ accordingly we have
that the running time given by \Cref{lem:fw} can be bounded by
\[
    O\!\left(m+kWn^{1+1/k+\eps'} + kn^{1+1/k+\eps'}\right) =
    O\!\left(m+kg^kn^{1+1/k+\eps'}\right) =
    O\!\left(m+n^{1+1/k+\eps''}\right)\ .
\]

\subsection{Details of the almost-greedy spanner}
Set $\eps=1$ \footnote{In \Cref{sec:quad_greedy} we let $0 < \eps < 1$ here to be arbitrary small parameter.}.
Our algorithm for \Cref{lem:fast_greedy} is described below in
\Cref{alg:fast_greedy}. 
It
computes a spanner of stretch $c_1(1+\eps)(2k-1)$, where $c_1 = O(1)$ 
is the
stretch of our incremental approximate distance oracle in
\Cref{thm:dist_oracle}. Let $t = c_1(1+\eps)(2k-1)$ throughout the
section.
\begin{algorithm}
    \caption{\FuncSty{Approximate-Greedy}}
    \label{alg:fast_greedy}
    \DontPrintSemicolon
    \SetKwInOut{Input}{input}\SetKwInOut{Output}{output}
    \Input{Graph $G = (V,E,w)$, Parameters $\eps,k$}
    \Output{Spanner $H$}
    \BlankLine
    Create $H = (V,\emptyset)$\;
    Initialize incremental distance oracle (\Cref{thm:dist_oracle}) on
    $H$ with $d = t\cdot W$\;
    \For{$(u,v)\in E$ in non-decreasing order}{
        \If{$\hat{d}_H(u,v) > t\cdot w(u,v)$}{
            Add $(u,v)$ to $H$\;
        }
    }
    \Return $H$\;
\end{algorithm}
With \Cref{alg:fast_greedy} defined we are now ready to prove
\Cref{lem:fast_greedy}.
\begin{proof}[Proof of \Cref{lem:fast_greedy}]
    Let $H$ be the spanner created by running
    \Cref{alg:fast_greedy} on the input graph $G$ with the input
    parameters.

    \textbf{Stretch:}
    We will bound the stretch by showing that for any edge $(u,v)\in E$ there
    is a path of length at most $t\cdot w(u,v)$ in $H$.
    Let $(u,v)$ be any edge considered in the for loop of
    \Cref{alg:fast_greedy}. If $(u,v)$ was added to $H$ we are done.
    Thus, assume that $(u,v)\notin H$. In this case we have
    $\hat{d}_H(u,v) \le t\cdot w(u,v)$ as $(u,v)$ would have been
    otherwise added to $H$. The lemma now follows by noting that $d(u,v)
    \le \hat{d}_H(u,v)$ by \Cref{thm:dist_oracle}.

    \textbf{Size and lightness:}
    Next we bound the size and lightness of $H$. Our proof is very similar to
    the proof of Filtser and Solomon for the greedy spanner~\cite{FiltserS16}.
    However, we need to be careful as we are using an approximate distance
    oracle and do not have the exact distances when inserting an edge.
    Let $H'$ be \emph{any} spanner of $H$ with stretch $(1+\eps)(2k-1)$. We
    will argue that $H' = H$.
    To see this let $(u,v)\in H\setminus H'$ be any edge contradicting the
    above statement. Then there must be a path $P$ in $H'$ connecting $u$ and
    $v$ with $w(P) \le (1+\eps)(2k-1)\cdot w(u,v)$. Let $(x,y)$ be the last
    edge in $P\cup \{(u,v)\}$ examined by \Cref{alg:fast_greedy}. It
    follows that $w(x,y)\ge w(u,v)$. As $P\cup \{(u,v)\}\in H$ it follows that
    all the edges of $(P\cup \{(u,v)\})\setminus (x,y)$ were already in $H$ 
    when $(x,y)$ was added. These edges form a path in $H$ connecting $x$ and
    $y$ of weight
    \[
        w(P) - w(x,y) + w(u,v) \le w(P) \le (1+\eps)(2k-1)\cdot w(u,v) \le
        (1+\eps)(2k-1)\cdot w(x,y)\ .
    \]
    It follows that $d(x,y) \le (1+\eps)(2k-1)\cdot w(x,y)\le d$ just
    before $(x,y)$ was added to $H$, and by \Cref{thm:dist_oracle} that
    $\hat{d}_H(x,y) \le t\cdot w(x,y)$. Thus
    \Cref{alg:fast_greedy} did not add the edge $(x,y)$ to $H$, which
    is a contradiction. We conclude that $H' = H$.

    Now, since $H'$ could be any spanner of $H$, we may in particular choose it
    to be the $(1+\eps)(2k-1)$ spanner from \Cref{lem:light_spanner}.
    It now follows immediately that $H=H'$ has size $O(n^{1+1/k})$. For the
    lightness we know that $H'$ has lightness $O(n^{1/k})$ with regard to the
    MST of $H$. Thus, if we can show that the MST of $H$ is the same as the MST
    of $G$ we are done. However, this follows by noting that
    \Cref{alg:fast_greedy} adds exactly the MST of $G$ to $H$ that
    would have been added by Kruskal's algorithm \cite{Kruskal56}, since each
    such edge connects two disconnected components. Thus the MST of $G$ and $H$
    have the same weight which completes the proof.

    \textbf{Running time:}
    In \Cref{alg:fast_greedy} we perform $m$ queries to the
    incremental distance oracle of \Cref{thm:dist_oracle} each of which
    take $O(1)$ time. We also perform $|E(H)|$ insertions to the incremental
    distance oracle. We invoke \Cref{thm:dist_oracle} using $\eps^*$
    picked such that $1/\eps^*$ is integer and $\eps^* + \eps^*/k  \leq \eps'$.
    Since $d = O(kW)$, it follows from \Cref{thm:dist_oracle} and the
    size bound above that running time of the for-loop of
    \Cref{alg:fast_greedy} is
    \[
        O(m+d|E(H)|^{1 + \eps^*}) = O(m+kWn^{1+1/k +\eps^* + \eps^*/k}) = O(m+kWn^{1+1/k +\eps'})\ .
    \]
    To achieve the non-decreasing order we may simply run the algorithm of
    Baswana and Sen~\cite{BaswanaS07} first with parameter $k = 1/\eps'$. This
    gives an additional factor of $O(1/\eps')$ to the stretch, but leaves us with a graph with only $O(n^{1+\eps'})$ edges which we may then sort.
\end{proof}

Next, we sketch the proof \Cref{lem:fast_greedy_notlight}, by explaining
how to modify the proof of \Cref{lem:fast_greedy}.
\begin{proof}[Proof of \Cref{lem:fast_greedy_notlight}]
    Recall that $c_1$ is defined as the constant stretch provided by
    \Cref{thm:dist_oracle}.
    We use \Cref{alg:fast_greedy} with the following modifications: (1)
    we pick $d = c_1(2k-1)$,~ (2) when adding an edge to
    the distance oracle we add it as an unweighted edge, ~(3) we add an edge if
    its endpoints are not already connected by a path of at most $d$ edges
    according to the approximate distance oracle.

    The stretch of the spanner follows
    by the same stretch argument as in \Cref{lem:fast_greedy} and the fact that we consider
    the edges in non-decreasing order. To see that the size of the spanner is
    $O(n^{1+1/k})$ consider an edge $(u,v)$ added to $H$ by the modified
    algorithm. Since $(u,v)$ was added to $H$ we know that the distance
    estimate was at least $c_1(2k-1)$. It thus follows from
    \Cref{thm:dist_oracle} that $u$ and $v$ have distance at least $2k$
    in $H$ and therefore $H$ has girth at least $2k+1$.
    It now follows that $H$ has $O(n^{1+1/k})$ edges by a standard
    argument. The running time of this modified algorithm follows directly
    from \Cref{thm:dist_oracle}.
\end{proof}

\subsection{Near-quadratic time implementation}\label{sec:quad_greedy}
The construction of the previous section used our result from
\Cref{thm:dist_oracle} to efficiently construct a spanner losing a
constant factor exponential in $1/\eps$ in the stretch. We may instead
use the seminal result of Even and Shiloach~\cite{EvenS81} to obtain the same
result with stretch $(1+\eps)(2k-1)$ at the cost of a slower running time as
detailed in \Cref{thm:slow_good}. It is well-known that the decremental data structure in~\cite{EvenS81} can be made to work with the same time guarantees in the incremental setting; we will make use of this result:
\begin{theorem}[\cite{EvenS81}]\label{thm:es-tree}
    There exists a deterministic incremental APSP data structure for graphs
    with integer edge weights, which answers distance queries within a given
    threshold $d$ in $O(1)$ time and has total update time $O(mnd)$.
\end{theorem}
Here, the threshold means that if the distance between two nodes is at most
$d$, the data structure outputs the exact distance and otherwise it outputs
$\infty$ (or some other upper bound).

To obtain \Cref{thm:slow_good} we use the framework of
\Cref{sec:fw}. For the algorithm $A_2$ we may simply use the
deterministic spanner construction of Roditty and Zwick~\cite{RodittyZ11}
giving stretch $2k-1$ and size $O(n^{1+1/k})$ in time $O(kn^{2+1/k})$. For
$A_1$ we will show the following lemma.
\begin{lemma}\label{lem:quad_greedy}
    Let $G = (V,E,w)$ be an undirected graph with $m = |E|$ and $n =
    |V|$, edge weights bounded from above by $W$ and where all MST
    edges have weight $1$. Let $k$ be a positive integer. Then one can deterministically construct a $(1+\eps)(2k-1)$-spanner of $G$ with size $\Oeps(n^{1+1/k})$ and lightness
    $\Oeps(n^{1/k})$ in time $\Oeps(m\log n + kWn^{2+1/k})$.
\end{lemma}
\begin{proof}[Proof sketch]
    The final spanner will be a union of two spanners. Since
    \Cref{thm:es-tree} requires integer weights. We therefore need to treat
    edges with weight less than $1/\eps$ separately. For these edges we use
    the algorithm of Roditty and Zwick~\cite{RodittyZ11} to produce a spanner
    with stretch $2k-1$, size $O(n^{1+1/k})$ and thus total weight
    $O(n^{1+1/k}/\eps)$.
 
    For the remaining edges with weight at least $1/\eps$ we now round up the
    weight to the nearest integer incurring a stretch
    of at most a factor of $1+\eps$. We now follow the approach of
    \Cref{alg:fast_greedy} using the incremental APSP data structure
    of \Cref{thm:es-tree} and a threshold in line 4 of
    $(1+\eps)(2k-1)\cdot w(u,v)$ instead. We use the distance threshold $d
    = (1+\eps)(2k-1)\cdot W$.

    The final spanner, $H$, is the union of the two spanners above. The
    stretch, size and lightness of the spanner follows immediately from the
    proof of \Cref{lem:fast_greedy}. For the running time, we add in the
    additional time to sort the edges and query the distances to obtain a total
    running time of
    \[
        \Oeps\!\left(m\log n + d\cdot |E(H)|\cdot |V(H)|\right)
        = \Oeps\!\left(m\log n + kWn^{2+1/k}\right)\ .\qedhere
    \]
\end{proof}
Now, recall that $W = g^k$, where $k\le \log n$ and $g>1$ is a fixed parameter
of our choice. By picking $g$ such that $g^{2k}\le n^{\eps'}$  we get
a running time of $O(n^{2+1/k+\eps'})$ for $A_1$.
\Cref{thm:slow_good} now follows from \Cref{lem:fw}.

\section{Almost Linear Spanner}\label{sec:fastCW}
Our algorithm builds on the spanner of Chechik and
Wulff-Nilsen~\cite{ChechikW16}. Here we first describe their algorithm and then
present the modifications. Chechik and Wulff-Nilsen implicitly used our
general framework, and thus provide two different algorithms $A^{\mbox{{\tiny
CW}}}_1$ and $A^{\mbox{{\tiny CW}}}_2$.
 $A^{\mbox{{\tiny CW}}}_2$ is simply the greedy
spanner algorithm.

$A^{\mbox{{\tiny CW}}}_1$ starts by partitioning the non-MST edges into $k$
buckets, such that the $i$th bucket contains all edges with weight in
$[g^{i-1},g^i)$. The algorithm is then split into $k$ levels with the
$i$th bucket being treated in the $i$th level. In the $i$th level, the
vertices are partitioned into $i$-clusters, where the $i$-clusters refine the
$(i-1)$-clusters. Each $i$-cluster has diameter $O(kg^i)$ and contains at least
$\Omega(kg^i)$ vertices. This is similar to the $(i,\eps)$-clusters in
\Cref{sec:fw} with the modification of having two types of clusters,
\emph{heavy} and \emph{light}. A cluster is heavy if it has many incident
$i$-level edges and light otherwise. For a light cluster, we add all the
incident $i$-level edges to the spanner directly. For the heavy clusters,
Chechik and Wulff-Nilsen~\cite{ChechikW16} create a special auxiliary
cluster graph and run the greedy spanner on this to decide which edges should
be added.

To bound the lightness of the constructed spanner, they show that each time a
heavy cluster is constructed the number of clusters in the next level is
reduced significantly. Then, using a clever potential function, they show that
the contribution of all the greedy spanners is bounded.
It is interesting to note, that in order to bound the weight of a single greedy spanner,
they use the analysis of \cite{ElkinNS14}. Implicitly, \cite{ElkinNS14} showed
that on graphs with $O(\poly(k))$ aspect ratio, the greedy
$(1+\eps)(2k-1)$-spanner has $O_\eps(n^{1/k})$ lightness and $O(n^{1+1/k})$
edges.

There are three time-consuming parts in \cite{ChechikW16}:
1) The clustering procedure iteratively grows the $i$-clusters as the union of
several $(i-1)$-clusters, but uses expensive exact diameter calculations in the
original graph.
2) They employ the greedy spanner several times as a subroutine during $A^{\mbox{{\tiny CW}}}_1$ for graphs with $O(poly(k))$ aspect ratio.
3) They use the greedy spanner as $A^{\mbox{{\tiny CW}}}_2$.

In order to handle 1) above we will grow clusters purely based on the
number of nodes in the $(i-1)$-clusters (in similar manner to $(i,\eps)$-clusters), thus making the clustering much more efficient without
losing anything significant in the analysis. To handle 2) We will use
the following lemma in place of the greedy spanner.
\begin{lemma}
	\label{lem:fast_spanner_aspect}
    Given a weighted undirected graph $G = (V,E,w)$ with $m$ edges and $n$
    vertices, a positive integer $k$, 
    $\epsilon > 0$, such that all the weights are within $[a,a\cdot\Delta)$, and the MST have weight $O(na)$. One can deterministically construct a
    $(2k-1)(1+\eps)$-spanner of $G$ with $O_{\epsilon}(n^{1+\frac{1}{k}})$ edges and
    lightness $O_{\epsilon}\left( n^{\frac{1}{k}}\cdot\log\left(\Delta\right)\right)$ in time $O\left(m+n\log
    n)\right)$.
\end{lemma}

The core of \Cref{lem:fast_spanner_aspect} already appears in
\cite{ElkinS16}, while here we analyze it for the special case where the aspect
ratio is bounded by $\Delta$. The main ingredient is an efficient spanner
construction by Halperin and Zwick~\cite{HZ96} for unweighted graphs (\Cref{thm:HZ96}).
The description of the algorithm of  \Cref{lem:fast_spanner_aspect} and its analysis can be found in
\Cref{sec:app_proof_fast_aspect}. Replacing the greedy spanner by
\Cref{lem:fast_spanner_aspect} above is the sole reason for the additional
$\log k$ factor in the lightness of \Cref{thm:fastCW}.

Imitating the analysis of \cite{ChechikW16} with the modified ingredients, we
are able to prove the following lemma, which we will use as $A_1$ in our
framework.
\begin{lemma}\label{lem:mainFCW}
    Given a weighted undirected graph $G = (V,E,w)$ with $m$ edges and $n$
    vertices, a positive integer $k\ge 640$, and
    $\epsilon > 0$, such that all MST edges
    have unit weight, and all weights bounded by $g^k$,    
    one can deterministically construct a
    $(2k-1)(1+\eps)$-spanner of $G$ with $O_\eps(n^{1+1/k})$ edges and
    lightness $O_{\epsilon}\left(\log k\cdot n^{\frac{1}{k}}\right)$ in time $O\left(m+nk\right)$.
\end{lemma}

To address the third time-consuming part we instead use the algorithm of \Cref{thm:SparseNoLight} as $A_2$.
 Replacing the greedy algorithm by
\Cref{thm:SparseNoLight} is the sole reason for the additional
$\log k$ factor in the sparsity of \Cref{thm:fastCW}.

Combining \Cref{lem:mainFCW}, \Cref{thm:SparseNoLight} and
\Cref{lem:fw} we get \Cref{thm:fastCW}. The remainder of this
section is concerned with proving \Cref{lem:mainFCW}.

\subsection{Details of the construction}\label{sec:description_high_level}
\Cref{alg:fastCW} below contains a high-level description of the
algorithm. We defer part of the exact implementation details and the analysis
of the running time to \Cref{sec:running_time}.
We denote $E_i = \left\{(u,v)\in E\mid w(u,v)\in [g^i,g^{i+1})\right\}$.

\begin{algorithm}
    \caption{\FuncSty{$A_1$ component of \Cref{thm:fastCW}}}\label{alg:fastCW}
    \DontPrintSemicolon
	\SetKwInOut{Input}{input}\SetKwInOut{Output}{output}
	\Input{Parameters $k,\eps$, weighted graph $G=(V,E,w)$ where all MST edges have unit weight and $\max_{e\in E}w(e)\le g^k$.}
	\Output{Spanner $E_{sp}$.}
    \BlankLine
	Fix $g=20$, $c=24$, $d=160$
	and $\mu=\log_g(k/\eps)$\;
    $E_{sp}\leftarrow MST(G)$\;
	\tcc{First phase:}
	Partition $V$ into $0$-clusters $\mathcal{C}_0$ such that for every $C\in
    \mathcal{C}_0$, $|C|\in[\frac{k}{c},\frac{k}{2}]$\;
	\For{$i=1$ to $k-1$}{
		Let $K_i$ be $G$ with each $C\in \mathcal{C}_{i-1}$ contracted. Retain
        only edges of weight $[g^{i},g^{i+1})$ (keeping $K_i$ simple).\;
        \tcc{Construct $i$-level \textbf{heavy clusters}}
        Let all nodes of $K_i$ be unmarked\;
        \For{$\vpi\in K_i$}{
            \If{$\deg(\vpi) \ge d$, $\vpi$ is unmarked, and all of $\vpi$'s
            neighbours are unmarked}{
                Create new heavy cluster $\hat{\vpi}$ with $\vpi$ and all
                    neighbours\;
                Mark all nodes of $\hat{\vpi}$\;
            }
        }
        \For{$\vpi \in K_i$}{
            \If{$\deg(\vpi) \ge d$ and $\vpi$ is unmarked}{
                \tcc{$\vpi$ must then have marked neighbour.}
                Add $\vpi$ to the heavy cluster of a marked neighbour\;      
            }
        }
        Mark all clustered, unmarked vertices $\vpi$\;
        Add all edges used to join heavy clusters to $E_{sp}$\;
        \tcc{Construct $i$-level \textbf{light clusters}}
        Add all edges incident to unmarked nodes to $E_{sp}$\;
        Join remaining nodes into clusters of size (number of original
        vertices) $\ge \frac{1}{c}\cdot kg^i$ and diameter $\le
        \frac{1}{2}\cdot kg^i$ using MST edges\;
        If a cluster cannot reach $\frac{1}{c}\cdot kg^i$ nodes. Add it to a
        neighbouring heavy cluster (via MST edge)\;
	}
    \tcc{Second phase:}
	Let $S_0$ be a subgraph of $G$ which contain only edges of
    weight at most $k/\eps$. Let $H_0$ be a $(2k-1)(1+\epsilon)$-spanner of
    $S_0$ constructed using \Cref{lem:fast_spanner_aspect}\label{Line:G0} \;
    Add $H_0$ to $E_{sp}$\;
	\For{$r=1$ to $\left\lceil k/\mu\right\rceil-1$}{
		\label{Line:Vr} Let $V_r$ to be the set of nodes obtained by
        contracting each $(r-1)\mu$ cluster contained in some $i$-level heavy
        cluster for $i\in [r\mu,(r+1)\mu)$ (deleting all the other
        $(r-1)\mu$-clusters)\;
		Let $\mathcal{E}_r$ be all the edges used to create $i$-clusters (heavy
        or light) for $i\in \left((r-1)\mu,(r+1)\mu\right]$\;
		\label{Line:Sr} Let $S_r$ be the graph with $V_r$ as its vertices and
        $\mathcal{E}_r\cup\bigcup_{i=(r-1)\mu}^{(r+1)\mu-1}E_i$ as its edges
        (keeping $S_r$ simple)\;
		Let $w_r(e)=\max\{w(e),kg^{(r-1)\mu}/\eps\}$ be the weight function of
        $S_r$\;
		\label{Line:Hr} Construct a $(2k-1)(1+\epsilon)$-spanner $H_r$ of $S_r$
        using \Cref{lem:fast_spanner_aspect}\;
		\label{Line:SrEsp} Add $H_r$ to $E_{sp}$\;
	}
	\Return $E_{sp}$\;
\end{algorithm}

Using our modified clustering we will need the following claim which is key
to the analysis. The claim is proved in \Cref{sec:running_time}. We
refer to the definitions from \Cref{alg:fastCW} in the following section.
\begin{claim}\label{clm:cluster_size_diam}
    For each $i$-level cluster $C\in\mathcal{C}_i$ produced by
    \Cref{alg:fastCW} it holds that:
	\begin{enumerate}
		\item $C$ has diameter at most $\frac{1}{2}kg^{i}$ (w.r.t to the
            current stage of the spanner $E_{sp}$).
		\item The number of vertices in $C$ is larger than its diameter and is
            at least $\frac{1}{c}kg^{i}$.
	\end{enumerate} 
\end{claim} 

Our analysis builds upon \cite{ChechikW16}.
The bound on the stretch of \Cref{lem:fast_spanner_aspect} follows as we
have only replaced the greedy spanner by alternative spanners with the same
stretch (and have similar guaranties on the clusters diameter). The proof
appears at \Cref{subapp:stretch}.

To bound the sparsity and lightness we consider the two phases of
\Cref{alg:fastCW}. During the $i$'th level of the first phase we add
at most $d$ edges per light cluster and at most $1$ edge per $(i-1)$-cluster to
form the heavy clusters. By \Cref{clm:cluster_size_diam}
each $i$-level cluster contains $\Omega(kg^i)$ vertices and thus the total
number of clusters over all levels is bounded by
$\sum_{i=0}^{k}O(\frac{n}{kg^i})=O(n/k)$. It follows that we add at most $O(n)$
edges during the first phase. For the lightness of these edges, note that edges
added during the $i$th level have weight at most $g^{i+1}$. Hence the total
weight added during the $i$th level is at most $O(\frac{n}{kg^{i-1}}\cdot
g^{i+1})$ for heavy clusters and at most $O(\frac{dn}{kg^i}\cdot g^{i+1})$ for
light clusters. Summing over all $k$ levels this contributes at most $O(n)$ to
the total weight from the first phase.

Next consider the second phase. 
First note $S_0$ has an MST of weight $n-1$ and only contains edges with weight
in $[1,\frac{k}{\eps})$. Thus, by \Cref{lem:fast_spanner_aspect}, 
$\left|H_{0}\right|=O_{\epsilon}(n^{1+\frac{1}{k}})$ and
$w(H_{0})=O_{\epsilon}\left(n^{\frac{1}{k}}\cdot\log\left(\frac{k}{\epsilon}\right)\right)=O_{\epsilon}\left(n^{\frac{1}{k}}\cdot\log
k\right)$.

Fix some $r\in\left[1,\left\lceil
k/\mu\right\rceil -1\right]$. Recall the definitions of $V_r$, $S_r$, and $H_r$:
$V_r$ is a set of vertices representing a subset of the $(r-1)\mu$-level
clusters. $S_r$ is a graph with nodes $V_r$ where all the edges have weight in
$[kg^{(r-1)\mu}/\eps,g^{(r+1)\mu}]=[g^{r\mu},g^{r\mu}k/\eps]$.
$H_r$ is a spanner of $S_r$ constructed using
\Cref{lem:fast_spanner_aspect}. 
Denote by $M_r$ the MST of $S_r$. The following lemma bound its weight. A proof can be found
in \Cref{subapp:lightness}.
\begin{lemma}
	\label{lem:M_r bound}
	The MSF $M_r$ of $S_r$ has weight $w_r(M_r)=O(\left|V_{r}\right|\cdot
	kg^{(r-1)\mu}/\eps)$.
\end{lemma}
By \Cref{lem:fast_spanner_aspect}, $|S_{r}|=\ensuremath{O_{\epsilon}(|V_{r}|^{1+\frac{1}{k}})}$.
Summing over all the indices $r$, we can
bound the number of edges added in second phase by
\begin{align*}
\sum_{r=0}^{\left\lceil k/\mu\right\rceil -1}\left|H_{i}\right| & =O_{\epsilon}(n^{1+\frac{1}{k}})+\sum_{r=1}^{\left\lceil k/\mu\right\rceil -1}O_{\epsilon}\left(|V_{r}|^{1+\frac{1}{k}}\right)\\
& =O_{\epsilon}\left(n^{1+\frac{1}{k}}+\sum_{r=1}^{\left\lceil k/\mu\right\rceil -1}\left(\frac{n}{kg^{(r-1)\mu}}\right)^{1+1/k}\right)=O_{\epsilon}\left(n^{1+\frac{1}{k}}\sum_{r=0}^{\infty}\frac{1}{g^{r}}\right)=O_{\epsilon}\left(n^{1+\frac{1}{k}}\right)~.
\end{align*}
Using a potential function, we show that the sum of the weights $\sum_r w(H_r)$ converges nicely. The details can be found in \Cref{subapp:lightness}.
\begin{lemma}
	\label{lem:total_weight_spanners}
	The total weight of the spanners constructed in the second phase of
    \Cref{alg:fastCW} is $O_{\epsilon}\left(n^{1+\frac{1}{k}}\cdot\log
    k\right)$.
\end{lemma}
The size and lightness of \Cref{lem:mainFCW} now follows. All that
is left is to describe the exact implementation details and analyze the running
time, which is done below.

\subsection{Exact implementation of \Cref{alg:fastCW}}
\label{sec:running_time}
In this section we give a detailed description of \Cref{alg:fastCW}
and bound its running time. In addition we prove
\Cref{clm:cluster_size_diam}.

\paragraph*{First phase}
Let $m_i=|E_i|$ and $n_i=|\mathcal{C}_i|$ be number of $i$-level clusters as
described in \Cref{alg:fastCW}.
For each $i$, the clusters $\mathcal{C}_i$ form a partition of
$V$, where $\mathcal{C}_i$ is a refinement of $\mathcal{C}_{i+1}$. To
efficiently facilitate certain operations we will maintain a forest
$\mathcal{T}$ representing the hierarchy of containment between the clusters in
different levels. Specifically, $\mathcal{T}$ will have levels going from $-1$
to $k$. For simplicity, we treat each vertex $v\in V$ as a $-1$-cluster. Each
$i$-cluster $\vpi$ will be represented by an $i$-level node $v_\vpi$. $v_\vpi$
will have a unique out-going edge to $v_{\vpi'}$, the $(i+1)$-level node that
represents the $(i+1)$-cluster $\vpi'$ containing $\vpi$. In addition, each
node $v_\vpi$ in $\mathcal{T}$ will store the size of the cluster it
represents. Further, every $-1$ level node $v$ in $\mathcal{T}$ will have a
link to each of its ancestors in $\mathcal{T}$ (i.e. nodes representing the
$i\ge 0$ clusters containing $v$).

\textbf{$0$-clusters}: are constructed upon $-1$-clusters ($V$). The
construction of $\mathcal{C}_0$ is done the same way as for $i\ge 1$ (see
below), where we start the construction right away from the construction of
light clusters.

\textbf{$i$-clusters}:
Fix some $i\in[1,k-1]$. We assume that $\mathcal{T}$ is updated.
Construct a graph $K_i$ with $\mathcal{C}_{i-1}$ as its vertices. We add all
the edges of $E_i$ to $K_i$ (deleting self-loops and keeping only the
lightest edge between two clusters). The construction of $K_i$ is finished in
$O(m_i)$ time (using $\mathcal{T}$).

The construction of $\mathcal{C}_i$ is done from $K_i$ in two parts.
In the first part we construct the \emph{heavy clusters}. In the beginning
all the nodes are unmarked. We now go over all the nodes, $\vpi$, in $K_i$
and consider the following cases: If $\vpi$ has at least $d$ neighbors
and both $\vpi$ and all its neighbors are unmarked we create a new $i$-level
heavy cluster $\tilde{\vpi}$ containing $\vpi$ and all of its neighbors.
We mark all the nodes currently in $\tilde{\vpi}$, called the \emph{origin} of
$\tilde{\vpi}$ (additional clusters might be added later). In addition, we add all the (representatives of the) edges
between $\vpi$ and its neighbors to $E_{sp}$.
At the end of this procedure, each unmarked node $\vpi$ with at least $d$
neighbors has at least one marked neighbor. We add each such $\vpi$ to a
neighboring $i$-level cluster (via and edge to its origin) and mark $\vpi$. We
also add the corresponding edge to $E_{sp}$. For every heavy cluster $\vpi$
created so far, we denote all the vertices currently in $\vpi$ as the
\emph{core} of $\vpi$ (additional clusters might be added later during the
formation of light clusters).

In the second part we construct the \emph{light clusters}. We start by adding
all the (representatives of the) edges incident to the remaining unmarked nodes
to $E_{sp}$.
Let $L_i$ be the graph with the remaining unmarked nodes as its
vertex set and the edges of the MST
going between these nodes (keeping the
graph simple) as its edge set.
The clustering is similar to the $(i,\eps)$ clustering described in \Cref{sec:fw}.
Iteratively, we pick an arbitrary node $\vpi\in L_i$, and grow a cluster around
it by joining arbitrary neighbors one at a time. Once the cluster has size
at least $kg^i/c$ (number of actual vertices from $G$) we stop and make it
an $i$-level light cluster
$\tilde{\vpi}$. We call the nodes currently in $\tilde{\vpi}$ the \emph{core}
of $\tilde{\vpi}$. If the cluster has size less than $kg^i/c$ and there is no
remaining neighboring vertices in $L_i$, we add it to an existing neighboring
cluster (heavy or light) via an MST edge to its core (note that this is always
possible). We continue doing
this until all nodes are part of an $i$-level cluster.

This finishes the description of the clustering procedure. We are now
ready to prove \Cref{clm:cluster_size_diam}.
\begin{proof}[Proof of \Cref{clm:cluster_size_diam}]
	Recall the value of our constants: $g=20$, $c=24$, $d=160$. 
	We also assumed that $k\ge 640$.
    
    We will prove the claim by induction on $i$. We start with $i=0$.
    Property (2) of \Cref{clm:cluster_size_diam} is straightforward from
    the construction as we used only unit weight edges. For property (1), note
    that the core of each $0$-cluster has diameter at most $\frac{k}{c}$. Each
    additional part has diameter at most $\frac{k}{c}-1$ and is connected via
    unit weight edge to the core. Hence the diameter of each $0$-cluster is
    bounded by $3\cdot\frac{k}{c}< \frac{k}{2}$.
	
	Now assume that the claim holds for $i-1$ and let $C\in\mathcal{C}_i$.
	Assume first that $C$ is a light cluster. From the construction, $C$
    contains at least $kg^i/c$ vertices. The size of $C$ is larger than the
    diameter by the induction hypothesis and the fact that we used only unit
    weight edges to join the light cluster. For the upper bound on the
    diameter, observe that the diameter of $C$ was at most $kg^i/c$ before the
    last $(i-1)$-cluster was added to the core of $C$. At this point we add
    the final $(i-1)$-cluster, which has diameter at most $kg^{i-1}/2$. We
    conclude that the diameter of the core of $C$ is at most
    $kg^{i}/c+kg^{i-1}/2$.
	Afterwards, we might add additional parts to $C$. However, each such part
    has diameter strictly smaller then $kg^{i}/c$ and are added with a unit
    weight edge to the core of $C$. Thus each light cluster $C$ has diameter at
    most $\frac{1}{c}kg^{i}+\frac{1}{2}kg^{i-1}+2\frac{1}{c}kg^{i}\le
    kg^{i}\cdot\left(\frac{3}{c}+\frac{1}{2g}\right)\le\frac{1}{2}kg^{i}$.

	Next, we consider a heavy cluster $C$. Let $\tilde{C}\subseteq C$ be the
    set of vertices that belonged to $C$ before the construction of light
    clusters (i.e. the core of $C$). Let $\vpi$ be the original $(i-1)$-cluster that formed $C$.
    Then each $(i-1)$-cluster of $\tilde{C}$ is at distance at most $2$ from
    $\vpi$ in $K_i$. Thus, by the induction hypothesis, the diameter of
    $\tilde{C}$ is at most $5\cdot\frac{1}{2}kg^{i-1}+4\cdot
    g^{i+1}=kg^{i}\cdot\left(\frac{5}{2g}+\frac{4g}{k}\right)\le
    kg^{i}/4$, 
	and its size is
    at least $d\cdot kg^{i-1}/c=kg^{i}\cdot\frac{d}{cg}= kg^{i}/3$. 
	During the construction of the light clusters we might add some
    ``semi-clusters'' to $\tilde{C}$ of diameter strictly smaller then $kg^{i-1}/c$ via unit
    weight edges. We conclude that the diameter of $C$ is at most
    $kg^{i}/4+2\cdot
    kg^{i}/c=kg^{i}\cdot\left(\frac{1}{4}+\frac{2}{c}\right)= kg^{i}/3$.
\end{proof}

To conclude the first phase we will analyze its running time. Level $i$
clustering is done in $O(n_{i-1}+m_i)$ time, and updating $\mathcal{T}$ takes an
additional $O(n)$ time. In total all the first phase takes us $O(kn+m)$
time.

\paragraph*{Second phase}
Recall that we pick $\mu = \log(k/\eps)$ and refer to \Cref{alg:fastCW} for definitions and details. Here we only analyze the running time. We denote
$\tilde{m}_r=\sum_{i=(r-1)\mu}^{(r+1)\mu-1}m_i$.

Creating $S_0$ (line \ref{Line:G0}) takes $O(m+n\log n)$ times, computing $H_0$ takes
$\Oeps(m+n\log n)$ time (according
to \Cref{lem:fast_spanner_aspect}).
Next we have $\frac{k}{\mu}$ step loop. For fixed $r$, we create the vertex set $V_r$ (line \ref{Line:Vr}) in  $O(n_{(r-1)\mu})$ time, using $\mathcal{T}$. \footnote{Just go
	from each $(r+1)\mu$-level cluster to all of its descendants and return each
	$(r-1)\mu$ cluster that had a heavy cluster as ancestor in the first $\mu$ steps.}
Upon $V_r$, we create the graph $S_r$ (line \ref{Line:Sr}).
This is done by first adding the edges of $\mathcal{E}_r$, and all the edges in $\cup_{i=(r-1)\mu}^{(r+1)\mu-1}E_i$. 
We can  maintain $\mathcal{E}_r$ during the first phase in no additional cost,
thus creating $S_r$ and modifying the weights will cost us $O(n_{(r-1)\mu}+\tilde{m}_r)$.
Finally, we compute a spanner $H_r$ of $S_r$ using \Cref{lem:fast_spanner_aspect} (line \ref{Line:Hr})
in $O_{\eps}\left(\tilde{m}_{r}+|V_{r}|\cdot\log|V_{r}|\right)=O_{\eps}\left(\tilde{m}_{r}+n_{(r-1)\mu}\log n\right)$ time.
Then we add (the representatives of) the edges in $H_r$ into $E_{sp}$ (line \ref{Line:SrEsp}) in $O_{\eps}\left(\tilde{m}_{r}+n_{(r-1)\mu}\right)$ time. 
Thus, the total time invested in creating $H_r$ is $O_{\eps}\left(\tilde{m}_{r}+n_{(r-1)\mu}\log n\right)$.
The total time is bounded by 
\begin{align*}
\sum_{r=0}^{\left\lceil k/\mu\right\rceil -1}\text{Time}(H_{r}) & =O_{\epsilon}(m+n)+O_{\eps}\left(\tilde{m}_{r}+n_{(r-1)\mu}\log n\right)\\
& =O_{\epsilon}\left(m+\sum_{r=1}^{\left\lceil k/\mu\right\rceil -1}\tilde{m}_{r}+\log n\cdot\sum_{r=1}^{\left\lceil k/\mu\right\rceil -1}\frac{n}{kg^{(r-1)\mu}}\right)=O_{\epsilon}\left(m+n\log n\right)~.
\end{align*}
\paragraph*{Running time}
Combing the first and second phases above, the total running time is  
$O(kn+m)+O\left(m+n\log n\right)=O\left(m+n\log n\right)$.

\section{Proof of \Cref{thm:SparseNoLight}}\label{sec:proofOfSparseOnly}
We restate the theorem for convenience:
\SparseNoLight*
The basic idea in the algorithm of \Cref{thm:SparseNoLight}, is to
partition the edges $E$ of $G$ into $O_\eps(\log k)$ sets $E_1,E_2,\dots$, such
that the edges in $E_i$ are ``well separated''. That is, for every $e,e'\in
E_i$, the ratio between $w(e)$ and $w(e')$ is either a constant or at least
$k$.
By hierarchical execution of a modified version of \cite{HZ96}, with
appropriate clustering, we show how to efficiently construct a spanner of
size $O(n^{1+1/k})$ for each such ``well separated'' graph.
Thus, taking the union of these spanners, \Cref{thm:SparseNoLight}
follows.

In \Cref{subseq:Alg} we describe the algorithm. In
\Cref{subseq:Stre} we bound the stretch, in
\Cref{subseq:Spar} the sparsity, and in \Cref{subseq:Time} the running time. In \Cref{subsec:unionFind} we introduce a relaxed version of the union/find problem (called \emph{prophet union/find}), and construct a data structure to solve it. The prophet union/find is used in the implementation of our algorithm.

\subsection{Algorithm}\label{subseq:Alg}
The following is our main building block. The description and the proof can be found in \Cref{app:ModHZ}.
\begin{restatable}[Modified \cite{HZ96}]{lemma}{ModHZ}
	\label{lem:ModHZ}
	Given an unweighted graph $G=(V,E)$ and a parameter $k$,
	\Cref{alg:Mod-HZ} returns a $(2k-1)$-spanner $H$ with
	$O(n^{1+1/k})$ edges in $O(m)$ time. Moreover, it holds that
	\begin{enumerate}
		\item $V$ is partitioned into sets $S_1,\dots,S_R$, such that at iteration $i$ of the loop, $S_i$ was deleted from $V'$.
		\item For every $i$, $\diam_H(S_i)\le 2k-2$.
		\item When deleting $S_i$, \Cref{alg:Mod-HZ} adds less then
		$|S_i|\cdot n^\frac{1}{k}$ edges. All these edges are either
		internal to $S_i$ or going from $S_i$ to $\cup_{j>i}S_j$.
		\item There is an index $t$, such that for every $i\le t$, $|S_i|\ge
		n^{1/k}$, and for every $i>t$, $|S_i|=1$ (called singletons).
	\end{enumerate} 
\end{restatable}

\sloppy For simplicity we assume that the minimal weight of an edge in $E$ is $1$.
Otherwise, we can scale accordingly. Let $c_{l}=O_{\epsilon}(1)$, such that
$(1+\epsilon)^{c_{l}\log k}\ge\frac{18k}{\epsilon}$. Let $E_{i}=\left\{ e\in
E\mid w(e_{i})\in\left[(1+\epsilon)^{i},(1+\epsilon)^{i+1}\right)\right\} $,
and let $G_{j}$ be the subgraph containing the edges $\cup_{i\ge0}E_{j+i\cdot
c_{l}\log k}$. Note that $G_0,\dots,G_{c_{l}\log k -1}$ partition the edges of
$G$. Next we build a different spanner $H_j$ for every $G_j$ and set the final
spanner to be $H=H_0\cup\ldots\cup H_{c_l\log k - 1}$.

Fix some $j$. Set the $0$-clusters to be the vertex set $V$. Similar to the
previous sections we will have $i$-clusters, which are constructed as the union
of $(i-1)$-clusters. Let $G_{j,i}$ be the unweighted graph with the
$i$-clusters as its vertex set and $E_{j+i\cdot c_{l}\log k}$ as its edges
(keeping the graph simple). Let $H_{j,i}$ be the $(2k-1)$-spanner of $G_{j,i}$
returned by the algorithm of \Cref{lem:ModHZ}. We add (the representatives)
of the edges in $H_{j,i}$ to $H_{j}$. Based on $H_{j,i}$ we create the
$(i+1)$-clusters as follows. Let $S_1,\dots,S_t,V'$ be the appropriate
partition of the vertex set, where $S_1,\dots,S_t$ are non-singletons, and all
the singletons are in $V'$. Each $S_a$ for $a\le t$ becomes a $(i+1)$-cluster.
Next, for each connected component $C$ in $G_{j,i}[V']$, we
divide $C$ into clusters of size at least $k$, and diameter at most $3k$ (in
the case where $|C| < k$ we let $C$ be an $(i+1)$-cluster). We then proceed
to the next iteration.

\subsection{Stretch}\label{subseq:Stre}
We start by bounding the diameter of the clusters.
\begin{claim}
	Fix $j$, for every $i$-cluster $\varphi$ of $G_{j,i}$,  $\diam_{H}(\varphi)\le\frac{1}{2}\cdot\epsilon\cdot(1+\epsilon)^{j+i\cdot c_{l}\cdot\log k}$
\end{claim}
\begin{proof}
	We show the claim by induction on $i$. For $i=0$, the diameter is
	$0$. For general $i$, in the unweighted graph $G_{j,i-1}$, we created
	clusters of diameter at most $2k-2$ for the non-singletons and $3k$ for the
    singletons. Thus the diameter of $\phi$ in $H$ is bounded
    by the sum of $3k$ edges in $E_{j + (i-1)\cdot c_l \log k}$, and $3k+1$
    diameters of $(i-1)$-clusters. By the induction hypothesis
	\begin{align*}
	\diam_{H}(\varphi) & \le3k\cdot(1+\epsilon)^{j+(i-1)\cdot c_{l}\log k+1}+\left(3k+1\right)\cdot\frac{1}{2}\cdot\epsilon\cdot(1+\epsilon)^{j+(i-1)\cdot c_{l}\cdot\log k}\\
	& \le3k\cdot(1+\epsilon)^{j+(i-1)\cdot c_{l}\log k}\left(1+\epsilon+\epsilon\right)\\
	& =\frac{3k\cdot\left(1+2\epsilon\right)}{(1+\epsilon)^{c_{l}\log k}}\cdot(1+\epsilon)^{j+i\cdot c_{l}\log k}\\
	& \le\frac{1}{2}\cdot\epsilon\cdot(1+\epsilon)^{j+i\cdot c_{l}\log k}~,
	\end{align*}
	where the last inequality follows as $(1+\epsilon)^{c_{l}\log k}\ge\frac{18k}{\epsilon}$.
\end{proof}
The rest of the proof follows by similar arguments as in \Cref{eq:cluster-Stretch}. See \Cref{fig:stretchCluster} for illustration.

\subsection{Sparsity}\label{subseq:Spar}

Again, we fix some $j\ge0$. We will bound $|H_j|$ by $O(n^{1+1/k})$ using a
potential function. For a graph $G'$ with $n_{G'}$ vertices, set potential
function $P(G')=2\cdot n_{G'}\cdot n^{1/k}$.
That is, we start with a graph $G_{j,0}$ with $n_{0}=n$ vertices
and potential $P(G_{j,0})=2\cdot n\cdot n^{1/k}$. In step $i$ we considered the
graph $G_{j,i}$. Let $m_i$ denote the number of edges added to $H_j$ in
this step.
We will prove that $P(G_{j,i})-P(G_{j,i+1})\ge m_{i}$ to conclude that
\[
|H_{j}|=\sum_{i\ge0}m_{i}\le\sum_{i\ge0}P(G_{j,i})-P(G_{j,i+1})=P(G_{j,0})=2\cdot n^{1+\frac{1}{k}}~.
\]
Let $S_1,\dots,S_R$ be the partition created by \Cref{lem:ModHZ}, where
$S_1,\dots,S_t$ are the non-singletons, and $V'=\cup_{r>t}S_r$ are the
singletons. Let  $C_{1},\dots,C_{R'}$  be the connected components in the
induced graph $G_{j,i}[V']$. We will look on the clustering procedure
iteratively, and evaluate the change in potential after each contraction.

Consider first the non-singletons. Fix some $r\le t$ and let $X_{r}$ be
the graph after we contract $S_{1},\dots,S_r$ (note that
$X_0=G_{j,i}$). For $r\ge 0$, let $\hat{m}_{r}$ be the number of edges
added to $H_{j,i}$ while creating $S_{r}$. Recall that $\hat{m}_{r}\le
|S_{r}|\cdot n^{\frac{1}{k}}$. Thus
\begin{align*}
P(X_{r-1})-P(X_{r})&=2\cdot\left|X_{r-1}\right|\cdot n^{1/k}-2\cdot\left|X_{r}\right|\cdot n^{1/k}\\&=2\cdot\left|X_{r-1}\right|\cdot n^{1/k}-2\cdot\left(\left|X_{r-1}\right|-(\left|S_{r}\right|-1)\right)\cdot n^{1/k}\\&=2\cdot\left(\left|S_{r}\right|-1\right)\cdot n^{1/k}\ge \hat{m}_{r}~,
\end{align*}
where the inequality follows as $S_r$ is not a singleton. 

Next we analyze the singletons. Consider some singleton $\{v\} = S_r$. Recall
that once the algorithm processed $S_r$ it only added edges to the spanner from
the connected component $C_{r'}$ of $G_{j,i}[V']$ containing $v$. Furthermore
it added at most $n^{1/k}$ such edges. Instead of analyzing the potential
change from deleting $S_r$, we will analyze the change from processing the
entire connected component $C_{r'}$. Denote by $\tilde{m}_{r'}$ the total
number of edges added to the spanner from $C_{r'}$.
It holds that $\tilde{m}_{r'}\le \left|C_{r'}\right|\cdot n^{\frac{1}{k}}$.
Let $Y_{r'}$ be the graph $G_{j,i}$ where we contract $S_{1},\dots,S_t$,
and all the clusters created from $C_{1},\dots,C_{r'}$ (note that $Y_0=X_{t}$
and $Y_{R'}=G_{j,i+1}$). Suppose $C_{r'}$ is divided into clusters
$A_{1},\dots,A_{z}$. Then we have
\begin{align*}
P(Y_{r'-1})-P(Y_{r'}) & =2\cdot\left|Y_{r'-1}\right|\cdot
    n^{1/k}-2\cdot\left|Y_{r'}\right|\cdot n^{1/k}=2\cdot\left(\left|C_{r'}\right|-z\right)\cdot n^{1/k}\,.
\end{align*}
We prove that $P(Y_{r'-1})-P(Y_{r'})\ge \tilde{m}_{r'}$ by case analysis:
\begin{itemize}
    \item $|C_{r'}|=1$. Then $z=1$, which implies $\tilde{m}_{r'}=0$. 
    \item $|C_{r'}|>1$ and $z=1$. Then $\tilde{m}_{r'}\le |C_{r'}|\cdot
        n^{1/k}\le 2\cdot \left(\left|C_{r'}\right|-z\right)\cdot n^{1/k}=
        P(Y_{r'-1})-P(Y_{r'})$.
	\item $|C_{r'}|>1$ and $z>1$. Necessarily for every $q$, $|A_{q}|\ge k\ge3$.
	Hence
	\begin{align*}
	\tilde{m}_{r'} & \le|C_{r'}|\cdot n^{1/k}=\sum_{q=1}^{z}|A_{q}|\cdot n^{1/k}<\sum_{q=1}^{z}2\cdot\left(|A_{q}|-1\right)\cdot n^{1/k}\\
	& =2\cdot\left(\left|C_{r'}\right|-z\right)\cdot n^{1/k}=
        P(Y_{r'-1})-P(Y_{r'})~.
	\end{align*}
\end{itemize}
Finally, 
\begin{align}
P(G_{i})-P(G_{i+1}) & =\sum_{r=0}^{t-1}\left[P(X_{r})-P(X_{r+1})\right]+\sum_{r=0}^{R'-1}\left[P(Y_{r})-P(Y_{r+1})\right]\nonumber\\ 
& \ge\sum_{r=0}^{t-1}\hat{m}_{r}+\sum_{r=0}^{R'-1}\tilde{m}_{r}=m_{i}~.\label{eq:m_rPotBound}
\end{align}
\subsection{Running Time}\label{subseq:Time}
We can assume that the number of edges $m$ is at least $n^{1+1/k}\log k$, as otherwise we can simply return the whole graph as the spanner. Assuming this, dividing the edges into the sets $E_0,E_1,\dots$, and creating the graphs $G_0,\dots,G_{c_l\cdot\log k -1}$ will take us $O(m+n\log k)=O(m)$ time (first create $\log k$ empty graphs, and then go over the edges, and add each edge to the appropriate graph).
Fix $j$, and set $m_j$ to be the number of edges in $G_j$. The creation of $H_{j,i}$, takes $O\left(\left|E_{j+i\cdot c_{l}\log k}\right|\right)$ time (\Cref{lem:ModHZ}) which summed over all $i$ is $O(m_{j})$. Clustering can be done while constructing  $H_{j,i}$ with a union/find data structure. Queries to this data structure are used to identify the clusters containing the endpoints of edges and union operations are used when forming clusters from sub-clusters. However, a union/find data structure will be too slow for our purpose since we seek linear time for almost all choices of $k$. In the next subsection, we present a variant of the union/find problem called \emph{prophet union/find}; solving this problem suffices in our setting. With the constant $s$ from~\Cref{thm:SparseNoLight}, we give a data structure for prophet union/find which for any fixed $j$ spends time $O(m_{j}s + n\log^{(s)}n) = O(m_{j} + n\log^{(s)}n)$ on all operations. Summed over all $j$, this is $O(m + n\log^{(s)}n\log k)$.

We may assume that $n\log^{(s)}n\log k > m$ since otherwise the time bound simplifies to linear. Since we also assumed $m > n^{1+1/k}\log k$, we have
\[
  n\log^{(s)}n\log k > n^{1+1/k}\log k\Leftrightarrow \log^{(s)}n > n^{1/k}\Leftrightarrow k > \log n/\log^{(s+1)}n.
\]
We conclude that the running time is linear if $k\leq \log n/\log^{(s+1)}n$. Now, assume $k > \log n/\log^{(s+1)}n$. Then $\log n < k\log^{(s+1)}n < k^2$, implying that $\log^{(s)}n = O(\log^{(s-1)}(k^2)) = O(\log^{(s-1)}k)$ and we get a time bound of $O(m + n\log^{(s-1)}k\log k)$, as desired.


\subsection{Prophet Union/Find}\label{subsec:unionFind}
Consider a ground set $A=\{x_1,\dots,X_n\}$ of $n$ elements, partitioned to clusters $\mathcal{C}$, initially consisting of all the singletons. We need to support two type of operations:
\emph{find} query, where we are given an element $x\in A$ and should return the cluster $C\in\mathcal{C}$ containing it, and \emph{union} operation, where we are given two elements $x,y\in A$ where $x\in C_x$, $y\in C_y$ ($C_x,C_y\in\mathcal{C}$), and where we should delete the clusters $C_x,C_y$ from $\mathcal{C}$ and add a new cluster $C_x\cup\ C_y$ to $\mathcal{C}$.
The problem described above is called \emph{Union/Find}.
Tarjan \cite{Tarjan75} constructed a data structure that processes $m$ union/find operations over a set of $n$ elements in $O(m\cdot\alpha(n))$ time, where $\alpha$ is the very slow growing inverse Ackermann function.	

A trivial solution to the union/find problem will obtain $O(m+n\log n)$ running time, which is superior to \cite{Tarjan75} for $m\ge n\log n$.
Indeed, one can simply store explicitly for each element the name of the current cluster containing it, and given a union operation for $x,y\in A$, where $x\in A,y\in B$ and w.l.o.g. $|A|\ge |B|$, one can simply change the membership of all the elements in $B$ to $A$. Each find operation will take constant time, while every vertex can update its cluster name at most $\lg n$ times (as each time the cluster name is updated, the cluster size is at least doubled). Thus in total, the running time is bounded by $O(m+n\log n)$.

We introduce a relaxed version of the union-find problem we call the \emph{Prophet union/find}. Here we are given a ground set $A=\{x_1,\dots,x_n\}$ of $n$ elements, and a series $q_1,q_2,\dots,q_m\in A^m$ of element queries known in advance.
Then we are asked these previously provided set of queries, with union operations intertwined between the find operations. While the union operations are unknown in advance, they are of a restricted form: a union operation arriving after the query $q_j$, must be of the form $\{q_{j-1},q_j\}$, that is a union of the clusters containing the two last find query elements.
For a parameter $s$, we solve the Prophet union/find problem in $O(m\cdot s+n\log^{(s)}n)$ time.
\begin{theorem}\label{thm:prophet}
	For any $s\le\log^*n$, a series of $m$ operations in the Prophet union/find problem over a ground set of $n$ elements can be performed in  $O(m\cdot s+n\log^{(s)}n)$ time.
\end{theorem}
\begin{proof}
	Set $\alpha_0=n$, $\alpha_{s}=1$, and for $i\in\{1,\dots,s-1\}$, set $\alpha_i=\log^{(i)}n$. 
	We will execute a modified version of the trivial algorithm described above.
	Specifically, at any point in time, we maintain a set $A$ of elements partitioned to clusters $\mathcal{C}$, where for each element we will store the name of the cluster it currently belongs to, and for each cluster we will store the number of elements it contains.
	Initially we are given a list $q_1,\dots,q_m$ of queries, which we will store as well. Further, for each element $x\in A$, we will store a linked list containing the indices of all the queries $q_j$ such that $q_j=x$. Note that this prepossessing step is done in $O(m)$ time.  
	
	Given a find query $q_j$, which is simply a name of an element, we will return in $O(1)$ time the stored cluster name.
	Given a union operation arriving after the $j$'th query, we know that it is between the clusters $\mathcal{C}_{j-1},\mathcal{C}_{j}$ containing the elements $q_{j-1}$ and $q_j$ accordingly. We find the clusters and their sizes in $O(1)$ time. Assume w.l.o.g. that $|\mathcal{C}_{j-1}|\le|\mathcal{C}_{j}|$. Let $t\in\{0,\dots,s\}$ such that $|\mathcal{C}_{j}|\in[\alpha_{t+1},\alpha_t)$.
	There are two cases:
	\begin{enumerate}
		\item If $|\mathcal{C}_{j-1}|+|\mathcal{C}_{j}|< \alpha_t$, we proceed as the trivial algorithm. Specifically, we go over the elements of $\mathcal{C}_{j-1}$, update their cluster to be $\mathcal{C}_{j}$, and update the size of the cluster $\mathcal{C}_{j}$ to be $|\mathcal{C}_{j-1}|+|\mathcal{C}_{j}|$.
		\item Else, we have $|\mathcal{C}_{j-1}|+|\mathcal{C}_{j}|\ge \alpha_t$; in this case, replace all the elements  $\mathcal{C}_{j-1}\cup\mathcal{C}_{j}$ in $A$ by a new element $y$. 
		Specifically, add a new element $y$ to $A$ that will belong to a singleton cluster $C_y=\{y\}$ the size of which will be updated set to be $|\mathcal{C}_{j-1}|+|\mathcal{C}_{j}|$. Then make a linked list of queries for $y$ by concatenating the linked lists of the elements in $\mathcal{C}_{j-1}\cup\mathcal{C}_{j}$. Finally, use the newly created linked list to go over all the find queries $q_j,\dots,q_m$, and replace every appearance of an element from $\mathcal{C}_{j-1}\cup\mathcal{C}_{j}$ with $y$. Note that the we now have a valid preprocessed instance of the prophet union/find problem.
	\end{enumerate}
	We finish with time analysis of the execution of the algorithm. Note that every find operation takes $O(1)$ time, as we explicitly store all the queries and their answers.
	There are two types of executions of the union operation above.
	Denote by $A_t$ all the artificial elements created during the execution of the algorithm such that the number of ground elements they are replacing is in $[\alpha_{t},\alpha_{t-1})$. Then $|A_t|\le \frac{n}{\alpha_{t}}$. Each element $y\in A_t$ actively participated (that is makes any changes) in at most $\log \alpha_{t-1}$ union operations of the first type. This is because each time this happens, the size of the cluster containing $y$ is (at least) doubled, and once it reaches the size of $\alpha_{t-1}$, a union operation of the second type will occur, and $y$ will be deleted. Note that processing $y$ in each such union operation takes only $O(1)$ time (updating the name of the cluster it belongs to, and updating the size of the cluster). Thus it total, the time consumed by all the union operations of the first type is bounded by
	\[
	O(1)\cdot\sum_{t=1}^{s}|A_{t}|\cdot\log\alpha_{t-1}=O(1)\cdot\left[n\cdot\log\alpha_{s-1}+\sum_{t=1}^{s-2}\frac{n}{\alpha_{t}}\cdot\alpha_{t}\right]=O(n\cdot s+n\cdot\log^{(s)}n)~.
	\]
	To bound the time consumed by union operations of the second type, note that each ground element $x\in A=A_{s}$, can go over at most $s$ such transitions (implicitly). For every query $q_j$ that initially was to $x$, we will pay $O(1)$ for each such transition (update the query and the linked list), and  thus $O(s)$ overall. We conclude that all the changes due to the second type union operations consume at most $O(m\cdot s)$ time. The theorem now follows.
\end{proof}

\section{Deterministic Incremental Distance Oracles for Small Distances}\label{sec:DistOracle}
In this section, we present a deterministic incremental approximate distance
oracle which can answer approximate distance queries between vertex pairs whose
actual distance is below some threshold parameter $d$. This oracle will give us
\Cref{thm:dist_oracle} and finish the proof of \Cref{thm:fast_light}. In fact,
we will show the following more general result. \Cref{thm:dist_oracle} follows
directly by setting $k = 1/\eps$ in the theorem below.
\begin{theorem}\label{Thm:DistOracle}
    Let $G = (V,E)$ be an $n$-vertex undirected graph that undergoes a series
    of edge insertions. Let $G$ have positive integer edge weights and set $E =
    \emptyset$ initially. Let $\varepsilon > 0$ and
    positive integers $k$ and $d$ be given. Then a deterministic approximate
    distance oracle for $G$ can be maintained under any sequence of operations
    consisting of edge insertions and approximate distance queries. Its total
    update time is $O_{\varepsilon}(m^{1+1/k}(3+\varepsilon)^{k-1}d(k+\log
    d)\log n)$ where $m$ is the total number of edge insertions; the value of
    $m$ does not need to be specified to the oracle in advance. Given a query
    vertex pair $(u,v)$, the oracle outputs in $O(k\log n)$ time an approximate
    distance $\tilde d(u,v)$ such that $\tilde d(u,v)\ge d(u,v)$ and such that
    if $d(u,v)\le d$ then $\tilde d(u,v)\le (2(3+\varepsilon)^{k-1}-1)d(u,v)$.
\end{theorem}
As discussed in \Cref{sec:overview}, a main advantage of our oracle is that,
unlike, e.g., the incremental oracle of Roditty and Zwick~\cite{RodittyZ12}, it
works against an adaptive adversary.
Hence, the sequence of edge insertions does not need to be fixed in advance and
we allow the answer to a distance query to affect the future sequence of
insertions. This is crucial for our application since the sequence of edges
inserted into our approximate greedy spanner depends on the answers to the
distance queries.

We assume in the following that $m\ge n$; if this is not the case, we simply
extend the sequence of updates with $n - m$ dummy updates.
We will present an oracle satisfying \Cref{Thm:DistOracle} except that
we require it to be given $m$ in advance. An oracle without this requirement
can be obtained from this as follows. Initially, an oracle is set up with $m =
n$. Whenever the number of edge insertions exceeds $m$, $m$ is doubled and a
new oracle with this new value of $m$ replaces the old oracle and the sequence
of edge insertions for the old oracle are applied to the new oracle. By a
geometric sums argument, the total update time for the final oracle dominates
the time for all the previous oracles. Hence, presenting an oracle that knows
$m$ in advance suffices to show the theorem.

Before describing our oracle, we need some definitions and notation. For an
edge-weighted tree $T$ rooted at a vertex $u$, let $d_T(v)$ denote the distance
from $u$ to $v$ in $T$, where $d_T(v) = \infty$ if $v\notin V(T)$. Let $r(T) =
\max_{v\in V(T)}d_T(v)$. Given a graph $H$ and $W\subseteq V(H)$, we let
$\deg_H(W) = \sum_{v\in W}\deg_H(v)$ and given a subgraph $S$ of $H$, we let
$\deg_H(S) = \deg_H(V(S))$. For a vertex $u$ in an edge-weighted graph $H$ and
a value $r\ge 0$, we let $B_H(u,r)$ denote the ball with center $u$ and radius
$r$ in $H$, i.e., $B_H(u,r) = \{v\in V(H)\vert d_H(u,v)\le r\}$. When $H$ is
clear from context, we simply write $B(u,r)$.

We use a superscript $(t)$ to denote a dynamic object (such as a graph or edge
set) or variable just after the $t$'th edge insertion where $t = 0$ refers to
the object prior to the first insertion and $t = m$ refers to the object after
the final insertion. For instance, we refer to $G$ just after the $t$'th update
as $G^{(t)}$.

In the following, let $\varepsilon$, $k$, and $d$ be the values and let $G =
(V,E)$ be the dynamic graph of \Cref{Thm:DistOracle}. For each
$i\in\{0,\ldots,k-1\}$, define $m_i = 2m^{(i+1)/k}$ and let $d_i$ be the
smallest integer power of $(1+\varepsilon)$ of value at least $(3+2\varepsilon)^id$.
For each $u\in V$ and each $t\in\{0,\ldots,m\}$, let $d_i^{(t)}(u)$ be the
largest integer power of $(1+\varepsilon)$ of value at most $d_i$ such that
$\deg_{G^{(t)}}(B^{(t)}(u,d_i^{(t)}(u)))\le m_i$; note that $d_i^{(t)}(u)\geq (1+\varepsilon)^{-1}$. We let $B_i^{(t)}(u) =
B^{(t)}(u,d_i^{(t)}(u))$ and let $T_i^{(t)}(u)$ be a shortest path tree from
$u$ in $B_i^{(t)}(u)$. Note that $T_i^{(t)}(u)$ need not be uniquely defined;
in the following, when we say that a tree is equal to $T_i^{(t)}$, it means
that the tree is equal to some shortest path tree from $u$ in $B_i^{(t)}(u)$.

The data structure in the following lemma will be used as black box in our distance oracle. One of its tasks is to efficiently maintain trees $T_i^{(t)}(u)$.
\begin{lemma}\label{Lem:MaintainIncTrees}
Let $U\subseteq V$ be a dynamic set with $U^{(0)} = \emptyset$ and let $i\in\{0,\ldots,k-1\}$ be given. There is a deterministic dynamic data structure which supports any sequence of update operations, each of which is one of the following types:
\begin{description}
    \item [$\FuncSty{Insert-Edge}(u,v)$:] this operation is applied whenever an edge $(u,v)$ is inserted into $E$,
    \item [$\FuncSty{Insert-Vertex}(u)$:] inserts vertex $u$ into $U$.
\end{description}
Let $t_{\max}$ denote the total number of operations and for each vertex $u$ inserted into $U$, let $t_u$ denote the update in which this happens. The data structure outputs in each update $t\in\{1,\ldots,t_{\max}\}$ a (possibly empty) set of trees $\overline T_i^{(t)}(u)$ rooted at $u$ for each $u\in U^{(t)}$ satisfying either $t > t_u$ and $d_i^{(t)}(u) < d_i^{(t-1)}(u)$ or $t = t_u$ and $d_i^{(t)}(u) < d_i$. For each such tree $\overline T_i^{(t)}(u)$, $r(\overline T_i^{(t)}(u))\le (1+\varepsilon)d_i^{(t)}(u)\le d_i$ and $\deg_{G^{(t)}}(\overline T_i^{(t)}(u)) > m_i$. Total update time is $O(m) + O_{\varepsilon}(|U^{(t_{\max})}|m_id_i\log n)$.

At any point, the data structure supports in $O(1)$ time a query for the value $d_i^{(t)}(u)$ and in $O(\log n)$ time a query for the value $d_{T_i^{(t)}(u)}(v)$ and for whether $v\in V(T_i^{(t)}(u))$, for any query vertices $u\in U$ and $v\in V$.
\end{lemma}
\begin{proof}
We assume in the following that each vertex of $V$ has been assigned a unique label from the set $\{0,\ldots,n-1\}$.

In the following, fix a vertex $u\in V$ such that $t_u$ exists, i.e., update
    $t_u$ is the operation $\FuncSty{Insert-Vertex}(u)$. Before proving the lemma, we describe a data structure $\mathcal D_u$ which maintains the following for each $t\in\{t_u,\ldots,m\}$: a tree $T^{(t)}(u)$ rooted at $u$, a distance threshold $d^{(t)}(u)$, and distances $d_{T^{(t)}(u)}(v)$ for all $v\in V(T^{(t)}(u))$. We will show that $\mathcal D_u$ maintains the following two properties:
\begin{enumerate}
\item $T^{(t)}(u) = T_i^{(t)}(u)$ and $d^{(t)}(u) = d_i^{(t)}(u)$ for all $t\in\{t_u,\ldots,m\}$,
\item in each update $t\in\{t_u,\ldots,m\}$ where either $t > t_u$ and $d_i^{(t)}(u) < d_i^{(t-1)}(u)$ or $t = t_u$ and $d_i^{(t)}(u) < d_i$, $\mathcal D_u$ outputs a tree $\overline T_i^{(t)}(u)$ rooted at $u$ such that $r(\overline T_i^{(t)}(u))\le (1+\varepsilon)d_i^{(t)}(u)\le d_i$ and $\deg_{G^{(t)}}(\overline T_i^{(t)}(u)) > m_i$. In all other updates, no tree is output.
\end{enumerate}

After any update $t$, $\mathcal D_u$ supports in $O(1)$ time a query for the value $d_i^{(t)}(u)$ and in $O(\log n)$ time a query for the value $d_{T_i^{(t)}(u)}(v)$ and for whether a given vertex $v\in V$ belongs to $V(T_i^{(t)}(u))$.

$\mathcal D_u$ maintains a tree $T^{(t)}(u)$ rooted at $u$ as well as a distance threshold $d^{(t)}(u)$; to simplify notation, we shall write $T^{(t)}$ instead of $T^{(t)}(u)$ and $d^{(t)}$ instead of $d^{(t)}(u)$. Later we show that $T^{(t)} = T_i^{(t)}(u)$ and $d^{(t)} = d_i^{(t)}(u)$. The tree $T^{(t)}$ is maintained by keeping a predecessor pointer for each vertex to its parent (with $u$ having a nil pointer) and where each vertex $v\in T^{(t)}$ is associated with its distance $d_{T^{(t)}}(v)$ from the root $u$.

Since $d^{(t)}$ is maintained explicitly by $\mathcal D_u$ and since $d^{(t)} = d_i^{(t)}(u)$, it follows that $\mathcal D_u$ can answer a query for the value $d_i^{(t)}(u)$ in $O(1)$ time. To answer the other two types of queries, $\mathcal D_u$ maintains $V(T^{(t)})$ as a red-black tree keyed by vertex labels; this clearly allows both types of queries to be answered in $O(\log n)$ time.

\subparagraph{Handling the update $t = t_u$ for $\mathcal D_u$:}
For the initial update $t = t_u$, a tree $\overline T_i^{(t)}(u)$ is computed by running Dijkstra's algorithm from $u$ in $G^{(t)}$ with the following modifications:
\begin{enumerate}
\item the priority queue initially contains only $u$ with estimate $0$; all other vertices implicitly have an estimate of $\infty$,
\item a vertex is only added to the priority queue if a relax operation caused its distance estimate to be strictly decreased to a value of at most $d_i$,
\item the algorithm stops when the priority queue is empty or as soon as $\deg_{G^{(t)}}(\overline T_i^{(t)}(u)) > m_i$.
\end{enumerate}
If the algorithm emptied its priority queue, $\mathcal D_u$ sets $T^{(t)}\leftarrow\overline T_i^{(t)}(u)$ and $d^{(t)}\leftarrow d_i$, finishing the update.

Now, assume that the algorithm did not empty its priority queue and let $v_{\max}$ denote the last vertex added to $\overline T_i^{(t)}(u)$. $\mathcal D_u$ lets $d^{(t)}$ be the largest power of $(1+\varepsilon)$ such that $d^{(t)} < d_{\overline T_i^{(t)}(u)}(v_{\max})$. Then it obtains $T^{(t)}$ as the subtree of $\overline T_i^{(t)}(u)$ consisting of all vertices of distance at most $d^{(t)}$ from $u$ in $\overline T_i^{(t)}(u)$. Finally, it outputs $\overline T_i^{(t)}(u)$.

\subparagraph{Handling updates $t > t_u$ for $\mathcal D_u$:}
Next, consider update $t > t_u$. $\mathcal D_u$ ignores updates to $U$ so we
assume that update $t$ is of the form $\FuncSty{Insert-Edge}(e^{(t)})$.
Assume that $\mathcal D_u$ has obtained $T^{(t-1)}$ and $d^{(t-1)}$ in the
previous update. To obtain $T^{(t)}$ and $d^{(t)}$, $\mathcal D_u$ regards
$e^{(t)}$ as two oppositely directed edges. Note that for at most one of
these edges $(v_1^{(t)},v_2^{(t)})$, we have $d_{T^{(t-1)}}(v_1^{(t)}) +
w(e^{(t)}) < d_{T^{(t-1)}}(v_2^{(t)})$. If no such edge exists, $\mathcal
D_u$ sets $T^{(t)} \leftarrow T^{(t-1)}$ and $d^{(t)} \leftarrow
d^{(t-1)}$, finishing the update. Otherwise, $\mathcal D_u$ applies a
variant of Dijkstra's algorithm. During initialization, this variant sets,
for each vertex $v\in V(T^{(t-1)})$, the starting estimate of $v$ to
$d_{T^{(t-1)}}(v)$ and sets its predecessor to be the parent of $v$ in
$T^{(t-1)}$; all other vertices implicitly have an estimate of $\infty$.
The priority queue is initially empty. In the last part of the
initialization step, the edge $(v_1^{(t)},v_2^{(t)})$ is relaxed. The rest
of the algorithm differs from the normal Dijkstra algorithm in the
following way:
\begin{enumerate}
\item a vertex $v$ is only added to the priority queue if a relax operation caused the estimate for $v$ to be strictly decreased to a value of at most $d^{(t-1)}$,
\item the algorithm stops when the priority queue is empty or as soon as the the total degree in $G^{(t)}$ of vertices belonging to the current tree found by the algorithm exceeds $m_i$.
\end{enumerate}

Let $\overline T_i^{(t)}(u)$ be the tree found by the Dijkstra variant. If the priority queue is empty at this point, $\mathcal D_u$ sets $T^{(t)} \leftarrow \overline T_i^{(t)}(u)$ and $d^{(t)} \leftarrow d^{(t-1)}$. Otherwise, $\mathcal D_u$ computes $T^{(t)}$ and $d^{(t)}$ in exactly the same manner as in the case above where $t = t_u$ and where the priority queue was not emptied; finally, $\mathcal D_u$ outputs $\overline T_i^{(t)}(u)$.

\subparagraph{Properties of $\mathcal D_u$:} We now show the two properties of $\mathcal D_u$ mentioned earlier. We repeat them here for convenience:
\begin{enumerate}
\item $T^{(t)}(u) = T_i^{(t)}(u)$ and $d^{(t)}(u) = d_i^{(t)}(u)$ for all $t\in\{t_u,\ldots,m\}$,
\item in each update $t\in\{t_u,\ldots,m\}$ where either $t > t_u$ and $d_i^{(t)}(u) < d_i^{(t-1)}(u)$ or $t = t_u$ and $d_i^{(t)}(u) < d_i$, $\mathcal D_u$ outputs a tree $\overline T_i^{(t)}(u)$ rooted at $u$ such that $r(\overline T_i^{(t)}(u))\le (1+\varepsilon)d_i^{(t)}(u)\le d_i$ and $\deg_{G^{(t)}}(\overline T_i^{(t)}(u)) > m_i$. In all other updates, no tree is output.
\end{enumerate}
The first property is shown by induction on $t\ge
t_u$. This is clear when $t = t_u$ so assume in the following that $t > t_u$,
that $T^{(t-1)} = T_i^{(t-1)}(u)$ and $d^{(t-1)} = d_i^{(t-1)}(u)$, and that
update $t$ is an operation $\FuncSty{Insert-Edge}(e^{(t)})$. The first property will follow if we can show that $T^{(t)} = T_i^{(t)}(u)$ and $d^{(t)} = d_i^{(t)}(u)$.

If the Dijkstra variant was not executed then no edge was relaxed which implies that $T^{(t)} = T^{(t-1)} = T_i^{(t-1)}(u) = T_i^{(t)}(u)$ and $d^{(t)} = d^{(t-1)} = d_i^{(t-1)}(u) = d_i^{(t)}(u)$, as desired. Otherwise, consider first the case where $d_i^{(t)}(u) = d_i^{(t-1)}(u)$. Then $d^{(t-1)} = d_i^{(t)}(u)$ so the priority queue of the Dijkstra variant must be empty at the end of update $t$. Combining this with the observation that any vertex $v$ whose distance from $u$ in $G^{(t)}$ is smaller than in $G^{(t-1)}$ must be on a $u$-to$-v$ path containing $e^{(t)}$, it follows that the Dijkstra variant computes $T^{(t)} = T_i^{(t)}(u)$ and $d^{(t)} = d^{(t-1)} = d_i^{(t)}(u)$, as desired. For the case where $d_i^{(t)}(u) < d_i^{(t-1)}(u)$, the priority queue of the Dijkstra variant is not emptied; it follows by definition of $T_i^{(t)}(u)$ and $d_i^{(t)}(u)$ that also in this case, $T^{(t)} = T_i^{(t)}(u)$ and $d^{(t)} = d_i^{(t)}(u)$.

To show that $\mathcal D_u$ satisfies the second of the properties above, consider an update $t \ge t_u$. Assume first that $t = t_u$. Then $\overline T_i^{(t)}(u)$ is output if and only if $d_i^{(t)}(u) < d_i$; this follows since $d^{(t)} = d_i^{(t)}(u)$, since $d^{(t)} = d_i$ when $\overline T_i^{(t)}$ is not output, and since $d^{(t)} < d_{\overline T_i^{(t)}(u)}(v_{\max})\le d_i$ when $\overline T_i^{(t)}$ is output. If $d_i^{(t)}(u) < d_i$, then Dijkstra's algorithm stopped without emptying its priority queue which implies that $\deg_{G^{(t)}}(\overline T_i^{(t)}(u)) > m_i$; furthermore, by the choice of $d^{(t)}$, $r(\overline T_i^{(t)}(u)) = d_{\overline T_i^{(t)}(u)}(v_{\max})\le (1+\varepsilon)d^{(t)} = (1+\varepsilon)d_i^{(t)}(u)$, as desired. The inequality $(1+\varepsilon)d_i^{(t)}(u)\le d_i$ holds since $d_i^{(t)}(u) < d_i$ implies that $d_i^{(t)}(u) \le d_i/(1+\varepsilon)$.

The case $t > t_u$ is quite similar. We may assume that this update inserts $e^{(t)}$ into $G^{(t)}$. If $d_i^{(t)}(u) = d_i^{(t-1)}(u)$ then as shown above, no tree is output in update $t$. Now, assume that $d_i^{(t)}(u) < d_i^{(t-1)}(u)$. Then the Dijkstra variant did not empty its priority queue so it outputs tree $\overline T_i^{(t)}(u)$. Clearly, $\deg_{G^{(t)}}(\overline T_i^{(t)}(u)) > m_i$ and since $d^{(t-1)} = d_i^{(t-1)}(u)$, the same argument as in the case where $t = t_u$ gives $r(\overline T_i^{(t)}(u))\le (1+\varepsilon)d_i^{(t)}(u)$. The inequality $(1+\varepsilon)d_i^{(t)}(u)\le d_i$ holds since $d_i^{(t)}(u) < d_i^{(t-1)}(u)$ implies that $d_i^{(t)}(u)\le d_i^{(t-1)}(u)/(1+\varepsilon)\le d_i/(1+\varepsilon)$. This shows the second of the two properties for $\mathcal D_u$ mentioned above.

\subparagraph{Bounding update time of $\mathcal D_u$:} We now bound the update
time for $\mathcal D_u$ where we ignore the cost of updates $t$ where $e^{(t)}$
is not incident to $T_i^{(t-1)}(u)$; when we use $\mathcal D_u$ in the final
data structure $\mathcal D$ below, $\mathcal D$ will ensure that
$\FuncSty{Insert-Edge}$ will only be applied to edges if they are incident to $T_i^{(t-1)}(u)$ and we show that this suffices to ensure the two properties of $\mathcal D_u$.

Consider an update $t_u$. Observe that our two Dijkstra variants (the one described in the case $t = t_u$ and the one described in the case $t > t_u$) are terminated as soon as the total degree in $G^{(t)}$ of vertices extracted from the priority queue exceeds $m_i$. Ignoring the cost of the initialization step of the second variant, it follows from a standard analysis of Dijkstra's algorithm that both variants run in time $O(m_i\log n)$. To bound the time for the initialization step of the second variant, note that the desired starting estimates and predecessor pointers are present in $T^{(t-1)}$ and this tree is available at the beginning of the update. Hence, the work done in the initialization step thus only involves relaxing a single edge $(v_1^{(t)},v_2^{(t)})$. With this implementation, the cost of the initialization step does not dominate the total cost of the update.

The number of updates $t > t_u$ for which $d_i^{(t)}(u) < d_i^{(t-1)}(u)$ is at most $\log_{1+\varepsilon} d_i = O_{\varepsilon}(i + \log d)$. As shown above, the time spent in each such update is $O(m_i\log n)$ which over all such updates is $O_{\varepsilon}(m_i\log n(i + \log d))$ time.

Now, consider a maximal range of updates $\{t_1,t_1+1\ldots,t_2\}\subseteq \{t_u,t_u+1,\ldots,t_{\max}\}$ where $d_i^{(t_2)}(u) = d_i^{(t_1)}(u)$ and consider an update $t\in\{t_1+1,t_1+2,\ldots,t_2\}$. Assuming that the Dijkstra variant is executed, it must empty its priority queue in this update. Let $V_1^{(t)} = V(T_i^{(t)}(u))\setminus V(T_i^{(t-1)}(u))$ and let $V_2^{(t)}$ be the set of vertices $v\in V(T_i^{(t)}(u))\cap V(T_i^{(t-1)}(u))$ such that $d_{T_i^{(t)}(u)}(v) < d_{T_i^{(t-1)}(u)}(v)$. Since the Dijkstra variant only adds a vertex of $T_i^{(t-1)}(u)$ to the priority queue if the distance estimate of the vertex is strictly decreased, we can charge the running time cost of the Dijkstra variant to $\deg_{G^{(t)}}(V_1^{(t)}\cup V_2^{(t)})\log n$. In the following, we bound $\deg_{G^{(t)}}(V_1^{(t)})\log n$ and $\deg_{G^{(t)}}(V_2^{(t)})\log n$ separately over all $t\in\{t_1+1,t_1+2,\ldots,t_2\}$ and all maximal ranges $\{t_1,t_1+1,\ldots,t_2\}$.

Let one such range $\{t_1,t_1+1,\ldots,t_2\}$ be given. Since $V(T_i^{(t-1)}(u))\subseteq V(T_i^{(t)}(u))$ for all $t\in\{t_1+1,t_1+2,\ldots,t_2\}$, we get
\[
      \sum_{t\in\{t_1+1,t_1+2,\ldots,t_2\}}\deg_{G^{(t)}}(V_1^{(t)})\log n
  \le \deg_{G^{(t_2)}}(T_i^{(t_2)}(u))\log n\le m_i\log n.
\]
for each range $\{t_1,t_1+1,\ldots,t_2\}$ which over all ranges $\{t_1,t_1+1,\ldots,t_2\}$ is $O_{\varepsilon}(m_i\log n(i + \log d))$.

Next, we bound the sum of $\deg_{G^{(t)}}(V_2^{(t)})\log n$ over all $t\in \{t_1+1,t_1+2,\ldots,t_2\}$ and all ranges $\{t_1,t_1+1,\ldots,t_2\}$. Let $\{t_1,t_1+1,\ldots,t_2\}$ and $v\in V$ be given. For each $t\in\{t_1+1,t_1+2,\ldots,t_2\}$ where $v\in V_2^{(t)}$, we have $d_{T_i^{(t)}(u)}(v) < d_{T_i^{(t-1)}(u)}(v)$. Since edge weights are integers, the sum of degrees of $v$ over all such $t$ is at most $(d_{T_i^{(t_1+1)}(u)}(v) - d_{T_i^{(t_2)}(u)}(v))\deg_{G^{(t_2)}}(v)\log n$. Observe that $v\in V_2^{(t)}$ for some $t\in\{t_1+1,t_1+2,\ldots,t_2\}$ implies that $v\in V(T_i^{(t_2)}(u))$. Summing over all $v$ thus gives
\begin{align*}
  \sum_{t\in\{t_1+1,t_1+2,\ldots,t_2\}}\deg_{G^{(t)}}(V_2^{(t)})\log n & \le\sum_{v\in V(T_i^{(t_2)}(u))}(d_{T_i^{(t_1+1)}(u)}(v) - d_{T_i^{(t_2)}(u)}(v))\deg_{G^{(t_2)}}(v)\log n\\
  & \le d_i^{(t_1)}(u)m_i\log n.
\end{align*}
Note that since $d_i^{(t_2+1)}(u) < d_i^{(t_1)}(u)$, we in fact have $d_i^{(t_2+1)}(u) \le d_i^{(t_1)}(u)/(1+\varepsilon)$. Summing over all ranges $\{t_1,t_1+1,\ldots,t_2\}$ thus gives a geometric sum of value $O_{\varepsilon}(m_id_i\log n)$.

We conclude that $\mathcal D_u$ requires time $O_{\varepsilon}(m_i\log n(d_i + i + \log d)) = O_{\varepsilon}(m_id_i\log n)$ over all updates $t$ consisting of the insertion of an edge $e^{(t)}$ which is incident to $T_i^{(t-1)}(u)$.

\subparagraph{The final data structure:} We have shown that $\mathcal D_u$ satisfies the two properties stated at the beginning of the proof and that the total update time over updates $t$ for which $e^{(t)}$ is incident to $T_i^{(t-1)}(u)$ is $O_{\varepsilon}(m_id_i\log n)$. We are now ready to give a data structure $\mathcal D$ satisfying the lemma.

Initially, $\mathcal D$ sets $U^{(0)} = \emptyset$. If update $t$ is an
operation of the form $\FuncSty{Insert-Vertex}(u)$, $\mathcal D$ initializes a new structure $\mathcal D_u$. For each $v\in V$, $\mathcal D$ keeps the set $U^{(t)}(v)$ of those vertices $u\in U^{(t)}$ for which $v$ belongs to the tree $T_i^{(t)}(u)$ maintained by $\mathcal D_u$. This set is implemented with a red-black tree keyed by vertex labels. We extend each data structure $\mathcal D_u$ so that when a vertex $v$ joins (resp.~leaves) $T_i^{(t)}(u)$, $u$ joins (resp.~leaves) $U^{(t)}(v)$. This can be done without affecting the update time bound obtained for $\mathcal D_u$ above.

If an update $t$ is of the form $\FuncSty{Insert-Edge}(e^{(t)})$ where $e^{(t)} = (v_1^{(t)},v_2^{(t)})$, $\mathcal D$ identifies the set $U^{(t-1)}(v_1^{(t)})\cup U^{(t-1)}(v_2^{(t)})$ and updates $\mathcal D_u$ with the insertion of $e^{(t)}$ for each $u$ in this set. This suffices to correctly maintain all data structures $\mathcal D_u$ since for each $u\in U^{(t-1)}\setminus U^{(t-1)}(v_1^{(t)})\cup U^{(t-1)}(v_2^{(t)})$, we have $T_i^{(t)}(u) = T_i^{(t-1)}(u)$ and $d_i^{(t)}(u) = d_i^{(t-1)}(u)$, implying that $\mathcal D_u$ need not be updated. Hence, $\mathcal D$ handles updates in the way stated in the lemma and has total update time $O(m) + O_{\varepsilon}(|U^{(t_{\max})}|m_id_i\log n)$, as desired.

Answering a query for values $d_{T_i(u)}(v)$ or $d_i(u)$ or for whether $v\in V(T_i(u))$ is done by querying $\mathcal D_u$. Since $\mathcal D_u$ can be identified in $O(1)$ time, the query time bounds for $\mathcal D$ match those for $\mathcal D_u$. This completes the proof.
\end{proof}

\subsection{The distance oracle}\label{sec:oracle}
We are now ready to present our incremental distance oracle. Pseudocode for the
preprocessing step is done with the procedure $\FuncSty{Initialize}(V,k)$ in
\Cref{alg:initialize}. Inserting an edge $(v_1,v_2)$ with integer
weight $w > 0$ is done with the procedure
$\FuncSty{Insert}(v_1,v_2,w(v_1,v_2))$ in \Cref{alg:insert} and a
query for the approximate distance
between two vertices $u$ and $v$ is done with the procedure
$\FuncSty{Query}(u,v)$ in \Cref{alg:query}.

The high level intuition of our construction is that we maintain increasingly smaller subsets of vertex sets denoted $A_i$, where $A_0$ is the entire vertex set $V$; see \Cref{fig:oracle}. For each vertex $v$, we grow a ball up to a threshold size, and we let the centers of a maximal set of disjoint balls be promoted to the next level $A_{i+1}$, where the same procedure happens. An implication is that $A_{i+1}$ is much smaller than $A_{i}$ and we can thus afford to grow larger balls as $i$ grows, i.e. we let the ball threshold size grow with $i$.

In order to bound stretch, we need balls to have roughly the same radius. To ensure this, we partition balls centered at vertices of $A_i$ into classes such that balls in the $j$th class all have radius within a constant factor of $(1+\varepsilon)^j$. For each class, we keep a maximal set of disjoint balls as described above and $A_{i+1}$ is the union of centers of these balls over all classes. In class $j$, each vertex $v$, which is the center of a ball not belonging to this maximal set points to a representative vertex $n_{i,j}(v)$. This representative vertex is picked in the intersection with another ball in class $j$ centered at a vertex of $A_{i+1}$. Every vertex $w$ in this other ball has a pointer $r_{i,j}(w)$ to the center. These pointers are used as navigation in the distance query algorithm when identifying a vertex $u_{i+1}\in A_{i+1}$ from a vertex $u_i\in A_i$; see \Cref{fig:oracle}. The fact that the two balls centered at $u_i$ resp.~$u_{i+1}$ have roughly the same radius is important to ensure that the stretch only grows by at most a constant factor in each iteration of the query algorithm.

\begin{algorithm}
    \caption{\FuncSty{Initialize}}
    \label{alg:initialize}
    \DontPrintSemicolon
    \SetKwInOut{Input}{input}\SetKwInOut{Output}{output}
    \Input{$V,k$}
    \BlankLine
    $A_0\gets V$\;
    Initialize $\mathcal{D}_0$ as an instance of the data structure of
    \Cref{Lem:MaintainIncTrees}\;
    \For{$i=1\to k-1$}{
        $A_i\gets\emptyset$\;
        Initialize $\mathcal{D}_i$ as an instance of the data structure of
        \Cref{Lem:MaintainIncTrees}\;
        \For{$j=0\to\log_{1+\eps}d_i$}{
            $W_{i,j}\gets\emptyset$\;
            Associate with each $v\in V$ uninitialized variables $n_{i,j}(v)$
            and $r_{i,j}(v)$\;
        }
    }
    \For{$u\in V$}{
        $\mathcal{D}_0.\FuncSty{Insert-Vertex}(u)$\;
    }
\end{algorithm}

\begin{algorithm}
    \caption{\FuncSty{Insert}}
    \label{alg:insert}
    \DontPrintSemicolon
    \SetKwInOut{Input}{input}\SetKwInOut{Output}{output}
    \Input{$v_1,v_2,w(v_1,v_2)$}
    \BlankLine
    Add $(v_1,v_2)$ to $E$ with weight $w(v_1,v_2)$\;
    $U_0,\ldots, U_{k-1}\gets \emptyset$\;
    \For{$i=0\to k-1$}{
        $\mathcal{T}_i\gets \mathcal{D}_i.\FuncSty{Insert-Edge}(v_1,v_2)$\;
        $A_i\gets A_i\cup U_i$\;
        \For{$u\in U_i$}{
            $\mathcal{T}_i\gets \mathcal{T}_i\cup
            \mathcal{D}_i.\FuncSty{Insert-Vertex}(u)$\;
        }
        \For{$\overline{T}_i(u)\in \mathcal{T}_i$}{
            $j\gets 1 + \log_{1+\eps}d_i(u)$\;
            $T_{i,j}(u)\gets \overline{T}_i(u)$\;
            \If{$V(T_{i,j}(u))\cap W_{i,j} = \emptyset$}{
                $W_{i,j}\gets W_{i,j}\cup V(T_{i,j}(u))$\;
                \For{$v\in V(T_{i,j}(u))$}{
                    $r_{i,j}(v)\gets u$\;
                }
                $U_{i+1}\gets U_{i+1}\cup \{u\}$\;
            }
            \Else{
                $n_{i,j}(u)\gets$ an arbitrary vertex of $W_{i,j}\cap
                V(T_{i,j}(u))$\;
            }
        }
    }
\end{algorithm}

\begin{algorithm}
    \caption{\FuncSty{Query}}
    \label{alg:query}
    \DontPrintSemicolon
    \SetKwInOut{Input}{input}\SetKwInOut{Output}{output}
    \Input{$u,v$}
    \Output{Estimated distance between $u$ and $v$}
    \BlankLine
    $u_0\gets u$\;
    $s_0\gets 0$\;
    \For{$i = 0\to k-1$}{
        \If{$v\in V(T_i(u_i))$}{
            \Return $s_i + d_{T_i(u_i)}(v)$\;
        }
        \If{$d_i(u_i) = d_i$}{
            \Return $\infty$\;
        }
        \If{$u_i\in A_{i+1}$}{
            $u_{i+1}\gets u_i$\;
            $s_{i+1}\gets s_i$\;
        }
        \Else{
            $j\gets 1 + \log_{1+\eps}d_i(u_i)$\;
            $w\gets n_{i,j}(u_i)$\;
            $u_{i+1}\gets r_{i,j}(w)$\;
            $s_{i+1}\gets s_i + d_{T_{i,j}(u_i)}(w) + d_{T_{i,j}(u_{i+1})}(w)$\;
        }
    }
\end{algorithm}
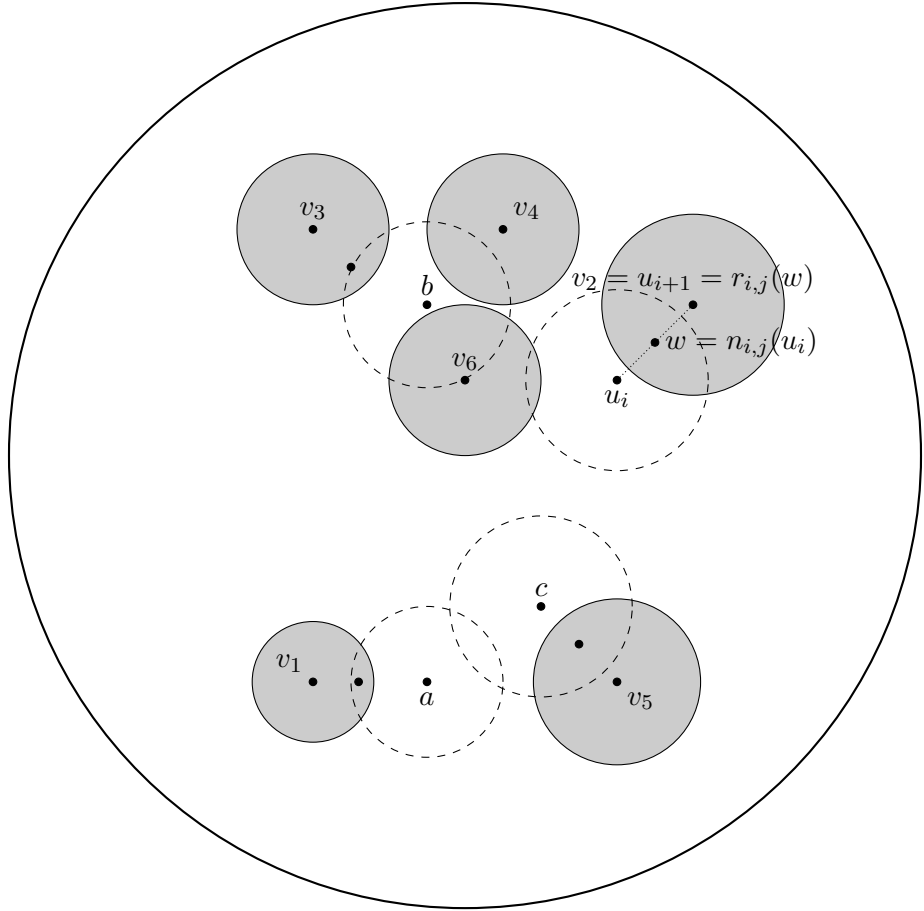
\begin{figure}
\centering
\begin{tikzpicture}
\definecolor{lyl}{gray}{0.8}

    \draw[thick] (0,0) circle (6cm);

     \draw[fill=lyl] (0,1) circle (1cm)node [above] {$v_6$};
      \draw[fill=lyl] (2,-3) circle (1.1cm)node [below right] {$v_5$};
        \draw[fill=lyl] (0.5,3) circle (1cm)node [above right] {$v_4$};
      \draw[fill=lyl] (-2,3) circle (1cm)node [above] {$v_3$};
           \draw[fill=lyl] (3,2) circle (1.2cm)node [above] {$v_2 = u_{i+1} = r_{i,j}(w)$};
      \draw[fill=lyl] (-2,-3) circle (0.8cm) node [above left] {$v_1$};

      \draw[fill=black] (0,1) circle (0.05cm);
       \draw[fill=black] (2,-3) circle (0.05cm);
        \draw[fill=black] (0.5,3) circle (0.05cm);
         \draw[fill=black] (-2,3) circle (0.05cm);
          \draw[fill=black] (3,2) circle (0.05cm);
           \draw[fill=black] (-2,-3) circle (0.05cm);

     \draw[dashed] (-0.5,-3) circle (1cm)node [below] {$a$};  
      \draw[dashed] (-0.5,2) circle (1.1cm)node [above] {$b$};
        \draw[dashed] (1,-2) circle (1.2cm)node [above] {$c$}; 
          \draw[dashed] (2,1) circle (1.2cm)node [below] {$u_i$};
                \draw[fill=black] (-0.5,-3) circle (0.05cm);
         \draw[fill=black] (-0.5,2) circle (0.05cm);
          \draw[fill=black] (1,-2) circle (0.05cm);
           \draw[fill=black] (2,1) circle (0.05cm);

            \draw[fill=black] (-1.5,2.5) circle (0.05cm);

            \draw[fill=black] (2.5,1.5) circle (0.05cm) node [right] {$w = n_{i,j}(u_i)$} ;
            \draw [densely dotted] (2.5,1.5)-- (2,1);
            \draw [densely dotted] (2.5,1.5) -- (3,2);

            \draw[fill=black] (1.5,-2.5) circle (0.05cm);

            \draw[fill=black] (-1.4,-3) circle (0.05cm) ;

\end{tikzpicture}
\caption{A high-level overview of the distance oracle construction. The vertices $v_1, \ldots, v_5$ are centers of disjoint (grey) balls and are thus promoted to $A_{i+1}$, while $W_{i,j}$ is the union over the vertices of the disjoint balls. The grey balls have radius roughly $(1+\eps)^j$, and we keep a set of balls for every $j \in \{1, \ldots,\log_{1+\eps}d_i \}$.  A query from a center $u_i$ of a non-disjoint ball has an assosicated vertex $w$ in an intersecting grey ball, which in turn has a pointer to the ball center $u_{i+1} = r_{i,j}(w)$.  }\label{fig:oracle}
\end{figure}

The following lemmas are crucial when we bound update and query time as well as
stretch. For $i = 0,\ldots,k-1$ and $j = 1,\ldots,\log_{1+\varepsilon}d_i$, let
$\mathcal T_{i,j}$ be the dynamic set of trees $T_{i,j}(u)$ obtained so far for
which the test in line $11$ of $\FuncSty{Insert}$ succeeded. Note that for any
$j$ in line $9$ of $\FuncSty{Insert}$, $1\le j =
\log_{1+\varepsilon}((1+\varepsilon)d_i(u))\le\log_{1+\varepsilon} d_i$ by
\Cref{Lem:MaintainIncTrees} so $W_{i,j}$ is well-defined and initialized
to $\emptyset$ in procedure $\FuncSty{Initialize}$.
\begin{lemma}\label{Lem:DisjointBalls}
After each update, the following holds. For any $i = 0,\ldots,k-2$ and any $j = 1,\ldots,\log_{1+\varepsilon} d_i$, $W_{i,j}$ is the disjoint union of $V(T_{i,j}(u))$ over all $T_{i,j}(u)\in\mathcal T_{i,j}$. Furthermore, $A_{i+1} = \bigcup_j\{u | T_{i,j}(u)\in\mathcal T_{i,j}\}$.
\end{lemma}
\begin{proof}
    For every $u$ added to $U_{i+1}$ in procedure
    $\FuncSty{Insert}(v_1,v_2,w(v_1,v_2))$, $V(T_{i,j}(u))\cap W_{i,j} =
    \emptyset$ just before the update in line $12$ and line $12$ is the only
    place where $W_{i,j}$ is updated. All vertices of $U_{i+1}$ are added to
    $A_{i+1}$ in line $5$ of iteration $i+1$ and this is the only place where
    $A_{i+1}$ is updated.
\end{proof}
\begin{lemma}\label{Lem:ASetSizeBound}
After each update, $|A_i| = O_{\varepsilon}((m^{1-i/k})(i+\log d))$ for $i = 0,\ldots,k-1$.
\end{lemma}
\begin{proof}
The lemma is clear for $i = 0$ since $|A_0| = n\le m$. Now, let $i\in\{1,\ldots,k-2\}$ be given. We will bound $|A_{i+1}|$. Consider any $j\in\{1,\ldots,\log_{1+\varepsilon}d_i\}$. Since the total degree of vertices in $G$ is at most $2m$ and since the sets $V(T_{i,j}(u))$ are pairwise disjoint for all $T_{i,j}(u)\in\mathcal T_{i,j}$ by \Cref{Lem:DisjointBalls}, it follows from \Cref{Lem:MaintainIncTrees} that the number of roots of these trees is less than $2m/m_i$. \Cref{Lem:DisjointBalls} then implies that $|A_{i+1}| = O((m/m_i)\log_{1+\varepsilon}d_i) = O_{\varepsilon}((m^{1 - (i+1)/k})(i+\log d))$. This shows the lemma. 
\end{proof}

\begin{lemma}\label{Lem:Detour}
After each update, the following holds. Let $i\in\{0,\ldots,k-2\}$ and $u\in A_i\setminus A_{i+1}$ be given and let $j = 1 + \log_{1+\varepsilon}d_i(u)$ and $w = n_{i,j}(u)$. If $d_i(u) < d_i$ then $w\in V(T_{i,j}(u))\cap V(T_{i,j}(r_{i,j}(w)))$ and $r_{i,j}(w)\in A_{i+1}$.
\end{lemma}
\begin{proof}
By \Cref{Lem:MaintainIncTrees}, since $d_i(u) < d_i$, there must have been some update to $\mathcal D_i$ that output a tree $\overline T_i(u)$; consider the last such tree. Then $d_i(u)$ has not changed since then and so $T_{i,j}(u)$ must be that tree. Since $u\notin A_{i+1}$, we must have $V(T_{i,j}(u))\cap W_{i,j}\ne\emptyset$ so $w\in V(T_{i,j}(u))\cap W_{i,j}$. Since $w\in W_{i,j}$, we have $w\in V(T_{i,j}(r_{i,j}(w))$. At some point, $r_{i,j}(w)$ was added to $U_{i+1}$ and hence to $A_{i+1}$. Since vertices are never removed from $A_{i+1}$, the lemma follows.
\end{proof}

\subsection{Bounding time and stretch}
After replacing $\varepsilon$ with $\varepsilon/2$, the following lemma gives the update time bound claimed in \Cref{Thm:DistOracle}.
\begin{lemma}
    A total of $O_{\varepsilon}(m^{1+1/k}(3+2\varepsilon)^{k-1}d(k+\log d)\log
    n)$ time is spent in all calls to procedure $\FuncSty{Insert}$.
\end{lemma}
\begin{proof}
    It is easy to see that each execution of lines $9$ to $17$ in procedure
    $\FuncSty{Insert}$ can be implemented to run in time $O(|V(\overline
    T_i(u))|)$. Hence, the total time spent in all calls to $\FuncSty{Insert}$
    is dominated by the total update time of data structures $\mathcal D_i$,
    for $i = 0,\ldots,k-1$. Note that after each update, $A_i$ is the current
    set of vertices added to $\mathcal D_i$. Letting $A_i^{(t_{\max})}$ be the
    set $A_i$ after the last update, it follows from
    Lemmas~\ref{Lem:MaintainIncTrees} and~\ref{Lem:ASetSizeBound} that
    $\mathcal D_0,\ldots,\mathcal D_{k-1}$ have total update time
\begin{align*}
  O(km) + \sum_{i = 0}^{k-1}O_{\varepsilon}(|A_i^{(t_{\max})}|m_id_i\log n)
   & = \sum_{i = 0}^{k-1}O_{\varepsilon}(m^{1+1/k}d_i(i+\log d)\log n)\\
   & = \sum_{i = 0}^{k-1}O_{\varepsilon}(m^{1+1/k}(3+2\varepsilon)^id(i+\log d)\log n)\\
   & = O_{\varepsilon}(m^{1+1/k}(3+2\varepsilon)^{k-1}d(k+\log d)\log n),
\end{align*}
where the last bound follows from a geometric sums argument.
\end{proof}

Finally, we bound query time and stretch with the following lemma; replacing $\varepsilon$ with $\varepsilon/2$ gives the bounds of \Cref{Thm:DistOracle}.
\begin{lemma}
    Procedure $\FuncSty{Query}(u,v)$ outputs in $O(k\log n)$ time a value $\tilde d_G(u,v)$ such that $d_G(u,v)\le \tilde d_G(u,v)$ and such that if $d_G(u,v)\le d$ then $\tilde d_G(u,v)\le (2(3+2\varepsilon)^{k-1}-1)d_G(u,v)$.
\end{lemma}
\begin{proof}
    To bound the stretch, we will first show the following loop invariant: at
    the beginning of the $i$th execution of the for-loop of procedure
    $\FuncSty{Query}(u,v)$, $u_i\in A_i$ and $s_i\le ((3+2\varepsilon)^i -
    1)d_G(u,v)$. This is clear when $i = 0$ so assume that $i > 0$ and that the
    loop invariant holds at the beginning of the $i$th iteration. We need to
    show that if the beginning of the $(i+1)$th iteration is reached, the loop
    invariant also holds at this point.

    We may assume that the tests in lines $4$ and $6$ fail since otherwise, the
    $(i+1)$th iteration is never reached. If $u_i\in A_{i+1}$ then at the
    beginning of the $(i+1)$th iteration, we have $u_{i+1} = u_i\in A_{i+1}$
    and $s_{i+1} = s_i\le ((3+2\varepsilon)^i - 1)d_G(u,v) <
    ((3+2\varepsilon)^{i+1} - 1)d_G(u,v)$, as desired.

    Now, assume that $u_i\notin A_{i+1}$. Then $u_i\in A_i\setminus A_{i+1}$.
    Since the tests in lines $4$ and $6$ fail, we have $d_i(u) < d_G(u_i,v)$
    and $d_i(u) < d_i$. Since $s_i$ is the weight of some $u$-to-$u_i$-path in
    $G$, we have $d_G(u,u_i)\le s_i$. \Cref{Lem:Detour} now implies that
    $u_{i+1} = r_{i,j}\in A_{i+1}$ and
\begin{align*}
  s_{i+1} & \le s_i + 2(1+\varepsilon)d_i(u) < s_i + 2(1+\varepsilon)d_G(u_i,v)\le s_i + 2(1+\varepsilon)(d_G(u_i,u) + d_G(u,v))\\
         & \le (3+2\varepsilon)s_i + 2(1+\varepsilon)d_G(u,v)\le ((3+2\varepsilon)^{i+1} - 1)d_G(u,v),
\end{align*}
as desired.

    We can now show the stretch bounds. First observe that $m_{k-1} = 2m$.
    Since at any time, the total degree of vertices in $G$ is at most $2m$, it
    follows that $d_{k-1}(u) = d_{k-1}$ for all $u\in V$. Hence,
    $\FuncSty{Query}$ outputs a value in some iteration.

    The bound $d_G(u,v)\le\tilde d_G(u,v)$ is clear if $\tilde d_G(u,v) =
    \infty$ and it also holds if $\tilde d_G(u,v)$ is output in line $4$ since
    $s_i + d_{T_i(u)}(v)$ is the weight of some path in $G$.

Next, we give the upper bound on stretch. If the test in line $4$ succeeds in some iteration $i$, it follows from the the loop invariant that
\begin{align*}
  \tilde d_G(u,v) & = s_i + d_{T_i(u_i)}(v) \le 2s_i + d_G(u,v)\le 2((3+2\varepsilon)^i-1)d_G(u,v) + d_G(u,v)\\
   & \le (2(3+2\varepsilon)^{k-1}-1)d_G(u,v),
\end{align*}
as desired.

Now, assume that the test in line $4$ fails in some iteration $i$, i.e., assume
    that $v\notin V(T_i(u_i))$. Then $d_G(u_i,v) > d_i(u_i)$. If the test in
    line $6$ succeeds in iteration $i$ then $d_G(u_i,v) > d_i(u_i) = d_i\ge
    (3+2\varepsilon)^id$. The loop invariant and the observation above that
    $d_G(u,u_i)\le s_i$ imply that
\[
  d_G(u,v)\ge d_G(u_i,v) - d_G(u_i,u) > (3+2\varepsilon)^id - s_i \ge (3+2\varepsilon)^id + d_G(u,v) - (3+2\varepsilon)^id_G(u,v).
\]
Hence, $(3+2\varepsilon)^id_G(u,v) > (3+2\varepsilon)^id$ which gives $d_G(u,v)
    > d$. Since the upper bound on stretch is only required when $d_G(u,v)\le
    d$, outputting $\tilde d_G(u,v) = \infty$ in line $6$ is thus valid.

It remains to bound query time. Consider any iteration $i$. By
    \Cref{Lem:MaintainIncTrees}, performing the tests in lines $4$ and $6$
    and computing distances in line $15$ can be done in $O(\log n)$ time. Over
    all iterations, this is $O(k\log n)$.
\end{proof}

\section{A note on the lightness of other spanners}\label{sec:bad_lightness}
To motivate the problem of computing light spanners efficiently, we will in this section consider notable spanner constructions and show that they do \emph{not} provide light spanners. More precisely, we consider the three celebrated spanner constructions of Baswana and Sen~\cite{BaswanaS07},
Roditty and Zwick~\cite{RodittyZ11}, and Thorup and Zwick~\cite{ThorupZ05}, respectively, and we show that they do
not provide light spanners.

We first consider the algorithm from \cite{RodittyZ11}. This algorithm
creates a spanner by considering the edges in non-decreasing
order by weight similar to the greedy algorithm. It maintains an incremental
distance oracle of an \emph{unweighted} version of the spanner, and adds an
edge $(u,v)$, if there is no path between $u$ and $v$ of at most $2k-1$
\emph{edges}. Consider now running this algorithm on the graph of
\Cref{fig:bad_cycle} consisting of a cycle of $n = 2k+1$ edges where $2k$
of them have weight $1$ and the last has an arbitrarily large weight $W$. In
this case the algorithm of \cite{RodittyZ11} would add every edge to the
spanner, since $u$ and $v$ are only connected by a path of length $2k+1$ when
the edge $(u,v)$ is considered (disregarding the weight of $(u,v)$). This gives
us a lightness of $\Omega(W/n)$. Since $W$ can be arbitrarily large it
follows that no guarantee in terms of $k$ and $n$ can be given on the lightness.

A key part of the algorithm of \cite{BaswanaS07} is to arrange the vertices in $k$
layers $\emptyset=A_k\subseteq A_{k-1}\subseteq \cdots \subseteq A_0=V$ and clustering
the vertices of each layer. Each layer is formed by randomly sampling the clusters
of the previous layer with probability $n^{-1/k}$. Consider a vertex $w$ and
let $A_i$ be the first layer where $w$ is not sampled.
If $w$ is not adjacent to any cluster in $A_i$, then the smallest-weight edge
from $w$ to each of the clusters of $A_{i-1}$ is added to the spanner. Thus, in
the example of \Cref{fig:bad_cycle}, if neither $u$ nor its
neighbours are sampled, then the edge $(u,v)$ is added to the spanner. This
happens with probability at least $(1-n^{-1/k})^3$ and thus we cannot
even give a guarantee on the expected lightness of the spanner, as $W$ could be
very large compared to this probability.
\begin{figure}[htbp]
    \begin{center}
        \begin{tikzpicture}[scale=1.5,auto]
            \tikzstyle{vertex} = [circle,draw,fill=black!50,inner
            sep=0pt,minimum width=5pt]
            \node[vertex] (0) at (0:2) {};
            \node[vertex] (45) at (45:2) {};
            \node[vertex] (90) at (90:2) {};
            \node[vertex] (135) at (135:2) {};
            \node[vertex] (180) at (180:2) {};
            \node[vertex] (225) at (225:2) {};
            \node[vertex] (270) at (270:2) {};
            \node[vertex] (315) at (315:2) {};

            \node[left=1pt of 180] {\Large $u$};
            \node[left=1pt of 225] {\Large $v$};

            \draw (0) -- node {1} (45);
            \draw[dashed] (45) -- (90);
            \draw (90) -- node {1} (135);
            \draw (135) -- node {1} (180);
            \draw[line width=2pt] (180) -- node {$W$} (225);
            \draw (225) -- node {1} (270);
            \draw (270) -- node {1} (315);
            \draw (315) -- node {1} (0);
        \end{tikzpicture}
    \end{center}
    \caption{Example of a bad input graph to the algorithms of
    \cite{BaswanaS07} and \cite{RodittyZ11}: A cycle of $2k+1$ edges with one
    very heavy edge. This bad instance implies $\Omega(W)$ lightness for both algorithms.}
    \label{fig:bad_cycle}
\end{figure}
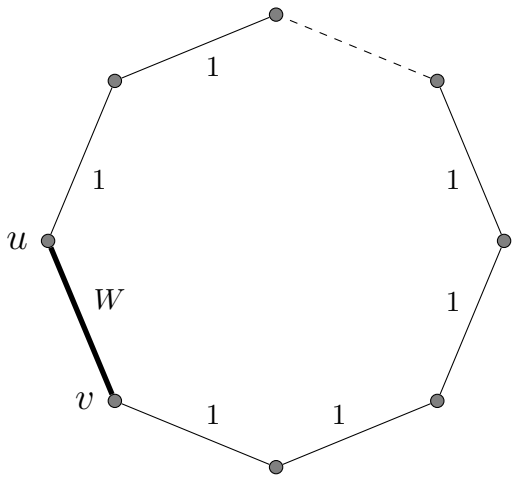

The spanner of \cite{ThorupZ05} also creates sets of vertices $\emptyset = A_k\subseteq
A_{k-1}\subseteq \cdots A_0 = V$, where each $A_i$ is formed by sampling the
vertices of $A_{i-1}$ independently with probability $n^{-1/k}$. For each vertex of
$v\in (A_i\setminus A_{i+1})$ they define the \emph{cluster} of $w$ to be the
set of all vertices in $V$ which are closer to $w$ than to any vertex in $A_{i+1}$.
The spanner they construct is simply the union of the shortest path trees of
each cluster with root in $w$. In particular, for the vertices $w\in A_{k-1}$ we
include the shortest path tree of the entire graph with root in $w$. We wish to
show that at least one of these shortest path trees have lightness $\Omega(n)$
with constant probability. To see this consider the graph of
\Cref{fig:tz_bad}. In this graph we have a complete graph $K$ on $n/2$
vertices with weights $1$ and a cycle $C$ on $n/2$ vertices with weights $1$. For each
vertex $u\in K$ and each vertex $v\in C$ there is an edge $(u,v)$ of large weight
$W$. Clearly the weight of the MST
for this graph is $W + n-2$, however the shortest path tree from any vertex $u\in
K$ has weight $nW/2 + n/2-1 = \Omega(nW)$. Since we expect half of the vertices of
$A_{k-1}$ to be from $K$ we see that the spanner has expected lightness at
least $\Omega(|A_{k-1}|\cdot n) =\Omega(n^{1+1/k})$ in this graph. We also
note that no edge of the spanner can have weight larger than that of the MST.
This follows because every edge of the spanner is part of some shortest-path
tree and if its weights was larger, we could simply replace it in the
shortest-path tree by the entire MST. Thus the lightness is also bounded from
above by $O(kn^{1+1/k})$.

\begin{figure}[htbp]
    \begin{center}
        \begin{tikzpicture}[scale=1.8,auto]
            \tikzstyle{vertex} = [circle,draw,fill=black!50,inner
            sep=0pt,minimum width=5pt]
            \draw (0:2) -- node[anchor=center,xshift=9pt,yshift=2pt] {$1$} (36:2);
            \draw (36:2) -- node[anchor=center,xshift=4pt,yshift=7pt] {$1$} (72:2);
            \draw (72:2) -- node[anchor=center,xshift=0pt,yshift=8pt] {$1$} (108:2);
            \draw (108:2) -- node[anchor=center,xshift=-4pt,yshift=7pt] {$1$}
            (144:2);
            \draw (144:2) -- node[anchor=center,xshift=-9pt,yshift=2pt] {$1$}
            (180:2);
            \draw (180:2) -- node[anchor=center,xshift=-9pt,yshift=-2pt] {$1$}
            (216:2);
            \draw (216:2) -- node[anchor=center,xshift=-4pt,yshift=-7pt] {$1$}
            (252:2);
            \draw (252:2) -- node[anchor=center,xshift=0pt,yshift=-8pt] {$1$}
            (288:2);
            \draw (288:2) -- node[anchor=center,xshift=4pt,yshift=-7pt] {$1$}
            (324:2);

            \foreach \x in {0,36,...,288,324}
            {
                \node[vertex] (\x) at (\x:2) {};
            };
            \draw[dashed] (324) -- (0);

            \node[circle,draw,minimum width=50pt,fill=black!15] (C) at (0:0)
            {$K$};
            \foreach \x in {0,36,...,324}
            {
                \draw (C) -- node {$W$} (\x);
            };

        \end{tikzpicture}
    \end{center}
    \caption{Example of a bad input graph for the spanner of \cite{ThorupZ05}.
    $K$ is the complete graph on $n^{1/k}$ vertices, where every edge has weight $1$. This bad instance implies $\Omega(n^{1+1/k})$ lightness for \cite{ThorupZ05}.}
    \label{fig:tz_bad}
\end{figure}
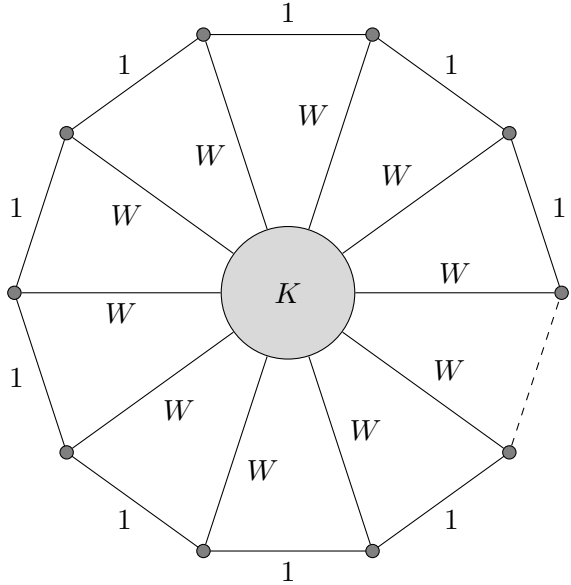

\section{Proof of \Cref{lem:fast_spanner_aspect}}\label{sec:app_proof_fast_aspect}

We build upon (the first variant of) the algorithm from \cite{ElkinS16}, while
we get an improved bound using the assumption of the small aspect ratio. 
The basic component of the algorithm is the spanner of Halperin and Zwick (see \Cref{thm:HZ96}). For simplicity we will assume that $a=1$. The construction/proof stays the same for general $a$.

	Fix $\rho=1+\epsilon$. 
	We start by computing the MST $T$.
	We divide the edges into $\log_\rho \Delta$ buckets. For $j\in[1,\log_{\rho}\Delta]$, let $E_{j}=\left\{ e\in E\mid w(e)\in [\rho^{j-1},\rho^{j})\right\}\setminus T$. 	
	We will construct separate spanner for each bucket. 
	We will use the $(i,\frac{\eps}{4})$-clustering as described in \Cref{sec:fw}. That is, for every $j$, we will have set of at most $n_{j}=\frac{4n}{\eps\rho^{j-1}}$ cluster, each with diameter bounded by $\eps\rho^{j-1}$ (in the MST metric). 
	Then, for each $j$, we contract each cluster and construct an unweighted graph
    $G_j$ with clusters as its vertices, where there is an edge between two
    clusters $\vpi_v$, $\vpi_u$, iff there are vertices  $v\in\vpi_v$ and
    $u\in\vpi_u$ such that $(u,v)\in E_j$. 
	Next we will construct a spanner $H_j$ for $G_j$ using \Cref{thm:HZ96}. 
	For each edge $\tilde{e}\in H_j$ we will add the edge $e\in E_j$ that created $\tilde{e}$ to our final spanner $H$ (if there multiple such edges, we add an arbitrary one).	
	Our final spanner $H$ contains the MST edges and the representatives of all the edges in $\bigcup_jH_j$.
	\begin{description}
        \item [{Stretch}] 
        As the diameter of every $j$-cluster is only an $\eps$ fraction of the weight of  edges  in $E_j$, bound on the stretch proof follows by similar arguments as in \Cref{eq:cluster-Stretch}. See \Cref{fig:stretchCluster} for illustration.

\item [{Number of edges}] by \Cref{thm:HZ96}
\begin{align*}
|H| & \le|T|+O(1)\cdot\sum_{j=1}^{\log_{\rho}\Delta}\left|n_{j}\right|^{1+\frac{1}{k}}=O(1)\cdot\left[\sum_{j=1}^{\log_{\rho}\Delta}\left(\frac{n}{\eps\rho^{j-1}}\right)^{1+\frac{1}{k}}\right]\\
& =O_{\epsilon}(n^{1+\frac{1}{k}})\cdot\left[\sum_{j=0}^{\infty}\left(\rho^{-j}\right)^{1+\frac{1}{k}}\right]=O_{\epsilon}(n^{1+\frac{1}{k}})~.
\end{align*}		
		\item [{Lightness}] as all the edges in $H_j$ have weight at most $\rho_j$,
\begin{align*}
|w(H)| & \le w(T)+O(1)\cdot\sum_{j=1}^{\log_{\rho}\Delta}\left|n_{j}\right|^{1+\frac{1}{k}}\cdot\rho^{j}=O_{\epsilon}(n^{1+\frac{1}{k}})\cdot\left[\sum_{j=1}^{\log_{\rho}\Delta}\left(\frac{1}{\rho^{j-1}}\right)^{1+\frac{1}{k}}\cdot\rho^{j}\right]\\
& =O_{\epsilon}(n^{1+\frac{1}{k}})\cdot\left[\sum_{j=0}^{\log_{\rho}\Delta}\left(\frac{1}{\rho^{j}}\right)^{\frac{1}{k}}\right]=O_{\epsilon}(n^{1+\frac{1}{k}}\cdot\log\Delta)~,
\end{align*}		
Thus the lightness bounded by $O_{\epsilon}(n^{\frac{1}{k}}\cdot\log\Delta)$.
		
	\item [{Running time}] Computing the MST takes $O(n\log n)$ times. Following the analysis of  \Cref{sec:fw}, the construction of the vertices
        for all the graphs $G_1,\dots,G_{\log_\rho\Delta}$ will take
        $O\left(\sum_{j=1}^{\log_{\rho}\Delta}\left|n_{j}\right|\right)= O\left(n\sum_{j=0}^{\log_{\rho}\Delta}\frac{1}{\rho^{j}}\right)=\Oeps(n)$ time.
	Adding the edges to the graphs will take additional $O(m+n\log n)$
    time. 
	Computing the spanners $H_j$ (using \Cref{thm:HZ96}) takes $\sum_j O(|E_j|)=O(m)$ time. All in all, a total of $O_{\epsilon}\left(m+n\log n\right)$ time.
	\end{description}

\section{Missing proofs from the analysis}	\label{sec:FastCWmissing}
\subsection{Stretch}\label{subapp:stretch}
In this section we bound the stretch of the spanner constructed in \Cref{alg:fastCW}  by $(1+O(\eps))(2k-1)$. 
Consider some edge $(u,v)=e\in E$.
If $w(e)\le \frac{k}{\eps}=g^\mu$, then $e$ is treated by $H_0$, the spanner constructed in line \ref{Line:G0} of \Cref{alg:fastCW}. 
Otherwise, let $i\ge \mu$ and $r\ge 1$ be such that $w(e)\in [g^i,g^{i+1})\subseteq [g^{r\mu},g^{(r+1)\mu})$. 
For any $j$, let $\vpi^{j}_v$ (resp. $\vpi^{j}_u$) denote the $j$-level
clusters containing $v$ (resp. $u$).

If $\vpi^{i}_v=\vpi^{i}_u$, by \Cref{clm:cluster_size_diam} $d_{E_{sp}}(v,u)\le \diam_{E_{sp}}(\vpi^{i}_v)\le \frac{kg^{i}}{2}\le \frac{k}{2}w(e)$ and we are done.

Otherwise, 
if $\vpi^{i}_v$ or $\vpi^{i}_u$ are light $i$-clusters, then during
the first phase, we add an edge $e'$ (of weight at most $w(e)$) between
$\vpi^{i-1}_v$ and $\vpi^{i-1}_u$. In particular
\begin{align*}
d_{E_{sp}}(v,u) & \le\diam_{E_{sp}}(\vpi_{v}^{i-1})+w(e')+\diam_{E_{sp}}(\vpi_{u}^{i-1}) \\
& \le\frac{kg^{i-1}}{2}+w(e)+\frac{kg^{i-1}}{2}\le(k/g+1)w(e)~.
\end{align*}

Finally consider the case where $\vpi^{i}_v$ and $\vpi^{i}_u$ are heavy $i$-clusters.
Recall the auxiliary graph $S_r$ constructed during the second phase. Its vertices were $V_{(r-1)\mu}$. In particular it contained an edge $e'$ from  $\vpi^{(r-1)\mu}_v$ to $\vpi^{(r-1)\mu}_u$, where $w(e')\le w(e)$. 
Note that the diameter of each $(r-1)\mu$ cluster is bounded by $\frac{k\cdot g^{(r-1)\mu}}{2}=\frac{\eps}{2}g^{r\mu}$, while in the used modified weight function $w_r(e')$ the minimal weight is $g^{r\mu}$. 
Following similar arguments to those in \Cref{eq:cluster-Stretch} there is a path in $E_{sp}$ of length $(1+O(\epsilon))(2k-1)\cdot w(e)$ from $v$ to $u$. See \Cref{fig:stretchCluster} for illustration.

\subsection{Proofs of \Cref{lem:M_r bound} and \Cref{lem:total_weight_spanners}}\label{subapp:lightness}
For $i$-level cluster $C$ (heavy or light), set $\widehat{\mbox{diam}}(C)$ to be the maximum value between the diameter (in $H$) of the cluster $C$ (in the time it was created) and $\frac{1}{c}kg^i$.

We start with proving some properties of the clusters:
\begin{claim}\label{clm:heavy_potential_diffrence} 
	Let $C$ be an
	$i$-level heavy cluster. Let $\mathcal{C}$ be the set of the $i-1$ clusters
	contained in $C$, then $\sum_{C'\in\mathcal{C}}\widehat{\mbox{diam}}(C')-\widehat{\mbox{diam}}(C)\ge\frac{\left|\mathcal{C}\right|g^{i-1}}{2c}\cdot k$.
\end{claim}
\begin{proof}
	By the definition of $\widehat{\mbox{diam}}$, and \Cref{clm:cluster_size_diam}
	\[
		\sum_{C'\in\mathcal{C}}\widehat{\mbox{diam}}(C')-\widehat{\mbox{diam}}(C)  \ge\frac{\left|\mathcal{C}\right|g^{i-1}k}{c}-\frac{g^{i}k}{2} \ge\frac{\left|\mathcal{C}\right|g^{i-1}k}{2c}+\frac{dg^{i-1}k}{2c}-\frac{g^{i}k}{2}=\frac{\left|\mathcal{C}\right|g^{i-1}k}{2c}~.\qedhere\]	
\end{proof}
\begin{claim}\label{clm:light_potential_diffrence} 
	Let $C$ be an
	$i$-light cluster. Let $\mathcal{C}$ be the set of the $i-1$ clusters
	contained in $C$, then $\widehat{\mbox{diam}}(C)\le\sum_{C'\in\mathcal{C}}\widehat{\mbox{diam}}(C')+\left|\mathcal{C}\right|-1$.
\end{claim}
\begin{proof}
	This is straightforward as the cluster $C$ was created from $\mathcal{C}$
	using only MST unit weight edges.
\end{proof}
\begin{claim}\label{clm:j_i_cluster_bound} 
	Let $C$ be an $i$ cluster
	and $\mathcal{C}$ be the set of the $j$ clusters contained in $C$
	for some $j<i$. Consider the graph $G\left[C\right]$ where we contract
	all the $j$-clusters and keep only the edges used to create clusters. Then
	$w\left(MST(G\left[C\right])\right)=O\left(\sum_{C'\in\mathcal{C}}\widehat{\mbox{diam}}(C')\right)$.
\end{claim}
\begin{proof}
	Denote by $\mathcal{C}_{r}$ the set of $r$-level clusters contained
	in $\mathcal{C}$. Let $E'_{r}$ be the set of edges used to create
	the clusters $\mathcal{C}_{r+1}$ from $\mathcal{C}_{r}$. Note that
	$\left|E'_{r}\right|<\left|\mathcal{C}_{r}\right|$, and that the weight
	of $e\in E'_{r}$ is bounded by $g^{r+2}$. Moreover, $E'=\cup_{r=j+1}^{i}E'_{r}$
	spans $G\left[C\right]$, and thus we can bound $w\left(MST(G\left[C\right])\right)$
	by $w(E')$. It holds that
	\[
	w\left(E'\right)=\sum_{r=j+1}^{i}w\left(E'_{r}\right)<\sum_{r=j+1}^{i}\left|\mathcal{C}_{r}\right|\cdot g^{r+2}=\sum_{r=j+1}^{i}\sum_{C'\in\mathcal{C}_{r}}g^{r+2}\le\sum_{r=j+1}^{i}\sum_{C'\in\mathcal{C}_{r}}\frac{c\cdot g^{2}}{k}\cdot\widehat{\mbox{diam}}(C')~.
	\]
	By  \Cref{clm:heavy_potential_diffrence} and  \Cref{clm:light_potential_diffrence},
	$\sum_{C'\in\mathcal{C}_{r}}\widehat{\mbox{diam}}(C')=O\left(\sum_{C'\in\mathcal{C}_{j}}\widehat{\mbox{diam}}(C')\right)$.
	We conclude
	\[
	w\left(E'\right)\le O\left(\sum_{r=j+1}^{i}\sum_{C'\in\mathcal{C}_{r}}\frac{1}{k}\cdot\widehat{\mbox{diam}}(C')\right)=O\left(\sum_{C'\in\mathcal{C}_{j}}\widehat{\mbox{diam}}(C')\right)~.\qedhere
	\]
\end{proof}
We now ready to prove \Cref{lem:M_r bound}.
\begin{proof}[Proof of \Cref{lem:M_r bound}]
	Recall that we used modified weights $w_r(e)=\max\left\{ kg^{(r-1)\mu}/\epsilon,w(e)\right\} $.
	The contribution of this change to the weight of $M_r$, bounded by $\left(\left|V_{r}\right|-1\right)kg^{\left(r-1\right)\mu}/\epsilon$.
	Thus we can ignore it, and bound $w(M_r)$ (original weight) instead of $w_r(M_r)$ (modified weight).
	
	Denote by $\mathcal{C}_{i}$ the set of $i$-level clusters. Let $\mathcal{H}_{r}$
	be the set of maximal heavy clusters in $\bigcup_{i=r\mu}^{(r+1)\mu-1}\mathcal{C}_{i}$
	(i.e. heavy clusters that does not contained in any other heavy cluster
	up to level $\left(r+1\right)\mu$). Note that $\mathcal{H}_{r}$
	form a partition of $V_{r}$. We will call the sets in $\mathcal{H}_{r}$
	bugs. We will construct a spanning tree $T$ of $S_{r}$. Trivially, $w(T)$ is upper bound on $w(M_r)$.
	$T$ will consist of spanning tree $T_{C}$ for
	every $C\in\mathcal{H}_{r}$, and in addition a set of cross-bug edges
	$T'$.
	
	First consider $C\in\mathcal{H}_{r}$. Let $\mathcal{C}_{C}$ be all
	the $\left(r-1\right)\mu$ clusters contained in $C$. By  \Cref{clm:j_i_cluster_bound},
	there is a spanning tree $T_{C}$ of weight $O\left(\sum_{C'\in\mathcal{C}_{C}}\widehat{\mbox{diam}}(C')\right)$
	that connects between all the clusters in $\mathcal{C}$. Note that all the edges in $T_C$ contained in $\mathcal{E}_r$, and thus in $S_r$.
	
	Next, let $T'$ be a set of edges between bugs of maximal cardinality,
	such that there is no cycles in $T'\cup\bigcup_{C\in\mathcal{H}_{r}}T_{C}$.
	Set $T=T'\cup\bigcup_{C\in\mathcal{H}_{r}}T_{C}$, note that $T$
	is a spanning forest of $S_{r}$. As each $C\in\mathcal{H}_{r}$ is
	already connected, necessarily $\left|T'\right|\le\left|\mathcal{H}_{r}\right|-1$.
	The weight of each edge $e\in T'$, is at most $g^{\left(r+1\right)\mu}=kg^{\mu r}/\eps$,
	while for every $C\in\mathcal{H}_{r}$, $\widehat{\mbox{diam}}(C)\ge\frac{kg^{r\mu}}{c}$.
	Hence $w(T')\le\left|\mathcal{H}_{r}\right|\cdot\frac{k}{\epsilon}\cdot g^{\mu r}\le\frac{c}{\epsilon}\cdot\sum_{C\in\mathcal{H}_{r}}\widehat{\mbox{diam}}(C)$.
	Using \Cref{clm:heavy_potential_diffrence} and \Cref{clm:light_potential_diffrence}
	\begin{align*}
		w(T) & \le w\left(T'\right)+\sum_{C\in\mathcal{H}_{r}}w(T_{C})=\frac{c}{\epsilon}\cdot\sum_{C\in\mathcal{H}_{r}}\widehat{\mbox{diam}}(C)+\sum_{C\in\mathcal{H}_{r}}O\left(\sum_{C'\in\mathcal{C}_{C}}\widehat{\mbox{diam}}(C')\right)\\
		& =O\left(\sum_{C\in\mathcal{H}_{r}}\sum_{C'\in\mathcal{C}_{C}}\widehat{\mbox{diam}}(C')\right)/\epsilon=O\left(\sum_{\varphi\in V_{r}}\widehat{\mbox{diam}}(\varphi)/\epsilon\right)~.\qedhere
	\end{align*}
\end{proof}
Define a potential function $D_{i}=\sum_{\varphi\in\mathcal{C}_{i}}\widehat{\mbox{diam}}(\varphi)+\left|\mathcal{C}_{i}\right|$.
According to \Cref{clm:light_potential_diffrence},
and \Cref{clm:heavy_potential_diffrence}, $D_i$
is not-increasing. 
\begin{claim}
	For every $r\ge 1$, $D_{\left(r-1\right)\mu}-D_{\left(r+1\right)\mu}=\Omega\left(\left|V_{r}\right|\cdot kg^{\left(r-1\right)\mu}\right)$.
\end{claim}
\begin{proof}
	Consider some $i$-level heavy cluster $C$. Let $\mathcal{C}$ be
	all the $i-1$ clusters contained in $C$. Let $\mathcal{D}$ be the
	potential function on the graph induced by $C$. Then by  \Cref{clm:heavy_potential_diffrence},
	\begin{eqnarray*}
		\mathcal{D}_{i} & = & \widehat{\mbox{diam}}(C)+1\le\sum_{C'\in\mathcal{C}}\widehat{\mbox{diam}}(C')-\frac{\left|\mathcal{C}\right|g^{i-1}k}{2c}+1 \\
		&\le& \frac{1}{2}\sum_{C'\in\mathcal{C}}\widehat{\mbox{diam}}(C')+1 \le
		\frac{1}{2}\left(\sum_{C'\in\mathcal{C}}\widehat{\mbox{diam}}(C')+\left|\mathcal{C}\right|\right)=
		\frac{1}{2}\cdot\mathcal{D}_{i-1}~.		
	\end{eqnarray*}
	Let $\mathcal{V}_{r}$ be all the $\left(r-1\right)\mu$ clusters
	contained in $C$. For $i>\left(r-1\right)\mu$ 
	\begin{eqnarray*}
		\mathcal{D}_{i} & \le & \frac{1}{2}\cdot\mathcal{D}_{i-1}\le\frac{1}{2}\cdot\mathcal{D}_{\left(r-1\right)\mu}=\mathcal{D}_{\left(r-1\right)\mu}-\frac{1}{2}\cdot\left(\sum_{\varphi\in\mathcal{V}_{r}}\widehat{\mbox{diam}}(\varphi)+\left|\mathcal{V}_{r}\right|\right)\\
		& \le & \mathcal{D}_{\left(r-1\right)\mu}-\frac{1}{2}\cdot\left(\sum_{\varphi\in\mathcal{V}_{r}}\left(\frac{g^{\left(r-1\right)\mu}k}{c}+1\right)\right)=\mathcal{D}_{\left(r-1\right)\mu}-\Omega\left(\left|\mathcal{V}_{r}\right|g^{\left(r-1\right)\mu}k\right)~.
	\end{eqnarray*}
	By applying this on all the maximal heavy clusters and get the claim.
\end{proof}
Now we ready to prove  \Cref{lem:total_weight_spanners}.
\begin{proof}[Proof  of \Cref{lem:total_weight_spanners}]
	Fix some $r$. 
	Note that the minimal weight of an edge in $S_{r}$
	is $kg^{(r-1)\mu}/\epsilon$, while by \Cref{lem:M_r bound}, $w_r(M_r)\le O\left(|V_{r}|\cdot kg^{(r-1)\mu}/\epsilon\right)$. 
	Using  \Cref{lem:fast_spanner_aspect},

	\[
	w(H_{r})\le w_{r}(H_{r})\le O_{\epsilon}\left(|V_{r}|^{\frac{1}{k}}\cdot\log\left(\frac{k}{\epsilon}\right)\right)\cdot w_{r}(M_{r})=O_{\epsilon}\left(n^{\frac{1}{k}}\cdot\log k\cdot|V_{r}|\cdot kg^{(r-1)\mu}\right)~.
	\]
	The total weight of the spanners added during the second phase is
	bounded by
	\begin{align*}
		\sum_{r=1}^{\left\lceil k/\mu\right\rceil -1}w(H_{r}) & =O_{\epsilon}\left(n^{\frac{1}{k}}\cdot\log k\cdot\sum_{r=1}^{\left\lceil k/\mu\right\rceil -1}|V_{r}|\cdot kg^{(r-1)\mu}\right)\\
		& =O_{\epsilon}\left(n^{\frac{1}{k}}\cdot\log k\cdot\sum_{r=1}^{\left\lceil k/\mu\right\rceil -1}D_{\left(r-1\right)\mu}-D_{\left(r+1\right)\mu}\right)\\
		& =O_{\epsilon}\left(n^{\frac{1}{k}}\cdot\log k\cdot\left(D_{0}+D_{1}\right)\right)\\
		& =O_{\epsilon}\left(n^{1+\frac{1}{k}}\cdot\log k\right)~,
	\end{align*}
	where the last step follows as $D_{1},D_{0}\le D_{-1}=|V|=n$, as
	all the $-1$-clusters are simply vertices of $G$.
\end{proof}

\section{Halperin Zwick spanner}
In this section we state and analyze the spanner construction of  \cite{HZ96}.
\begin{theorem}[\cite{HZ96}]\label{thm:HZ96}
	For any unweighted graph $G = (V,E)$ and integer $k \ge 1$, a $(2k-1)$-spanner with $O(n^{1+1/k})$ edges can be built in $O(m)$ time.
\end{theorem}
\begin{algorithm}[]
	\caption{$\texttt{HZ-Spanner}(G=(V,E),k)$}
	\label{alg:HZ}
    \DontPrintSemicolon
	$H=(V,\emptyset)$. $V'=V$. $n=|V|$. Throughout the algorithm, $G'$ denotes
    $G[V']$\;
    \While{$V'\ne\emptyset$}{
		Let $v\in V'$ be arbitrary vertex.\label{line:ChooseV}\;
		Let $r\in \mathbb{N}$ be minimal such that \label{line:ChooseR}
        $\left|B_{G'}(v,r)\right|\cdot
        n^{\frac{1}{k}}\ge\left|B_{G'}(v,r+1)\right|$.\;
		Let $T$ be a BFS tree in $B_{G'}(v,r+1)$, rooted at $v$\;
		$H\leftarrow H\cup T$\;
		$V'\leftarrow V'\setminus B_{G'}(v,r)$\;
    }
	\Return $H$\;
\end{algorithm}
\begin{proof}
	We analyze \Cref{alg:HZ}. Note that in line \ref{line:ChooseR}, necessarily $r\le k-1$ as otherwise 	
	\[
	\left|B_{G}(v,r)\right|\ge\left|B_{G}(v,k)\right|>\left|B_{G}(v,k-1)\right|\cdot n^{\frac{1}{k}}>\cdots>\left|B_{G}(v,k-i)\right|\cdot n^{\frac{i}{k}}>\left|B_{G}(v,0)\right|\cdot n^{\frac{k}{k}}=n~.
	\]
    To bound the stretch of $H$, consider an edge $e=(x,y)$. Let
    $v_x,r_x$ (resp. $v_y,r_y$) such that $x$ (resp. $y$) was removed from $V'$
    as part of $B_{G'}(v_x,r_x)$ (resp .$B_{G'}(v_y,r_y)$).
	If $v_x=v_y$, then $d_H(x,y)\le r_x+r_y\le 2(k-1)$. Otherwise, assume
    w.l.o.g that $v_x$ was removed before $v_y$. As $y$ is neighboring vertex
    of $B_{G'}(v_x,r_x)$, necessarily there is a vertex $z\in B_{G'}(v_x,r_x)$,
    such that we added $(z,y)$ to $H$. By triangle inequality
	\[d_H(x,y)\le d_H(x,z)+d_H(z,y)\le 2\cdot r_x+1\le 2k-1~.\]
	
	To bound the sparsity, note that when deleting $B_{G'}(v,r)$, we add  $|B_{G'}(v,r+1)|-1\le |B_{G'}(v,r)|\cdot n^{\frac{1}{k}}$ edges. Thus by charging $n^{\frac{1}{k}}$ on each deleted vertex, we can bound the total number of edges by $O(n^{1+\frac{1}{k}})$.
				
	The runtime is straightforward, as we consider each edge at most twice.
\end{proof}

\subsection{Modified \cite{HZ96} Spanner}\label{app:ModHZ}
\Cref{alg:HZ} picks an arbitrary vertex in line \ref{line:ChooseV} and grow a ball around it. Our spanner in \Cref{thm:SparseNoLight} uses \Cref{alg:HZ} as sub-procedure. However we will need additional property from the spanner. Specifically, we will prefer to pick a vertex with at least $n^{\frac{1}{k}}-1$ active neighbors.
The modified algorithm presented in \Cref{alg:Mod-HZ}. We denote by $\deg_{G'}(v)$, the degree of $v$ in $G'$.
\begin{algorithm}[]
	\caption{$\texttt{Modified-HZ-Spanner}(G=(V,E),k)$}
    \label{alg:Mod-HZ}
    \DontPrintSemicolon
	$H=(V,\emptyset)$. $V'=V$. $n=|V|$. Throughout the algorithm, $G'$ denotes
    $G[V']$\;
	\While{$V'\ne\emptyset$}{
		If possible, pick $v\in V'$ such that $\deg_{G'}(v)\ge
        n^{\frac{1}{k}}-1$. If not, pick  arbitrary vertex $v\in V'$.
        \label{line:ModChooseV}\;
		Let $r\in \mathbb{N}$ be minimal such that
        $\left|B_{G'}(v,r)\right|\cdot
        n^{\frac{1}{k}}\ge\left|B_{G'}(v,r+1)\right|$.\label{line:ModChooseR}\;
		Let $T$ be a BFS tree in $B_{G'}(v,r+1)$, rooted at $v$\;
		$H\leftarrow H\cup T$\;
		$V'\leftarrow V'\setminus B_{G'}(v,r)$\;
	}
	\Return $H$\;
\end{algorithm}
%
\ModHZ*
\begin{proof}
	The stretch and sparsity follows from \Cref{thm:HZ96} as we only specify (the prior arbitrary) order of choosing vertices in line \ref{line:ChooseV}.
	Property 2 follows as the radius chosen in line \ref{line:ModChooseR} bounded by $k-1$.
	Properties 1,3,4 are straightforward from  line \ref{line:ModChooseV} of \Cref{alg:Mod-HZ}. Thus we only need to bound the running time.
	
	It will be enough to provide an efficient way to pick vertices in line \ref{line:ModChooseV}.  We will maintain $\deg(v)$ for every vertex $v$, and a set $A$ of all the vertices with degree at least $n^{\frac{1}{k}}$.
	The degrees are computed in the beginning of the algorithm, and all the relevant vertices inserted to $A$. Then, in iteration $i$, after deleting $S_i$, we go over each deleted vertex, decrease the degree of each neighboring vertex, and update $A$ accordingly.
	Using $A$, the decision in line \ref{line:ModChooseV} can be executed in constant time.
	The maintenance of $A$ and the degrees can be done in $O(m)$ time, as we refer to each edge at most constant number of times.	
\end{proof}

\bibliographystyle{alpha}
\bibliography{spanners}

\newcommand{\etalchar}[1]{$^{#1}$}
\begin{thebibliography}{MPVX15}

\bibitem[ADD{\etalchar{+}}93]{Althofer1993}
Ingo Alth{\"o}fer, Gautam Das, David Dobkin, Deborah Joseph, and Jos{\'e}
  Soares.
\newblock On sparse spanners of weighted graphs.
\newblock {\em Discrete {\&} Computational Geometry}, 9(1):81--100, 1993.

\bibitem[Awe85]{Awerbuch:1985:CNS:4221.4227}
Baruch Awerbuch.
\newblock Complexity of network synchronization.
\newblock {\em J. ACM}, 32(4):804--823, October 1985.

\bibitem[BFN19]{BartalFN16}
Yair Bartal, Arnold Filtser, and Ofer Neiman.
\newblock On notions of distortion and an almost minimum spanning tree with
  constant average distortion.
\newblock {\em Journal of Computer and System Sciences}, 2019.

\bibitem[BHN16]{DBLP:conf/stoc/BhattacharyaHN16}
Sayan Bhattacharya, Monika Henzinger, and Danupon Nanongkai.
\newblock New deterministic approximation algorithms for fully dynamic
  matching.
\newblock In {\em Proc. 48th ACM Symposium on Theory of Computing (STOC)},
  pages 398--411, 2016.

\bibitem[BS07]{BaswanaS07}
Surender Baswana and Sandeep Sen.
\newblock A simple and linear time randomized algorithm for computing sparse
  spanners in weighted graphs.
\newblock {\em Random Structures \& Algorithms}, 30(4):532--563, 2007.
\newblock See also ICALP'03.

\bibitem[CDNS92]{ChandraDNS92}
Barun Chandra, Gautam Das, Giri Narasimhan, and Jos{\'e} Soares.
\newblock New sparseness results on graph spanners.
\newblock In {\em Proc. 8th ACM Symposium on Computational Geometry (SoCG)},
  pages 192--201, 1992.

\bibitem[Che13]{Chechik13a}
Shiri Chechik.
\newblock Compact routing schemes with improved stretch.
\newblock In {\em Proc.ACM Symposium on Principles of Distributed Computing
  (PODC)}, pages 33--41, 2013.

\bibitem[Che14]{Chechik14}
Shiri Chechik.
\newblock Approximate distance oracles with constant query time.
\newblock In {\em Proc. 46th ACM Symposium on Theory of Computing (STOC)},
  pages 654--663, 2014.

\bibitem[Che15]{Chechik15}
Shiri Chechik.
\newblock Approximate distance oracles with improved bounds.
\newblock In {\em Proc. 47th ACM Symposium on Theory of Computing (STOC)},
  pages 1--10, 2015.

\bibitem[CW18]{ChechikW16}
Shiri Chechik and Christian Wulff{-}Nilsen.
\newblock Near-optimal light spanners.
\newblock {\em {ACM} Trans. Algorithms}, 14(3):33:1--33:15, 2018.

\bibitem[DN97]{DasN97}
Gautam Das and Giri Narasimhan.
\newblock A fast algorithm for constructing sparse euclidean spanners.
\newblock {\em Int. J. Comput. Geometry Appl.}, 7(4):297--315, 1997.

\bibitem[EN19]{ElkinN17c}
Michael Elkin and Ofer Neiman.
\newblock Efficient algorithms for constructing very sparse spanners and
  emulators.
\newblock {\em {ACM} Trans. Algorithms}, 15(1):4:1--4:29, 2019.

\bibitem[ENS15]{ElkinNS14}
Michael Elkin, Ofer Neiman, and Shay Solomon.
\newblock Light spanners.
\newblock {\em {SIAM} J. Discrete Math.}, 29(3):1312--1321, 2015.

\bibitem[Erd64]{Erdos64}
Paul Erdős.
\newblock Extremal problems in graph theory.
\newblock In {\em Theory of Graphs and its Applications}, pages 29--36, 1964.

\bibitem[ES81]{EvenS81}
Shimon Even and Yossi Shiloach.
\newblock An on-line edge-deletion problem.
\newblock {\em J. {ACM}}, 28(1):1--4, 1981.

\bibitem[ES16]{ElkinS16}
Michael Elkin and Shay Solomon.
\newblock Fast constructions of lightweight spanners for general graphs.
\newblock {\em ACM Transactions on Algorithms}, 12(3):29:1--29:21, 2016.
\newblock See also SODA'13.

\bibitem[FPZW04]{FarleyPZW04}
Arthur~M. Farley, Andrzej Proskurowski, Daniel Zappala, and Kurt Windisch.
\newblock Spanners and message distribution in networks.
\newblock {\em Discrete Applied Mathematics}, 137(2):159 -- 171, 2004.

\bibitem[FS20]{FiltserS16}
Arnold Filtser and Shay Solomon.
\newblock The greedy spanner is existentially optimal.
\newblock volume~49, pages 429--447, 2020.

\bibitem[HKN16]{HenzingerKN16}
Monika Henzinger, Sebastian Krinninger, and Danupon Nanongkai.
\newblock Dynamic approximate all-pairs shortest paths: Breaking the o(mn)
  barrier and derandomization.
\newblock {\em SIAM Journal on Computing}, 45(3):947--1006, 2016.
\newblock See also FOCS'13.

\bibitem[HZ96]{HZ96}
Shay Halperin and Uri Zwick.
\newblock Linear time deterministic algorithm for computing spanners for
  unweighted graphs, 1996.

\bibitem[Kru56]{Kruskal56}
Joseph~B. Kruskal.
\newblock {On the Shortest Spanning Subtree of a Graph and the Traveling
  Salesman Problem}.
\newblock {\em Proceedings of the American Mathematical Society}, 7(1):48--50,
  February 1956.

\bibitem[KX16]{KoutisX16}
Ioannis Koutis and Shen~Chen Xu.
\newblock Simple parallel and distributed algorithms for spectral graph
  sparsification.
\newblock {\em {TOPC}}, 3(2):14:1--14:14, 2016.

\bibitem[MPVX15]{MillerPVX15}
Gary~L. Miller, Richard Peng, Adrian Vladu, and Shen~Chen Xu.
\newblock Improved parallel algorithms for spanners and hopsets.
\newblock In {\em Proc. 27th ACM Symposium on Parallel Algorithms and
  Architectures (SPAA)}, pages 192--201, 2015.

\bibitem[PS89]{PelegS89}
David Peleg and Alejandro~A. Schäffer.
\newblock Graph spanners.
\newblock {\em Journal of Graph Theory}, 13(1):99--116, 1989.

\bibitem[PU87]{PelegU87}
David Peleg and Jeffrey~D. Ullman.
\newblock An optimal synchronizer for the hypercube.
\newblock In {\em Proc. 6thACM Symposium on Principles of Distributed Computing
  (PODC)}, pages 77--85, 1987.

\bibitem[PU88]{PelegU88}
David Peleg and Eli Upfal.
\newblock A tradeoff between space and efficiency for routing tables.
\newblock In {\em Proc 20th ACM Symposium on Theory of Computing (STOC)}, pages
  43--52, 1988.

\bibitem[RTZ05]{RodittyTZ05}
Liam Roditty, Mikkel Thorup, and Uri Zwick.
\newblock Deterministic constructions of approximate distance oracles and
  spanners.
\newblock In {\em Proc. 32nd International Colloquium on Automata, Languages
  and Programming (ICALP)}, pages 261--272, 2005.

\bibitem[RZ11]{RodittyZ11}
Liam Roditty and Uri Zwick.
\newblock On dynamic shortest paths problems.
\newblock {\em Algorithmica}, 61(2):389--401, 2011.

\bibitem[RZ12]{RodittyZ12}
Liam Roditty and Uri Zwick.
\newblock Dynamic approximate all-pairs shortest paths in undirected graphs.
\newblock {\em SIAM Journal on Computing}, 41(3):670--683, 2012.
\newblock See also FOCS'04.

\bibitem[Tar75]{Tarjan75}
Robert~Endre Tarjan.
\newblock Efficiency of a good but not linear set union algorithm.
\newblock {\em J. ACM}, 22(2):215–225, apr 1975.

\bibitem[TZ01]{ThorupZ01}
Mikkel Thorup and Uri Zwick.
\newblock Compact routing schemes.
\newblock In {\em Proc. 13th ACM Symposium on Parallel Algorithms and
  Architectures (SPAA)}, pages 1--10, 2001.

\bibitem[TZ05]{ThorupZ05}
Mikkel Thorup and Uri Zwick.
\newblock Approximate distance oracles.
\newblock {\em Journal of the ACM}, 52(1):1--24, January 2005.
\newblock See also STOC'01.

\bibitem[Wul12]{CWN12}
Christian Wulff{-}Nilsen.
\newblock Approximate distance oracles with improved preprocessing time.
\newblock In {\em Proc. 23rd ACM/SIAM Symposium on Discrete Algorithms (SODA)},
  pages 202--208, 2012.

\bibitem[Wul13]{DBLP:conf/soda/Wulff-Nilsen13}
Christian Wulff{-}Nilsen.
\newblock Approximate distance oracles with improved query time.
\newblock In {\em Proceedings of the Twenty-Fourth Annual {ACM-SIAM} Symposium
  on Discrete Algorithms, {SODA} 2013, New Orleans, Louisiana, USA, January
  6-8, 2013}, pages 539--549, 2013.

\bibitem[Wul17]{Wul17}
Christian Wulff{-}Nilsen.
\newblock Fully-dynamic minimum spanning forest with improved worst-case update
  time.
\newblock In {\em Proceedings of the 49th Annual {ACM} {SIGACT} Symposium on
  Theory of Computing, {STOC} 2017, Montreal, QC, Canada, June 19-23, 2017},
  pages 1130--1143, 2017.

\end{thebibliography}
\end{document}